\documentclass[preprint,12pt]{elsarticle}

\usepackage[T1]{fontenc}
\usepackage[utf8]{inputenc}
\usepackage{lmodern}
\usepackage{graphicx}
\usepackage[hyphens]{url}
\usepackage[caption=false,font=footnotesize,hangindent=10pt]{subfig}
\usepackage{amsmath,amsfonts,amsthm,amssymb}
\usepackage[linesnumbered,vlined,boxed,ruled]{algorithm2e}
\usepackage[colorlinks]{hyperref}
\usepackage{color,comment,xparse}
\usepackage{tabu}
\usepackage{tikz}
\usetikzlibrary{shapes.geometric,shapes.misc,positioning,backgrounds,fit}
\usepackage{array}
\newcolumntype{C}{@{}>{\centering\arraybackslash}p{.333\linewidth}@{}}

\graphicspath{{./}{figure/}}

\def\Snospace~{\S{}}

\SetCommentSty{myAlgComment}
\SetKwProg{Fn}{Function}{:}{}
\SetAlFnt{\small}
\SetAlCapFnt{\small}
\SetAlCapNameFnt{\small}

\newcommand{\header}[1]{\medskip\noindent\textbf{#1}}

\newcommand{\remarks}{\medskip\noindent\emph{Remarks}}

\newcommand{\E}{\mathbb{E}}
\newcommand{\argmax}{\mathop{\arg\max}}

\newcommand{\NbrO}[1]{\Gamma_{\!\text{out}}(#1)}
\newcommand{\NbrI}[1]{\Gamma_{\!\text{in}}(#1)}
\newcommand{\Dp}[2]{\Delta p_{#1}^{#2}}
\newcommand{\Dh}[2]{\Delta h_{#1}^{#2}}
\newcommand{\Dhp}[2]{\Delta\hat{p}_{#1}^{#2}}
\newcommand{\Dhh}[2]{\Delta\hat{h}_{#1}^{#2}}
\newcommand{\CW}{\mathcal{W}}
\newcommand{\FAP}{F_{\text{AP}}}
\newcommand{\FHT}{F_{\text{HT}}}
\newcommand{\hFAP}{\hat{F}_{\text{AP}}}
\newcommand{\hFHT}{\hat{F}_{\text{HT}}}
\newcommand{\dAP}{\delta_{\text{AP}}}
\newcommand{\dHT}{\delta_{\text{HT}}}
\newcommand{\hdAP}{\hat{\delta}_{\text{AP}}}
\newcommand{\hdHT}{\hat{\delta}_{\text{HT}}}

\newtheorem{definition}{Definition}
\newtheorem{problem}{Problem}
\newtheorem{theorem}{Theorem}

\journal{Information Sciences}

\begin{document}

\begin{frontmatter}

\title{Optimizing Node Discovery on Networks: Problem Definitions, Fast
    Algorithms, and Observations}

\author[cuhk]{Junzhou Zhao}
\author[xjtu]{Pinghui Wang}
\author[cuhk]{John C.S. Lui}
\address[cuhk]{The Chinese University of Hong Kong}
\address[xjtu]{Xi'an Jiaotong University}

\begin{abstract}
Many people dream to become famous, YouTube video makers also wish their videos to
have a large audience, and product retailers always hope to expose their products
to customers as many as possible.
Do these seemingly different phenomena share a common structure?
We find that fame, popularity, or exposure, could be modeled as a node's {\em
  discoverability} on some properly defined network, and all of the previously
mentioned phenomena can be commonly stated as {\em a target node wants to be
  discovered easily by the other nodes in the network}.
In this work, we explicitly define a node's {\em discoverability} in a network,
and formulate a general {\em node discoverability optimization} problem, where the
goal is to create a budgeted set of incoming edges to the target node so as to
optimize the target node's discoverability in the network.
Although the optimization problem is proven to be NP-hard, we find that the
defined discoverability measures have good properties that enable us to use a
greedy algorithm to find provably near-optimal solutions.
The computational complexity of a greedy algorithm is dominated by the time cost
of an {\em oracle call}, i.e., calculating the marginal gain of a node.
To scale up the oracle call over large networks, we propose an {\em
  estimation-and-refinement} approach, that provides a good trade-off between
estimation accuracy and computational efficiency.
Experiments conducted on real-world networks demonstrate that our method is
thousands of times faster than an exact method using dynamic programming, thereby
allowing us to solve the node discoverability optimization problem on large
networks.

\end{abstract}

\begin{keyword}

submodular/supermodular set function\sep
greedy algorithm\sep
MCMC simulation\sep
random walk

\end{keyword}

\end{frontmatter}

\section{Introduction}
\label{sec:introduction}

We consider a general problem of adding a budgeted set of new edges to a graph,
that each new edge connects an existing node in the graph to a newly introduced
{\em target node}, so that the target node can be {\em discovered} easily by
existing nodes in the new graph.
We refer to this problem as the target node {\em discoverability optimization
  problem} on networks.

\header{Motivation.}
The problem of optimizing node discoverability on networks appears in a wide range
of applications.
For example, a YouTube video maker may wish her videos to have a large audience
and click traffic.
In YouTube, each video is related to a set of recommended videos, and the majority
of YouTube videos are discovered and watched by viewers following related
videos~\cite{Zhou2010}.
Hence, if a video maker could make her video related to a set of properly chosen
videos (i.e., make her video appear in each chosen video's related video list),
her video may have more chance to be discovered and watched.
This task is known as the {\em related video optimization
  problem}~\cite{related_video}, and in practice, a video maker can make her video
related to some other videos by writing proper descriptions, choosing the right
title, adding proper meta-data and keywords~\cite{grow_audience}.
In this application, we can build a {\em video network}, where a node represents a
video, and a directed edge represents one video relating to another.
Then, making a target video related to a set of existing videos is equivalent to
adding a set of edges from existing nodes to the target node in the video network.
Therefore, the related video optimization problem is actually a target node
discoverability optimization problem.

As another application, let us consider the advertising service provided by many
retail websites, such as Amazon~\cite{amazon} and Taobao~\cite{taobao}.
A major concern of product sellers is that whether customers could easily discover
their products on these retail websites~\cite{Antikacioglu2015}.
One important factor that affects the discoverability of an item on a retail
website is {\em what other items' detail pages display this item}.
For example, on Amazon, a seller's product could be displayed on a related
product's detail page in the list ``sponsored products related to this item''.
If an item was displayed on several popular or best selling products' detail
pages, the item would be easily discovered by many customers, and have good sells.
A product seller has some control to decide how strong her item is related to some
other items, e.g., a book writer on Amazon can choose proper keywords or features
to describe her book, set her interests, other similar books, and cost-per-click
bid~\cite{amazon-ads}.
In this application, we can build an {\em item network}, where a node represents
an item, and a directed edge represents one item relating to another.
Therefore, optimizing the discoverability of an item by relating to other proper
items on a retail website can be formulated as the target node discoverability
optimization problem.

In the third application, we consider the message forwarding processes on a
follower network, such as tweet re-tweeting on Twitter.
In a follower network, a person (referred to as a \emph{follower}) could follow
another person (referred to as a \emph{followee}), and then the follower could
receive messages posted or re-posted by her followees.
In this way, messages diffuse on a follower network through forwarding by users
(with direction from a followee to her followers).
Hence, what followees a person chooses to follow determines what messages she
could receive and how soon the messages could arrive at the person.
The problem of choosing an optimal set of followees for a new user to maximize
information coverage and minimize time delay is known as the {\em whom-to-follow
  problem}~\cite{Zhao2014a}.
On the other hand, if we consider this problem from the perspective of messages,
we actually want messages to reach the user efficiently (through re-posting) by
adding few new edges in the follower network.
Therefore, the whom-to-follow problem can also be formulated as the target node
discoverability optimization problem.

\header{Related Work.}
Despite the pervasive applications of the node discoverability optimization
problem in practice, it is surprising that there is even no explicit definition of
node discoverability in a network in the literature.
Suppose we could leverage the concept of node centrality~\cite{Freeman1978}, say,
the closeness centrality~\cite{Cohen2014}, to quantify a node's discoverability in
a network, i.e., a node is closer to other nodes in the network, it is more
discoverable.
However, how to optimize a node's closeness centrality by adding new edges in the
network could be extremely difficult, especially for large networks.
Antikacioglu et al.~\cite{Antikacioglu2015} study the web discovery optimization
problem in an e-commerce website, and the goal is to add links from a small set of
popular pages to new pages to make as many new pages discoverable as possible
(under some constraints).
Here, a page is discoverable if it has at least $a\geq 1$ links from popular pages
in the site.
However, such a definition of discoverability is too strict, as it actually
assumes that a user is only allowed to browse a website for at most one hop to
discover a page.
In practice, a user may browse the site for several hops, and finally discover a
page, even though the page may have no link from popular pages at all.
Rosenfeld and Globerson~\cite{Rosenfeld2016} study the optimal tagging problem in
a network consisting of tags and items, and their goal is to pick $k$ tags for
some new item in order to maximize the new item's incoming traffic.
This problem is formulated as maximizing the absorbing probability of an absorbing
state (representing the new item) in a Markov chain by adding $k$ new transitions
to the absorbing state, and the nice theoretical results of absorbing Markov
chains facilitate expressing absorbing probability in terms of the fundamental
matrix of the chain~\cite{Trivedi2016}.
We notice that, measuring a node's discoverability by absorbing probability
relieves the restriction of~\cite{Antikacioglu2015}, but it implicitly assumes
that a user has infinite amount of time or patience to browse the network to
discover an item, which is, however, not usually the case in
practice~\cite{Simon1971,Scaria2014}.

\header{Present work.}
In this work, we study the general problem of node discoverability optimization on
networks.
We will no longer mention the particular application in the following discussion,
and we consider the problem in a general weighted directed graph, which could
represent the video network, item network, or follower network.
We first propose several definitions of node discoverability in a network, which
measure node discoverability from different perspectives, and then provide a
unified framework for optimizing node discoverability by adding new edges in the
network.
The main result of our work is an efficient graph computation system that enables
us to address the node discoverability optimization problem over million scale
large graphs on a common PC.

\header{Measuring node discoverability by finite length random walks.}
To quantify a node's discoverability in a network, we propose two measures based
on finite length random walks~\cite{Lovasz1993}.
More specifically, we measure discoverability of the target node by analyzing a
collection of random walks that start from existing nodes in the network, and we
consider (1) the probability that a random walk could finally hit the target node,
and (2) the average number of steps that a random walk finally reaches the target
node.
Intuitively, if a random walk starting from a node $i$ could reach the target node
with high probability, and use few steps on average, then we say that the target
node is discoverable by node $i$.
Using random walks to measure discoverability is general, because many real-world
processes are indeed suitable to be modeled as random walks, e.g., user watching
YouTube videos by following related videos~\cite{Kumar2015}, people's navigation
and searching behaviors on the Web~\cite{Scaria2014} and peer-to-peer
networks~\cite{Gkantsidis2006}, and some diffusion processes such as letter
forwarding in Milgram's small-world experiment~\cite{Travers1969}.

\header{Efficient optimization via estimating-and-refining.}
The optimization problem asks us to add a budgeted set of new edges to the graph
that each new edge connects an existing node to the target node, to optimize the
target node's discoverability in the new graph.
The optimization problem is NP-hard, which inhibits us to find the optimal
solutions for a large network.
We find that the two objectives are submodular and supermodular, respectively, and
hence we can obtain approximate solutions by the greedy algorithm, which has
polynomial time complexity and constant approximation
factor~\cite{Nemhauser1978,Nemhauser1978a,Minoux1978}.
The main challenge here is how to scale up the greedy algorithm over large
networks containing millions of nodes/edges.
The computational complexity of the greedy algorithm is dominated by the time cost
of an {\em oracle call}, i.e., calculating the marginal gain of a node.
To speed up the oracle call, we propose an {\em estimation-and-refinement}
approach, that has a good trade-off between estimation accuracy and computational
efficiency.
Our final designed system is built on top of the contemporary efficient MCMC
simulation systems~\cite{Fogaras2005,Kyrol2013,Liu2016c}, and is empirically
demonstrated to be thousands of times faster than a naive approach based on
dynamic programming.

\header{Contributions.}
We make following contributions in this work:
\begin{itemize}
\item We formally define the node discoverability on networks, and propose a
  unified framework for the problem of optimizing node discoverability on
  networks.
  The problem is general and appears in a wide range of practical applications.
\item We prove the objectives satisfying submodular and supermodular properties,
  respectively.
  We propose an efficient estimation-and-refinement approach to implement the
  oracle call when using the greedy algorithm to find quality guaranteed
  solutions.
  Our proposed approach has a good trade-off between estimation accuracy and
  computational efficiency.
\item We conduct extensive experiments on real networks to evaluate our proposed
  method.
  The experimental results demonstrate that the estimation-and-refinement approach
  is thousands of times faster than a naive method based on dynamic programming,
  and hence the approach enables us to handle large networks with millions of
  nodes/edges.
\end{itemize}

\header{Outline.}
The reminder of this paper will proceed as follows.
In~\autoref{sec:problem}, we formally define the node discoverability, and
formulate two versions of node discoverability optimization problem.
In~\autoref{sec:method}, we elaborate three methods to address the optimization
problem.
In~\autoref{sec:experiment}, we conduct experiments to validate the proposed
methods.
In~\autoref{sec:applications}, we give applications of the node discoverability
optimization problem.
\autoref{sec:relatedwork} provides more related work in the literature, and
finally~\autoref{sec:conclusion} concludes.
Proofs of our main results are provided in Appendix.

\section{Preliminaries and Problem Formulation}
\label{sec:problem}

In this section, we propose two definitions of node discoverability on a network,
and then formulate two versions of the node discoverability optimization problem.
Finally, we discuss several important properties of the optimization problem.

\begin{table}[t]
  \small
  \centering
  \caption{Frequently used notations\label{tab:notations}}
  \begin{tabular}{@{ }c@{ }|@{ }l}
    \hline
    \hline
    {\bf symbol} & {\bf description} \\
    \hline
    $G = (V, E)$ & digraph with node set $V\!=\!\{0,\ldots,n\!-\!1\}$ and
                   edge set $E$ \\
    $n\notin V$ & the target node, or the size of $V$ \\
    $S\subseteq V$ & connection sources \\
    $E_S\!\triangleq\!\{(i,n)\colon i\!\in\! S\}$ & newly added edges \\
    $G' = (V', E')$ & graph after adding node $n$ and edges in $E_S$ \\
    $\NbrO{i}, \NbrI{i}\subseteq V'$ & out- and in-neighbors of node $i$ in
                                       graph $G'$ \\
    $w_{ij}, p_{ij}$ & weight and transition probability on edge $(i,j)$ \\
    $p_i^t, h_i^t$  & truncated absorbing probability/hitting time
                      from $i$ to $n$ \\
    $\Dp{i}{t}(s), \Dh{i}{t}(s)$ & change of truncated absorbing
                                   probability/hitting time \\
    $\FAP(S),\FHT(S)$     & D-AP and D-HT \\
    $\dAP(s;S),\dHT(s;S)$ & marginal gains \\
    $T$ & maximum length of a walk \\
    $R$ & number of walks from each node \\
    $D$ & refinement depth \\
    $b_{ir},b_w\in\{0,1\}$ & the corresponding walk hit or miss target $n$\\
    $t_{ir},t_w\in[0,T]$ & the number of steps walked by the walk\\
    \hline
  \end{tabular}
\end{table}

\subsection{Node Discoverability Definitions}

Let $G=(V,E)$ denote a general weighted and directed graph, where $V=\{0, \ldots,
n-1\}$ is a set of $n$ nodes, and $E\subseteq V\times V$ is a set of edges.
Each edge $(i,j)\in E$ is associated with a positive weight $w_{ij}$.
For example, in the YouTube video network, $w_{ij}$ could represent the
relationship strength that video $j$ is related to video $i$.
For the convenient of our following discussion, if a node has no out-neighbor,
i.e., a dangling node, we manually add a self-loop edge on this node with weight
one, which is equivalent to turn this node into an absorbing node.

We consider the discoverability of a newly introduced node, denoted by $n$, e.g.,
a newly uploaded video in YouTube, or a new product for sale on Amazon.
Node $n$ can improve its discoverability by creating an additional set of edges
$E_S\triangleq\{(i,n)\colon i\in S\subseteq V\}$, and this forms a new graph
$G'=(V',E')$ where $V'=V\cup\{n\}$ and $E'=E\cup E_S$.
$S\subseteq V$ is referred to as the {\em connection sources}, which we need to
choose from $V$.
For example, in YouTube, creating new edges $E_S$ means relating the new video $n$
to existing videos in $S$ (through writing proper descriptions, choosing the right
title, adding proper meta-data and keywords, etc.~\cite{grow_audience}), and hence
video $n$ could appear in the related video list of each video in connection
sources $S$.

We propose to quantify the discoverability of target node $n$ by random
walks~\cite{Lovasz1993}.
Let $\NbrO{i}, \NbrI{i}\subseteq V'$ denote the sets of out- and in-neighbors of
node $i$ in graph $G'$, respectively.
A random walk starts from a node in $V$, and at each step, it randomly picks an
out-neighbor $j\in\NbrO{i}$ of the currently resident node $i$ to visit, with
probability $p_{ij}\triangleq w_{ij}/\sum_{k\in\NbrO{i}}w_{ik}$.
The random walk stops once it hits the target node $n$ for the first time, or has
walked a maximum number of $T$ steps.
For such a {\em finite length random walk}, we are interested in the following two
measures.

\begin{definition}[{\bf Truncated Absorbing Probability}] \label{def:ap}
  The truncated absorbing probability of a node $i\in V$ is the probability that a
  finite length random walk starting from node $i$ will end up at the target node
  $n$ by walking at most $T$ steps, i.e., $p_i^T\triangleq P(X_t=n, t\leq
  T|X_0=i)$.
\end{definition}

The truncated absorbing probability satisfies the following useful equation.
For $t=0,\ldots,T$,
  \begin{equation} \label{eq:ap}
  p_i^t =
  \begin{cases}
    1                                & \text{if } i=n, \\
    0                                & \text{if } t=0 \text{ and } i\neq n, \\
    \sum_{k\in\NbrO{i}} p_{ik}p_k^{t-1} & \text{otherwise}.
  \end{cases}
\end{equation}
A random walk starting from node $i$ and hitting target node $n$ by walking at
most $T$ steps can be thought of as a Bernoulli trial with success probability
$p_i^T$.
Intuitively, if many random walks from different nodes in $V$ could finally hit
target node $n$ within $T$ steps, i.e., many Bernoulli trials succeed, then the
target node $n$ is discoverable, and it should have ``good'' discoverability in
graph $G'$.
This immediately leads to the following definition of node discoverability by
truncated absorbing probabilities.

\begin{definition}[{\bf Discoverability based on Truncated Absorbing
    Probabilities, abrv. D-AP}]\label{def:D-AP}
  Assume that a random walk starts from a node in $V$ uniformly at random, then
  the discoverability of target node $n$ is defined as the {\bf\em expected
    truncated absorbing probability} that a random walk starting from a node in
  $V$ could hit $n$ within $T$ steps, i.e., $\sum_{i\in V}p_i^T/n$.
\end{definition}

The value of D-AP is in the range $[0,1]$, and has a probabilistic explanation.
Although D-AP can describe the probability that a random walk starting from a node
in $V$ could hit target node $n$ within $T$ steps, it does not provide any
information about the number of steps that the walker has walked before hitting
$n$.
This inspires us to use a {\em truncated hitting time} to define another version
of node discoverability, and the truncated hitting time is defined as follows.

\begin{definition}[{\bf Truncated Hitting Time}] \label{def:ht}
  The truncated hitting time of a node $i\in V$ is the expected number of steps
  that a finite length random walk starting from node $i$ hits target node $n$ for
  the first time, or terminates at the maximum step $T$, i.e., $h_i^T\triangleq
  \E{\min\{\min\{t\colon X_0=i,X_t=n\}, T\}}$.
\end{definition}

Similar to the truncated absorbing probability, the truncated hitting time also
has a useful recursive definition.
For $t=0,\ldots,T$,
\begin{equation} \label{eq:ht}
  h_i^t =
  \begin{cases}
    0                                    & \text{if } i=n \text{ or } t=0, \\
    1+\sum_{k\in \NbrO{i}} p_{ik}h_k^{t-1}  & \text{otherwise}.
  \end{cases}
\end{equation}
The truncated hitting time was first introduced to measure the pairwise node
similarity in a graph~\cite{Sarkar2007b,Sarkar2008}.
Here, we leverage truncated hitting time to measure the discoverability of a node
in a network.
Intuitively, if random walks starting from nodes in $V$ could hit target node $n$
with small truncated hitting times on average, then node $n$ can be easily
discovered in the graph.
This implies the following definition.

\begin{definition}[{\bf Discoverability based on Truncated Hitting Times, abrv.
    D-HT}]\label{def:D-HT}
  Assume that a random walk starts from a node in $V$ uniformly at random, then
  the discoverability of target node $n$ is the {\bf\em expected number of steps}
  that a random walk starting from a node in $V$ hits $n$ for the first time, by
  walking at most $T$ steps, i.e., $\sum_{i\in V}h_i^T/n$.
\end{definition}

The value of D-HT is in the range $[0,T]$, and has a physical meaning as the
expected number of steps that a walker has walked before hitting node $n$ for the
first time.

\remarks.
(1) We use finite length random walks rather than infinite length random walks to
characterize node discoverability because that people's searching and navigation
behaviors on the Internet usually consist of finite length click paths due to time
or attention limitations~\cite{Scaria2014}.
Such a treatment can thus be viewed as a trade-off between two extremes, i.e., web
discovery optimization~\cite{Antikacioglu2015} using $T=1$, and optimal
tagging~\cite{Rosenfeld2016} using $T=\infty$.

(2) It is also straightforward to extend the two basic node discoverability
definitions to more complex definitions that encompass both truncated absorbing
probability and truncated hitting time.
For example, we can construct the following extension of node discoverability
$\sum_i(\alpha p_i^T + \beta h_i^T)/n$, where constants $\alpha\geq 0$ and
$\beta\leq 0$ represent the importance of the two parts, respectively.

\subsection{Node Discoverability Optimization}
\label{ss:objectives}

Equipped with the clear definitions of node discoverability, we are now ready to
formulate the node discoverability optimization problem.
To be more specific, the optimization problem seeks to introduce a set of new
edges $E_S=\{(s,n)\colon s\in S\subseteq V\}$ to graph $G$, and form a new graph
$G'=(V',E')$ with $V'=V\cup\{n\}$ and $E'=E\cup E_S$, so that node $n$'s
discoverability is optimal in $G'$.
Because the inclusion of new edges $E_S$ will change the graph structure, the
probability transition $p_{ij}$, truncated absorbing probability $p_i^T$, and
truncated hitting time $h_i^T$ are all functions of the connection sources $S$,
denoted by $p_{ij}(S), p_i^T(S)$ and $h_i^T(S)$, respectively.
For the two definitions of node discoverability, we formulate two instances of
node discoverability optimization problem, respectively.

\begin{problem}[{\bf D-AP Maximization}] \label{p:max}
  Given budget $B$, the objective is to create new edges $E_S$ in graph $G$, so
  that D-AP is maximized in the new graph $G'=(V',E')$, i.e.,
  \begin{align}
    \max_{S\subseteq V}\FAP(S) &\triangleq \frac{1}{n}\sum_{i\in V}p_i^T(S) \\
    s.t. \sum_{s\in S} c_s &\leq B,
  \end{align}
  where $c_s$ denotes the cost of creating edge $(s,n)\in E_S$.

\end{problem}

\begin{problem}[{\bf D-HT Minimization}] \label{p:min}
  Given budget $B$, the objective is to create new edges $E_S$ in graph $G$, so
  that D-HT is minimized in the new graph $G'=(V',E')$, i.e.,
  \begin{align}
    \min_{S\subseteq V}\FHT(S) &\triangleq \frac{1}{n}\sum_{i\in V}h_i^T(S) \\
    s.t. \sum_{s\in S}c_s &\leq B,
  \end{align}
  where $c_s$ denotes the cost of creating edge $(s,n)\in E_S$.
\end{problem}

\remarks.
(1) For brevity, we sometimes omit $S$ in above equations if no confusion arises.

(2) The cost $c_s$ of creating an edge $(s,n)$ may have different meanings in
different applications.
For example, in Amazon's item network, the cost-per-click bid is an important
factor that Amazon uses to decide whether to display the target item on some
related item's detail page~\cite{amazon-ads}.
If the related item is popular, the cost-per-click bid will also be high
accordingly; therefore, the cost of creating an edge from a popular item is
usually higher than from a less popular item.
If $c_i\equiv const.,\forall i\in V$, the knapsack constraint then degenerates to
the cardinality constraint.

(3) We can also formulate more complex instances of the node discoverability
optimization problem, that maximize D-AP and minimize D-HT at the same time.
For example, using the previous extension of node discoverability, we can
formulate a composite optimization problem:
\begin{equation}\label{eq:composite_opt}
  \max_{S\subseteq V}\frac{1}{n}\sum_{i\in V}[\alpha p_i^T(S)+\beta h_i^T(S)]
  \quad s.t. \quad \sum_{s\in S} c_s\leq B.
\end{equation}

\subsection{Discussion on Node Discoverability Optimization}

We find that it is impractical to find the optimal solutions to
Problems~\ref{p:max} and~\ref{p:min} for large networks.

\begin{theorem}\label{thm:np-hard}
  Problems~\ref{p:max} and~\ref{p:min} are NP-hard.
\end{theorem}
\begin{proof}
  Please refer to the Appendix.
\end{proof}

While finding the optimal solutions is hard, we will now show that objectives
$\FAP$ and $\FHT$ satisfy submodularity and supermodularity respectively, which
will allow us to find provably near-optimal solutions to these two NP-hard
problems.

A set function $F\colon 2^V\mapsto\mathbb{R}$ is {\em submodular} if whenever
$S_1\subseteq S_2\subseteq V$ and $s\in V\backslash S_2$, it holds that
$F(S_1\cup\{s\})-F(S_1)\geq F(S_2\cup\{s\})-F(S_2)$, i.e., adding an element $s$
to set $S_1$ gains more score than adding $s$ to set $S_2$.
In addition, we say a submodular set function $F$ is {\em normalized} if
$F(\emptyset) = 0$.
We have the following conclusion about $\FAP$.

\begin{theorem} \label{thm:FAP}
  $\FAP(S)$ is a normalized non-decreasing submodular set function.
\end{theorem}
\begin{proof}
  Please refer to the Appendix.
\end{proof}

A set function $F\colon 2^V\mapsto\mathbb{R}$ is {\em supermodular} if
$-F$ is submodular.
We have the following conclusion about $\FHT$.

\begin{theorem} \label{thm:FHT}
  $\FHT(S)$ is a non-increasing supermodular set function.
\end{theorem}
\begin{proof}
  Please refer to the Appendix.
\end{proof}

Note that it is straightforward to convert $\FHT(S)$ into a normalized submodular
set function.
Because $\FHT(S)\in[0,T]$, thus $T-\FHT(S)$ is a normalized non-decreasing
submodular set function.

A commonly used heuristic to maximize a normalized non-decreasing submodular set
function $F$ with a {\em cardinality constraint} is the {\em simple greedy
  algorithm}.
The simple greedy algorithm starts with an empty set $S_0=\emptyset$, and
iteratively, in step $k$, adds an element $s_k$ which maximizes the {\em marginal
  gain}, i.e., $s_k=\argmax_{s\in V\backslash S_{k-1}} \delta(s;S_{k-1})$.
The marginal gain of an element $s$ regarding a set $S$ is defined by
\begin{equation}\label{eq:gain_cardinality}
  \delta(s;S)\triangleq F(S\cup\{s\})-F(S).
\end{equation}
The algorithm stops once it has selected enough elements, or the marginal gain
becomes less than a threshold.
The classical result of~\cite{Nemhauser1978} states that the output of the simple
greedy algorithm is at least a constant fraction of $1-1/e$ of the optimal value.

For the more general knapsack constraint, where each element has non-constant
cost, it is nature to redefine the marginal gain to
\begin{equation}\label{eq:gain_knapsack}
  \delta'(s;S)\triangleq \frac{F(S\cup\{s\})-F(S)}{c_s},
\end{equation}
and apply the simple greedy algorithm.
However, Khuller et al.~\cite{Khuller1999} prove that the simple greedy algorithm
using this marginal gain definition has unbounded approximation ratio.
Instead, they propose that one should consider the best single element as
alternative to the output of the simple greedy algorithm, which then guarantees a
constant factor $\frac{1}{2}(1-1/e)$ of the optimal value.
We describe this {\em budgeted greedy algorithm} in
Algorithm~\ref{alg:budgeted-greedy}.
Note that even in the case of knapsack constraint, the approximation ratio $1-1/e$
is achievable using a more complex algorithm~\cite{Khuller1999,Sviridenko2004}.
However, the algorithm requires $O(|V|^5)$ function evaluations which is
prohibitive for handling large graphs in our problem.

\begin{algorithm}[htp]
  \caption{Budgeted greedy algorithm in~\cite{Khuller1999}
    \label{alg:budgeted-greedy}}
  \KwIn{set $V$ and budget $B>0$}
  \KwOut{$S\subseteq V$ s.t. $c(S)\leq B$}
  $s^*\gets
  \argmax_{s\in V\wedge c_s\leq B}F(\{s\})$\tcc*{obtain the best single element}
  $S_1\gets\{s^*\}, S_2\gets\emptyset, U\gets V$\;
  \While(\tcc*[f]{construct $S_2$ using greedy heuristic})
  {$U\neq\emptyset$}{
    $s\gets\argmax_{i\in U}\delta'(i;S_2)$\;
    \lIf{$c(S_2)+c_s\leq B$}{$S_2\gets S_2\cup\{s\}$}
    $U\gets U\backslash\{s\}$\;
  }
  \KwRet $\argmax_{S\in\{S_1,S_2\}}F(S)$\tcc*{return the best solution}
\end{algorithm}

To implement the greedy algorithms, we need to compute the marginal gain for a
node.
The {\bf oracle call} in a greedy algorithm refers to the procedure of calculating
the marginal gain for a given node.
We list the formulas of computing marginal gains for the two optimization problems
under different constraints in Table~\ref{tab:marginal_gains}.
For a greedy algorithm, the {\em number of oracle calls} and the {\em time cost of
  an oracle call} dominate the computational complexity.
Both the two greedy algorithms need $O(|S|\cdot |V|)$ oracle calls, and this can
be further reduced by leveraging the trick of {\em lazy
  evaluation}~\cite{Minoux1978}, which, however, does not guarantee always
reducing the number of oracle calls.
Thus, reducing the time cost of an oracle call becomes key to improve the
computational efficiency of a greedy algorithm.
In the following section, we elaborate on how to implement an efficient oracle
call.
\footnote{Because submodularity is closed under non-negative linear combinations,
  and $-h_i^T(S)$ is proven to be submodular in Appendix, hence the objective in
  previous composite optimization problem~\eqref{eq:composite_opt} is still
  submodular, and it also falls into our framework.}

\begin{table}[htp]
  \centering
  \caption{Marginal gains in D-AP maximization and D-HT minimization
    \label{tab:marginal_gains}}
  \smallskip
  \setlength{\tabulinesep}{2pt}
  \begin{tabu}{c|cc}
    \hline\hline
    marginal gain & cardinality constraint & knapsack constraint \\
    \hline
    $\dAP(s;S)$ & $\FAP(S\cup\{s\})-\FAP(S)$
    & $\frac{\FAP(S\cup\{s\})-\FAP(S)}{c_s}$ \\
    \hline
    $\dHT(s;S)$ & $\FHT(S)-\FHT(S\cup\{s\})$
    & $\frac{\FHT(S)-\FHT(S\cup\{s\})}{c_s}$ \\
    \hline
  \end{tabu}
\end{table}

\section{Efficient Node Discoverability Optimization}
\label{sec:method}

Implementing the greedy algorithm boils down to implementing the oracle call.
In this section, we design fast methods to implement the oracle calls.
We first describe two basic methods, i.e., the dynamic programming (DP) approach,
and an estimation approach by simulating random walks (RWs).
Each method has its advantages and disadvantages: the DP approach is accurate but
computationally inefficient; the RW estimation approach is fast but inaccurate.
Then, to address the limitations of the two methods, we propose an {\em
  estimation-and-refinement} approach that is faster than DP, and also more
accurate than RW estimation.

For each method, we first describe how to calculate or estimate $p_i^T(S)$ and
$h_i^T(S)$ for a given set of connection sources $S$, then it will motivate us to
propose the oracle call implementation.

\subsection{Exact Calculation via Dynamic Programming}

\subsubsection{Calculating $p_i^T$ and $h_i^T$ Given $S$}

We can leverage the recursive definitions of truncated absorbing probability and
truncated hitting time to directly calculate the exact values of $p_i^T$ and
$h_i^T$ for each node $i$ using the dynamic programming (DP) approach.
This approach is described in Algorithm~\ref{alg:dp}, and it has time complexity
$O(T(|V|+|E|))$.

\SetKwFunction{Calculate}{DP}
\begin{algorithm}[htp]
  \caption{Exact calculation via DP\label{alg:dp}}
  \Fn{\Calculate{$T$}}{
    \tcp{initialization}
    $p_i^0\gets 0,\forall i\neq n$, and $p_n^t\gets 1, \forall t$\;
    $h_i^0\gets 0,\forall i$, and $h_n^t=0, \forall t$\;
    \tcp{recursively calculating $p_i^t$ and $h_i^t$}
    \For{$t\gets 1$ \KwTo $T$}{
      \ForEach{$i\in V$}{
        $p_i^t\gets\sum_{k\in\NbrO{i}}p_{ik}p_k^{t-1}$\;
        $h_i^t\gets 1+\sum_{k\in\NbrO{i}}p_{ik}h_k^{t-1}$\;
      }
    }
    \KwRet $\{p_i^T,h_i^T\}_{i\in V}$\;
  }
\end{algorithm}

\subsubsection{Implementing Oracle Call}

It is also convenient to use DP to implement the oracle call.
For example, if we want to calculate the marginal gain $\dAP(s;S) =
\FAP(S\cup\{s\}) - \FAP(S)$, we can apply Algorithm~\ref{alg:dp} for set $S$ and
$S\cup\{s\}$ respectively, and finally obtain the exact value of $\dAP(s;S)$.

This oracle call implementation has the same time complexity as Algorithm~\ref{alg:dp},
i.e., $O(T(|V|+|E|))$.
However, this time complexity is too expensive when the oracle call is used in the
greedy algorithm.
Because the greedy algorithm requires $|V|\times K$ oracle calls to obtain $K$
connection sources.
Therefore, the final time complexity is $O(KT|V|(|V|+|E|))$, which is unaffordable
when the graph is large.
For example, on the HepTh citation network with merely $27$K nodes, DP costs about
$38$ hours to calculate the marginal gain for each node.
This requires us to devise faster oracle call implementations.

\subsection{Estimation by Simulating Random Walks}
\label{ss:rw}

\subsubsection{Estimating $p_i^T$ and $h_i^T$ Given $S$}

Truncated absorbing probability and truncated hitting time are defined using
finite length random walks.
We thus propose an estimation approach to estimate $p_i^T$ and $h_i^T$ by
simulating a large number of random walks from each node.

To be more specific, we simulate $R$ independent random walks of length at most
$T$ from each node in $V$.
For the $r$-th random walk starting from node $i$, we assume that it terminates at
step $t_{ir}\leq T$, and we also use a binary variable $b_{ir}=1$ or $0$ to
indicate whether it finally hits target node $n$.
Then the following conclusion holds.
\begin{theorem}\label{thm:unbiased}
  $\hat{p}_i^T\triangleq\sum_{r=1}^Rb_{ir}/R$ and $\hat{h}_i^T
  \triangleq\sum_{r=1}^Rt_{ir}/R$ are unbiased estimators of $p_i^T$ and
  $h_{i}^T$, respectively.
  $\hFAP\triangleq\sum_{i\in V}\hat{p}_i^T/n$ and $\hFHT\triangleq\sum_{i\in
    V}\hat{h}_i^T/n$ are unbiased estimators of $\FAP$ and $\FHT$, respectively.
\end{theorem}
\begin{proof}
  By definition, $\{b_{ir}\}_{r=1}^R$ are i.i.d. Bernoulli random variables with
  success probability $p_i^T$.
  Hence, $\E(\hat{p}_i^T)=\sum_r\E(b_{ir})/R=p_i^T$.
  Similarly, $\{t_{ir}\}_{r=1}^R$ are i.i.d. random variables with expectation
  $\E(t_{ir})=h_i^T$.
  Hence, $\E(\hat{h}_i^T)=\sum_r\E(t_{ir})/R=h_i^T$.
  Then, it is straightforward to obtain that $\hFAP$ and $\hFHT$ are also unbiased
  estimators of $\FAP$ and $\FHT$, respectively.
\end{proof}

Furthermore, we can bound the number of random walks $R$ in order to guarantee the
estimation precision by applying the Hoeffding inequality~\cite{Hoeffding1963}.
\begin{theorem}\label{thm:bound1}
  Given constants $\delta,\epsilon>0$, and set $S$, in order to guarantee
  $P(|\hFAP(S) - \FAP(S)|\geq\delta )\leq\epsilon$, and $P(|\hFHT(S) -
  \FHT(S)|\geq\delta T)\leq\epsilon$, the number of random walks $R$ should be at
  least $\frac{1}{2n\delta^2} \ln\frac{2}{\epsilon}$.
\end{theorem}
\begin{proof}
  Please refer to Appendix.
\end{proof}

Thanks to the recent development of MCMC simulation
systems~\cite{Fogaras2005,Kyrol2013,Liu2016c}, we are now able to simulate
billions of random walks on just a PC.
We re-implement an efficient random walk simulation system based
on~\cite{Kyrol2013}.
In our implementation, a walk is encoded by a $64$-bit C++ integer, as illustrated
in Figure~\ref{fig:walk_encoding}.
Hence, simulating $1$ billion walks requires only $8$GB RAM (without considering
other space costs).
Based on this powerful RW simulation system, we can obtain estimates $\hat{p}_i^T$
and $\hat{h}_i^T$ by Alg.~\ref{alg:rw-estimate}, and hence obtain $\hFAP$ and
$\hFHT$ by the estimators in Theorem~\ref{thm:unbiased}.

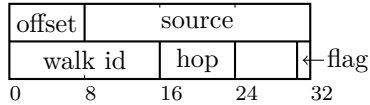
\begin{figure}[ht]
  \centering
  \begin{tikzpicture}[
  txt/.style={font=\footnotesize, inner sep=1pt},
  tick/.style={font=\scriptsize,inner sep=0pt, yshift=-8pt, anchor=south west},
  ]
\draw[thick] (0,0) rectangle (4,1);
\draw[thick] (0,.5) -- (4,.5);
\draw[thick] (1,1) -- (1,.5);
\draw[thick] (2,.5) -- (2,0);
\draw[thick] (3,.5) -- (3,0);
\draw[thick] (3.825,.5) -- (3.825,0);

\node[txt] at (.5,.75) {offset};
\node[txt] at (2.5,.75) {source};
\node[txt] at (1,.25) {walk id};
\node[txt] at (2.5,.25) {hop};
\node[txt] at (3.44,.25) {};
\node[txt] at (4.5,.25) (flag) {flag};
\draw[->] (flag) -- (3.9,.25);

\draw (1,0) -- (1,.05);
\draw (2,1) -- (2,.95);
\draw (3,1) -- (3,.95);

\node[tick] at (0,0) {$0$};
\node[tick] at (1,0) {$8$};
\node[tick] at (2,0) {$16$};
\node[tick] at (3,0) {$24$};
\node[tick] at (4,0) {$32$};

\end{tikzpicture}

  \caption{Walk encoding.
    In the implementation~\cite{Kyrol2013}, walks are grouped into buckets by the
    nodes where they are currently resident, and hence a walk only needs to record
    its relative ``offset'' to the first node in the corresponding bucket to know
    its resident node.
    ``source'' records the starting node of the walk.
    ``walk id'' records the ID of the walk that starts from the same ``source''.
    ``hop'' records the number of hops the walk has walked.
    ``flag'' is used to indicate whether the walk finally hits target node.}
  \label{fig:walk_encoding}
\end{figure}

\SetKwFunction{RawEst}{RWEstimate}
\begin{algorithm}[ht]
  \caption{Estimating $p_i^T$ and $h_i^T$ by simulating RWs}
  \label{alg:rw-estimate}
  \tcp{$R$ is the number of walks, $T$ is the maximum walk length}
  \Fn{\RawEst{$R, T$}}{
    \ForEach{node $i\in V$}{
      \For{$r\gets 1$ \KwTo $R$}{
        start a walk from $i$, and walk at most $T$ steps\;
        $b_{ir}\gets$ whether the walk hits target node $n$\;
        $t_{ir}\gets$ number of steps walked\;
      }
      $\hat{p}_i^T\gets\sum_r b_{ir}/R$\;
      $\hat{h}_i^T\gets\sum_r t_{ir}/R$\;
    }
    \KwRet $\{\hat{p}_i^T, \hat{h}_i^T\}_{i\in V}$\;
  }
\end{algorithm}

\subsubsection{Implementing Oracle Call}

To estimate the marginal gain of selecting a node $s\in V\backslash S$ as a
connection source, we need to estimate the change of truncated absorbing
probability/hitting time $\Dhp{i}{T}(s)\triangleq\hat{p}_i^T(S')-\hat{p}_i^T(S)$
and $\Dhh{i}{T}(s)\triangleq\hat{h}_i^T(S)-\hat{h}_i^T(S')$ for each node $i\in
V$, where $S'\triangleq S\cup\{s\}$.
Then, the marginal gains of $s$ are estimated by $\hdAP(s;S)=\frac{1}{n}\sum_{i\in
  V}\Dhp{i}{T}(s)/c_s$ and $\hdHT(s;S)=\frac{1}{n}\sum_{i\in V}\Dhh{i}{T}(s)/c_s$.

It is not necessary to re-simulate all the walks.
Because the inclusion of a node $s$ into $S$ only affects the walks that passed
through $s$, we only need to re-simulate these affected walks after $s$ in their
sample paths, and estimate $\{\Dhp{i}{T}(s), \Dhh{i}{T}(s)\}_{i\in V}$
incrementally.

In more detail, we first query the walks that hit node $s$, denoted by
$\CW_s\triangleq\{(w,t)\colon$walk $w$ hits node $s$ for the first time at
$t<T\}$.
For each walk-step pair $(w,t)\in\CW_s$, we update walk $w$'s sample path after
$s$, i.e., re-walk $w$ from $s$ for the remaining (at most) $T-t$ steps.
Then, walk $w$'s statistics are updated, i.e., its hit/miss indicator
$b_w\in\{0,1\}$ and hitting time $t_w\in[0,T]$.
Finally, we obtain $\Dhp{i}{T}(s)$ and $\Dhh{i}{T}(s)$ for each $i\in \{i\colon
i\text{ is the source of a walk }w\in \CW_s\}$ (and for the other nodes,
$\Dhp{i}{T}(s)=\Dhh{i}{T}(s)=0$).

To apply such an approach, the number of walks $R$ needs to satisfy the following
condition\footnote{ Estimating $\dAP$ (or $\dAP$) requires more walks than
  estimating $\FAP$ (or $\FHT$) because in the later case we do not need to
  guarantee a per-node-wise estimation accuracy.}.

\begin{theorem}\label{thm:bound2}
  Given constants $\delta,\epsilon>0$, and set $S$, in order to guarantee
  $P(\exists s\in V\backslash S,|\hdAP(s;S) - \dAP(s;S)|\geq\delta/c_s)\leq
  \epsilon$, and $P(\exists s\in V\backslash S,|\hdHT(s;S) - \dHT(s;S)|\geq
  \delta T/c_s) \leq \epsilon$, the number of random walks $R$ should be at
  least $\frac{2}{n\delta^2} \ln\frac{4n}{\epsilon}$.
\end{theorem}
\begin{proof}
  Please refer to Appendix.
\end{proof}

Because we only need to update a small fraction of the walks, oracle call
implemented by simulating random walks will be much more efficient than solving
DP.
We give an example of estimating marginal gain $\dAP(s;S)$ of a node $s$ in
Algorithm~\ref{alg:gain-ht}.

\SetKwFunction{RawDelta}{RawDeltaAP}
\begin{algorithm}
  \caption{Estimating $\dAP(s;S)$ by RW estimation\label{alg:gain-ht}}
  \Fn{\RawDelta($s,T$)}{
    $\Dhp{i}{T}\gets 0, \forall i\in V,\,U^T\gets\emptyset$\tcp*{initialization}
    $\CW_s\gets\{(w,t)\colon$walk $w$ hits node $s$ at time $t<T\}$\;
    \ForEach{$(w,t)\in\CW_s$}{
      re-walk $w$ from $s$ for at most $T-t$ steps\;
      $\Delta b_w\gets b_w'-b_w$\tcp*{$b_w'$ is the new walk's hit/miss indicator}
      $\Dhp{i(w)}{T}\gets \Dhp{i(w)}{T} + \Delta b_w/R$\;
      $U^T\gets U^T\cup \{i(w)\}$\tcp*{record nodes whose hitting time changes}
    }
    \KwRet $U^T,\{\Dhp{i}{T}\}_{i\in U^T}$%
    \tcp*{$\hdAP(s;S)=\frac{1}{n}\sum_{i\in U^T}\Dhp{i}{T}/c_s$}
  }
\end{algorithm}

\subsection{An Estimation-and-Refinement Approach}
\label{ss:refinement}

So far we have developed two methods, namely, DP and RW estimation.
Each method has its advantages and disadvantages: DP is accurate but
computationally inefficient; RW estimation is fast but inaccurate.
To address these limitations, we propose an {\em estimation-and-refinement}
approach, that is faster than DP, and also more accurate than RW estimation.

\subsubsection{Estimating $p_i^T$ and $h_i^T$ Given $S$}

The basic idea of the estimation-and-refinement approach is that, we first use the
RW estimation to obtain {\em raw estimates} of truncated absorbing
probability/hitting time, then we improve their accuracy by an additional {\em
  refinement} step.

In the first stage of the algorithm, we simulate {\em fewer} and {\em shorter}
walks on the graph than in the previous RW estimation.
Let $D\in[0,T]$ be a given constant, we simulate $R$ walks with maximum length
$T-D$ (Line~\ref{l:est} of Algorithm~\ref{alg:est-refine}).
Here $R$ could be less than the required least number of walks.
After this step, we obtain {\em raw estimates} $\{\hat{p}_i^{T-D},
\hat{h}_i^{T-D}\}_{i\in V}$ using the previously develop RW estimation.
At first glance, if $D\neq 0$, these raw estimates are useless, because to
estimate D-AP and D-HT, we have to know $\hat{p}_i^T$ and $\hat{h}_i^T$; and they
are also inaccurate if $R$ does not satisfy the requirement of
Theorem~\ref{thm:bound1}.

In the second stage, we propose an additional {\em refinement} step that leverage
the raw estimates to obtain $\hat{p}_i^T$ and $\hat{h}_i^T$, and also improve
estimation accuracy simultaneously (Line~\ref{l:refine} of
Algorithm~\ref{alg:est-refine}).
The refinement is due to the observation that the recursive definitions of
absorbing probability and hitting time share the common structure of a {\em
  harmonic function}~\cite{Doyle1984}, that the function value at $x$ is a
smoothed average of the function values at $x$'s neighbors.
Thus, if we have obtained raw estimate for each node, we can refine a node's
estimate by averaging the raw estimates at its neighbors, and the smoothed
estimate will be more accurate than the raw estimate.

\SetKwFunction{FineEst}{EstimateAndRefine}
\SetKwFunction{Refine}{Refine}
\begin{algorithm}[ht]
  \caption{An estimation-and-refinement approach}
  \label{alg:est-refine}
  \tcp{$D$ is the refinement depth}
  \Fn{\FineEst{$R, T, D$}}{
    $\{\hat{p}_i^{T-D},\hat{h}_i^{T-D}\}_{i\in V}\gets$\RawEst{$R,T-D$}\;
    \label{l:est}
    \KwRet\Refine{$\{\hat{p}_i^{T-D},\hat{h}_i^{T-D}\}_{i\in V}, D$}\;
    \label{l:refine}
  }
  \Fn{\Refine{$\{\hat{p}_i^{T-D},\hat{h}_i^{T-D}\}_{i\in V}, D$}}{
    \For{$t\gets T-D+1$ \KwTo $T$}{
      \ForEach{$i\in V$}{
        $\hat{p}_i^t\gets\sum_{k\in\NbrO{i}}p_{ik}\hat{p}_k^{t-1}$\;
        $\hat{h}_i^t\gets 1+\sum_{k\in\NbrO{i}}p_{ik}\hat{h}_k^{t-1}$\;
      }
    }
    \KwRet $\{\hat{p}_i^T,\hat{h}_i^T\}_i$\;
  }
\end{algorithm}

\begin{figure}[htp]
  \centering
  \begin{tikzpicture}[
  nd/.style={draw,circle,thick,minimum size=15pt, inner sep=0pt,
  font=\footnotesize},
  arr/.style={->, thick},
  txt/.style={font=\footnotesize, inner sep=1pt},
  txte/.style={txt, anchor=east},
  txtw/.style={txt, anchor=west},
  ]

  \node[nd,draw=blue] (i) {$i$};
  \node[nd,below left=.6 and .4 of i] (j1) {$j_1$};
  \node[nd,below right=.6 and .4 of i] (j2) {$j_2$};
  \node[nd,below left=.6 and 0 of j1] (k1) {$k_1$};
  \node[nd,below right=.6 and 0 of j1] (k2) {$k_2$};
  \node[nd,below left=.6 and 0 of j2] (k3) {$k_3$};
  \node[nd,below right=.6 and 0 of j2] (k4) {$k_4$};
  \draw[arr] (i) -- (j1);
  \draw[arr] (j1) -- (k1);
  \draw[arr] (j1) -- (k2);
  \draw[arr] (i) -- (j2);
  \draw[arr] (j2) -- (k3);
  \draw[arr] (j2) -- (k4);

  \coordinate[left=1.6 of i] (A);
  \node[txte] at (A |- i) {$D=0$};
  \node[txte] at (A |- j1) {$D=1$};
  \node[txte] at (A |- k1) {$D=2$};

  \coordinate[right=1.6 of i] (B);
  \node[txtw] at (B |- i) {$\hat{p}_i^T$};
  \node[txtw] at (B |- j2) {$\hat{p}_j^{T-1}, j\in\{j_1,j_2\}$};
  \node[txtw] at (B |- k4) {$\hat{p}_k^{T-2}, k\in\{k_1,k_2,k_3,k_4\}$};
\end{tikzpicture}

  \caption{Illustration of refining $\hat{h}_i^T$ by
    Algorithm~\ref{alg:est-refine}.
    If $D=1$, $\{\hat{p}_j^{T-1}\}_j$ are used for refining $\hat{p}_i^T$; if
    $D=2$, $\{p_k^{T-2}\}_k$ are used for refining $\hat{p}_i^T$.
  }
  \label{fig:refine}
\end{figure}
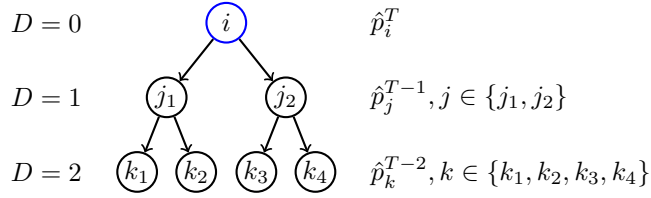

We use the graph in Figure~\ref{fig:refine} to illustrate how the
estimation-and-refinement method is used to obtain $\hat{p}_i^T$.
Let $D=1$, we first obtain raw estimate $\hat{p}_j^{T-1}$ for each node $j\in V$
by simulating random walks of length $T-1$.
To refine the estimate of a node, say, node $i$, we can leverage the relation
$\hat{p}_i^T = \sum_{j\in\NbrO{i}}p_{ij}\hat{p}_j^{T-1} =
p_{ij_1}\hat{p}_{j_1}^{T-1} + p_{ij_2}\hat{p}_{j_2}^{T-1}$, which smooths the raw
estimates of $i$'s out-neighbors, and intuitively, we are using the walks of
neighbor $j_1$ and $j_2$, i.e., $2R$ walks, to estimate $p_i^T$, which will be
more accurate than using only $R$ walks of node $i$.
Similarly, we can use $i$'s two-hop neighbors' raw estimates
$\{\hat{h}_k^{T-2}\}_k$ to refine $i$'s estimate ($D=2$), and we will obtain even
better estimate.
When $D=T$, there is no need to run the first step, and the refinement actually
becomes DP, which obtains the true value of $p_i^T$.

We now formally show that the variance of estimates obtained by the
estimation-and-refinement approach is indeed no larger than the variance of
estimates obtained by RW estimation.
Let us consider the random walks starting from an arbitrary node $i\in V$.
At the first step of the walk, assume that $R_j$ of the walks are at a neighbor
node $j\in\NbrO{i}$.
It is easy to see that $[R_j]_{j\in\NbrO{i}}$ follows a multinomial distribution
parameterized by $[p_{ij}]_{j\in\NbrO{i}}$ and $R$, and $\E(R_j)=Rp_{ij}$.
Then, the RW estimator in~\autoref{ss:rw} estimates $p_i^T$ by
\[
  \hat{p}_i^T
  = \frac{1}{R}\sum_{r=1}^Rb_{ir}^T
  = \frac{1}{R}\sum_{j\in\NbrO{i}}\sum_{r=1}^{R_j}b_{jr}^{T-1}
\]
where $b_{ir}^t$ is a binary variable indicating whether a walk starting from node
$i$ finally hits target node $n$ within $t$ steps.
The variance of above estimator satisfies
\begin{align*}
  var(\hat{p}_i^T)
  &\geq
    \E\left[var\left(
    \frac{1}{R}\sum_{j\in\NbrO{i}}\sum_{r=1}^{R_j}b_{jr}^{T-1}\big|\{R_j\}
    \right)\right]
  = \frac{1}{R}\sum_{j\in\NbrO{i}}p_{ij}\cdot var(b_{jr}^{T-1}) \\
  &= \sum_{j\in\NbrO{i}}\frac{p_{ij}^2}{\E(R_j)}\cdot var(b_{jr}^{T-1})
  \geq \sum_{j\in\NbrO{i}}\frac{p_{ij}^2}{R}\cdot var(b_{jr}^{T-1})
\end{align*}
where the inequality is due to the relation $var(X) = var[\E(X|Y)] + \E[var(X|Y)]
\geq \E[var(X|Y)]$.

In contrast, the estimation-and-refinement approach estimates $p_i^T$ by
\[
  \check{p}_i^T
  =\sum_{j\in\NbrO{i}}\frac{p_{ij}}{R}\sum_{r=1}^Rb_{jr}^{T-1},
\]
and its variance is
\[
  var(\check{p}_i^T)
  =\sum_{j\in\NbrO{i}}\frac{p_{ij}^2}{R}\cdot var(b_{jr}^{T-1})
  \leq var(\hat{p}_i^T).
\]

Hence, the estimation-and-refinement approach indeed has smaller variance than the
RW estimator for estimating $p_i^T$.
It is straightforward to extend the above analysis to show that the
estimation-and-refinement also has smaller variance for estimating $h_i^T$.

\subsubsection{Implementing Oracle Call}

Using the similar idea, we design an estimation-and-refinement approach for better
estimating the marginal gain of a node.
We observe that $\Dp{i}{t}(s)$ and $\Dh{i}{t}(s)$ exhibit similar recursive
definitions as $p_i^t$ and $h_i^t$, i.e., for $t=0,\ldots,T$ and denote
$S'=S\cup\{s\}$, then
\begin{align*}
  \Dp{i}{t}(s)
  &= \sum_{j\in\NbrO{i}}\left[
    p_{ij}(S')p_{j}^{t-1}(S')-p_{ij}(S)p_{j}^{t-1}(S)
    \right] \\
  &=
    \begin{cases}
      \sum_{j\in\NbrO{i}}p_{ij}\Dp{i}{t-1}(s), & i\neq s, \\
      \sum_{j\in\NbrO{s}}\left[
        p_{sj}(S')p_{j}^{t-1}(S')-p_{sj}(S)p_{j}^{t-1}(S)
      \right], & i=s,
    \end{cases}
\end{align*}
and
\begin{align*}
  \Dh{i}{t}(s)
  &= \sum_{j\in\NbrO{i}}\left[
    p_{ij}(S)h_{j}^{t-1}(S)-p_{ij}(S')h_{j}^{t-1}(S')
    \right] \\
  &=
    \begin{cases}
      \sum_{j\in\NbrO{i}}p_{ij}\Dh{i}{t-1}(s), & i\neq s, \\
      \sum_{j\in\NbrO{s}}\left[
        p_{sj}(S)h_{j}^{t-1}(S)-p_{sj}(S')h_{j}^{t-1}(S')
      \right], & i=s.
    \end{cases}
\end{align*}
Note that if $i$ is selected as a connection source, then transition probabilities
from $i$ to other nodes will change, i.e., $p_{ij}(S)\neq p_{ij}(S')$.

The above recursive relations allow us to use the random walk to obtain raw
estimates of $\Delta p_{i}^{T-D}(s)$ and $\Delta h_{i}^{T-D}(s)$, and then refine
their precision similar to the previous discussion.
We give an example of estimating and refining $\dAP(s;S)$ in
Algorithm~\ref{alg:gain-refine}.

\SetKwFunction{Estimate}{EstimateAndRefine\_DeltaAP}
\SetKwFunction{Refine}{Refine\_DeltaAP}
\begin{algorithm}[htp]
  \caption{Estimating $\dAP(s;S)$ by estimation-and-refinement}
  \label{alg:gain-refine}
  \Fn{\Estimate{$s, T, D$}}{
    $U^{T-D},\{\Dhp{j}{T-D}\}_j\gets$\RawDelta{$s,T-D$}\;
    \KwRet\Refine{$U^{T-D},\{\Dhp{j}{T-D}\}_j$}\;\label{l:gain-refine}
  }
  \Fn{\Refine{$s,U^{T-D},\{\Dhp{j}{T-D}\}_j$}}{
    \For{$t\gets T-D+1$ \KwTo $T$}{
      \ForEach(\tcp*[f]{$i\neq s$}){$j\in U^{t-1}$}{
        \ForEach{$i\in\NbrI{j}\wedge i\neq s$}{
          $\Dhp{i}{t}\gets\Dhp{i}{t}+p_{ij}\Dhp{j}{t-1}$\;
          $U^t\gets U^t\cup\{i\}$\;
        }
      }
      $U^t\gets\{s\}$\tcp*{$i=s$}
      $\Dhp{s}{t}\gets
      \sum_{j\in\NbrO{s}}\left[
        p_{sj}(S')\hat{p}_j^{t-1}(S')-p_{sj}(S)\hat{p}_j^{t-1}(S) \right]$\;
    }
    \KwRet $U^T,\{\Dhp{i}{T}\}_{i\in U^T}$\;
  }
\end{algorithm}

\section{Validating the Estimation Methods}
\label{sec:experiment}

In this section, we devote to validate the proposed estimation methods, and in the
next section we give some applications of the node discoverability optimization
problem.
We conduct experiments on real graphs of various types and scales to validate the
accuracy and efficiency of our proposed methods.
First, we briefly introduce the datasets.
Then, we compare the estimation accuracy and computational efficiency for
estimating truncated absorbing probability/hitting time and marginal gain.
Finally, we evaluate the performance of greedy algorithm by comparing with
baseline methods.

\subsection{Datasets}

We use public available graphs of different types and scales from the SNAP graph
repository~\cite{SNAP} as our test beds.
For an edge in a graph, we assume it has a unitary weight one.
The basic statistics of these graphs are summarized in Table~\ref{tab:data}.

All the experiments are performed on a laptop running 64-bit Ubuntu 16.04 LTS,
with a dual-core 2.66GHz Intel i3 CPU, 8GB of main memory, and a 500GB 5400RPM
hard disk.

\begin{table}[htp]
  \small
  \centering
  \caption{Graph statistics \label{tab:data}}
  \begin{tabular}{l|l|r|r}
    \hline\hline
    {\bf graph} & {\bf description}  & {\bf \# of nodes} & {\bf \# of edges} \\
    \hline
    HepTh       & citation network, directed    & $27,400$    & $355,057$    \\
    Enron       & email communication           & $33,696$    & $180,811$    \\
    Gowalla     & location based social network & $196,591$   & $950,327$    \\
    DBLP        & coauthor network              & $317,080$   & $1,049,866$  \\
    Amazon      & product network               & $334,863$   & $925,872$    \\
    YouTube     & friendship network            & $1,134,890$ & $2,987,624$  \\
    Patents     & citation network, directed    & $3,774,768$ & $18,204,370$ \\
    \hline
    Weibo~\cite{Zhao2014a}
                & follower network, directed    & $323,069$   & $1,937,008$ \\
    Douban~\cite{Zhao2011}
                & follower network, directed    & $1,760,297$ & $23,379,254$ \\
    \hline
  \end{tabular}
\end{table}

\subsection{Evaluating Absorbing Probability/Hitting Time Estimation
  Accuracy}

In the first experiment, we evaluate the accuracy of estimating $p_i^T(S)$ and
$h_i^T(S)$ by different methods when connection sources $S$ are given.
We set $S=V$, i.e., connect every node in the graph to target node $n$ with weight
one.
This corresponds to the case that D-AP is maximum and D-HT is minimum.
DP in Algorithm~\ref{alg:dp} is an exact method which hence allows us to obtain
the groundtruth $p_i^T$ and $h_i^T$ on a graph.
In this experiment, we use three smaller graphs, HepTh, Enron, and Gowalla, for
the convenience of calculating groundtruth.

First, we show how close the estimate $\hat{p}_i^T$ (or $\hat{h}_i^T$) is to its
groundtruth $p_i^T$ (or $h_i^T$) by evaluating their ratio $\hat{p}_i^T/p_i^T$ (or
$\hat{h}_i^T/h_i^T$).
We randomly pick a few nodes from each graph, and estimate $p_i^T$ and $h_i^T$ for
each node sample $i$ using different methods (or parameter settings) and different
number of RWs.
We then calculate the ratio $\hat{p}_i^T/p_i^T$ and $\hat{h}_i^T/h_i^T$ for each
node sample $i$, and show their values versus the number of RWs as scatter plots
in Figures~\ref{fig:scatters_p} and~\ref{fig:scatters_h}.
In addition, we roughly separate nodes into two categories, i.e., low degree nodes
which have degrees smaller than the average degree of the graph, and high degree
nodes which have degrees larger than the average degree, to study the difference
of their estimation accuracy.
We observe that both the RW estimation approach and the estimation-and-refinement
approach can provide good estimates, and generally, the estimates become more
accurate when the number of walks per node increases.
Furthermore, the estimation-and-refinement approach indeed can refine the
estimation accuracy significantly, and with larger refinement depth $D$, we obtain
even more accurate estimates.
For nodes in different categories, however, we do not observe significant
estimation accuracy difference, indicating that these methods are not sensitive to
node degrees.

\begin{figure}[t]
  \small
  \centering
  \begin{tabular}{CCC}
    \hline
    HepTh & Enron & Gowalla \\
    \hline
    \multicolumn{3}{@{}c@{}}{
    \subfloat[$\hat{p}_i^T/p_i^T$ for low degree nodes $i$]{%
      \includegraphics[width=.33\linewidth]{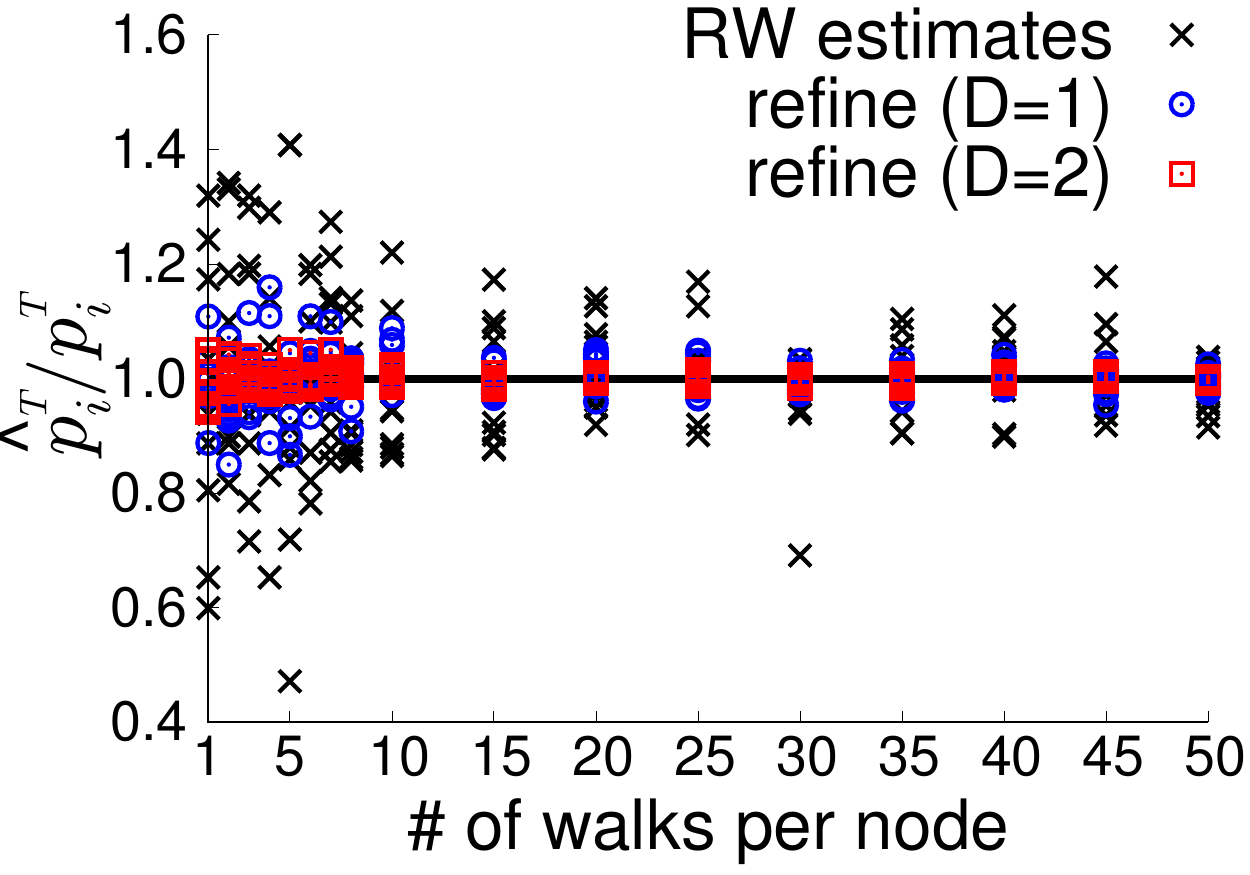}%
      \includegraphics[width=.33\linewidth]{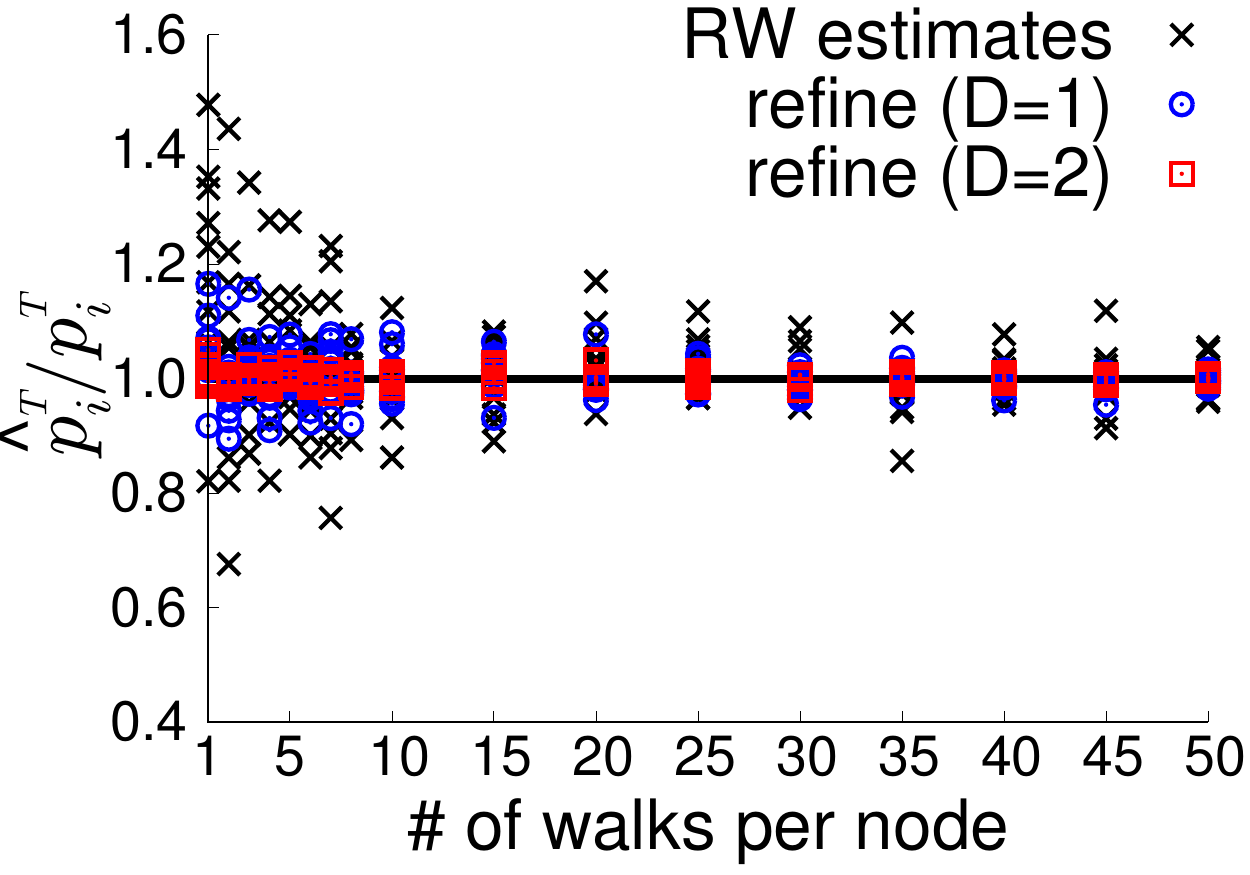}%
      \includegraphics[width=.33\linewidth]{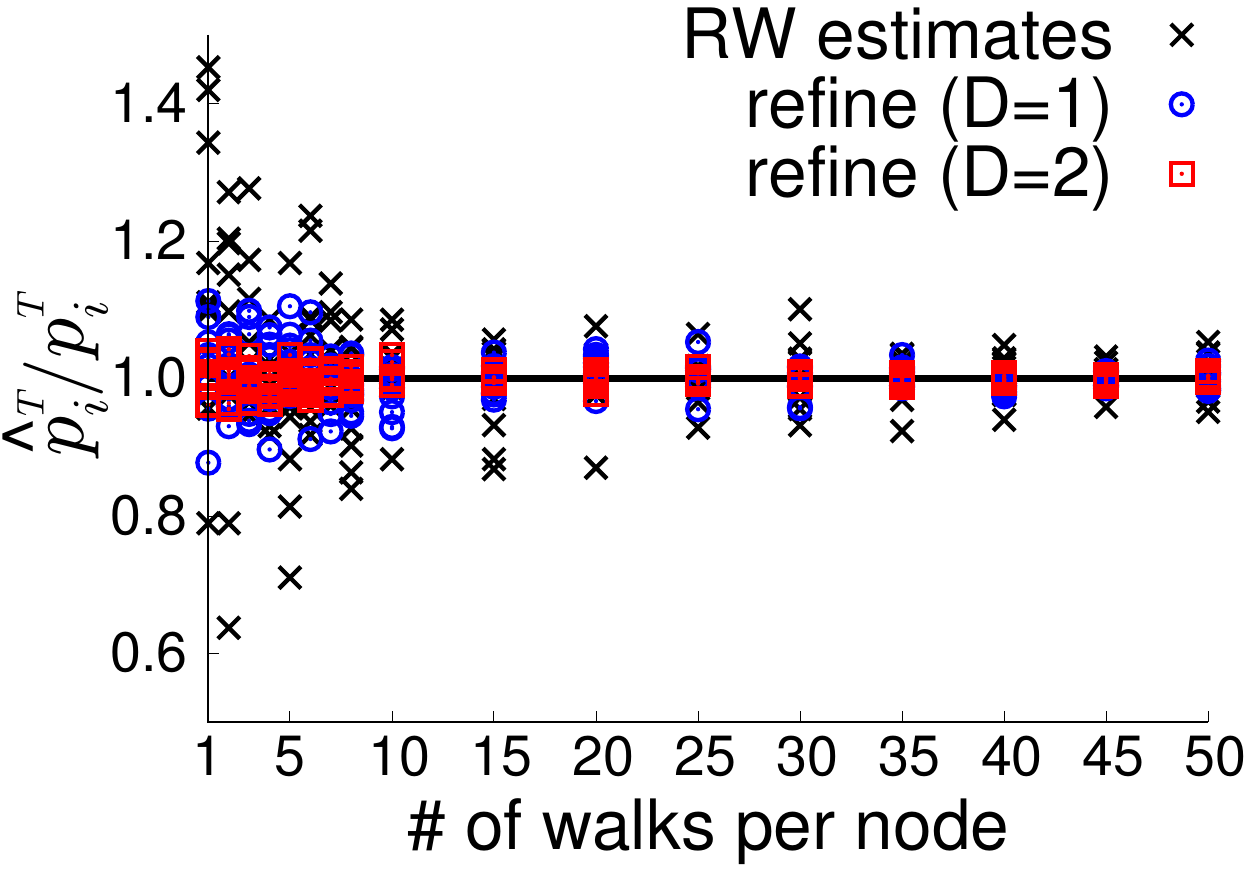}}
    } \\
    \multicolumn{3}{@{}c@{}}{
    \subfloat[$\hat{p}_i^T/p_i^T$ for high degree nodes $i$]{%
      \includegraphics[width=.33\linewidth]{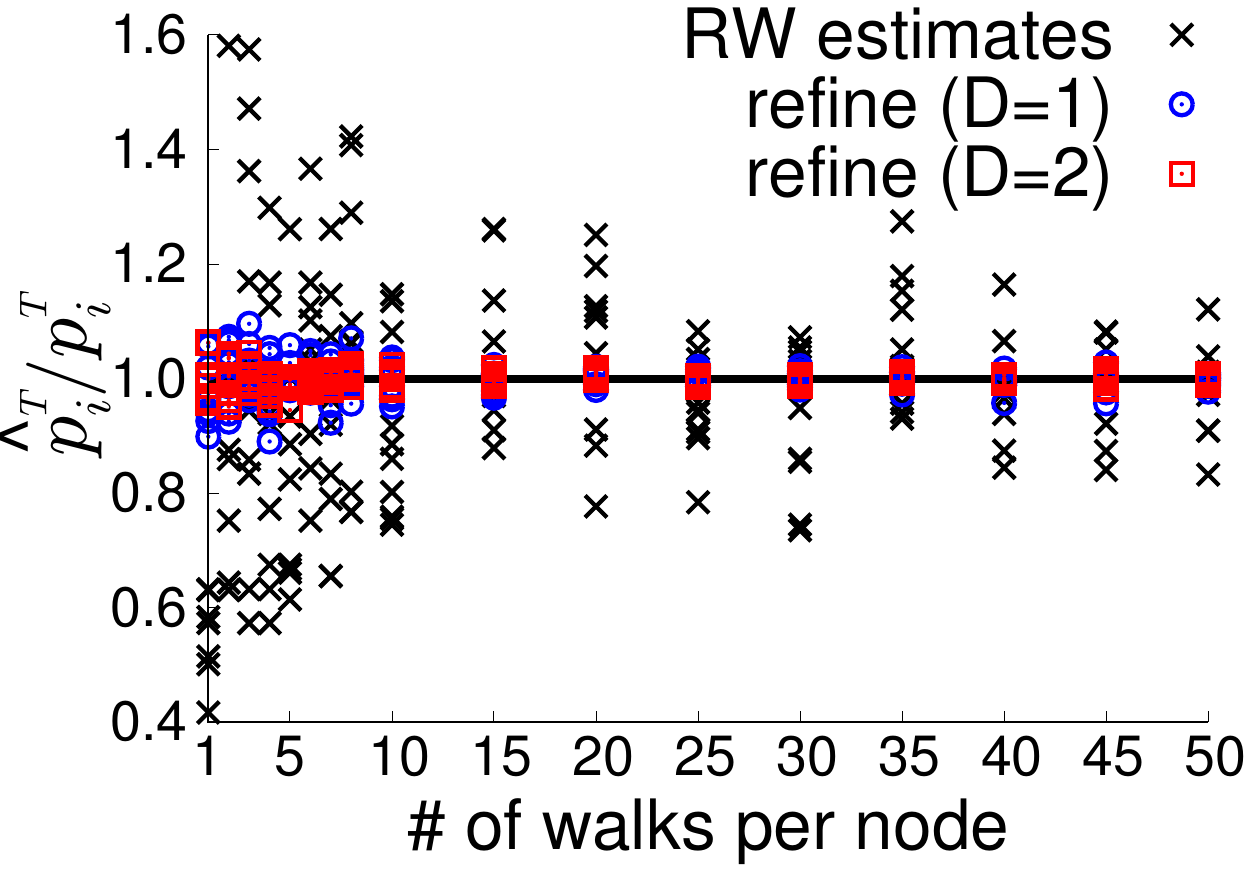}%
      \includegraphics[width=.33\linewidth]{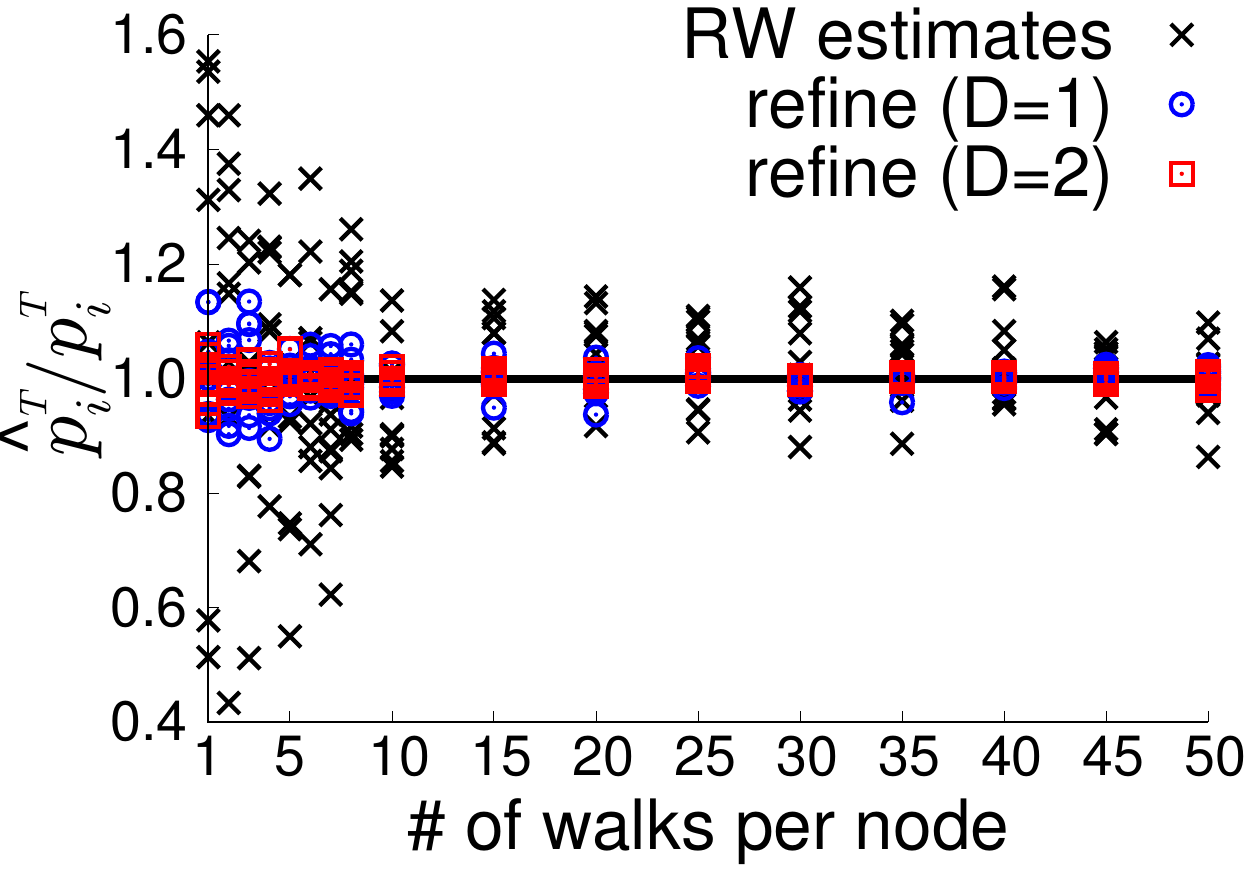}%
      \includegraphics[width=.33\linewidth]{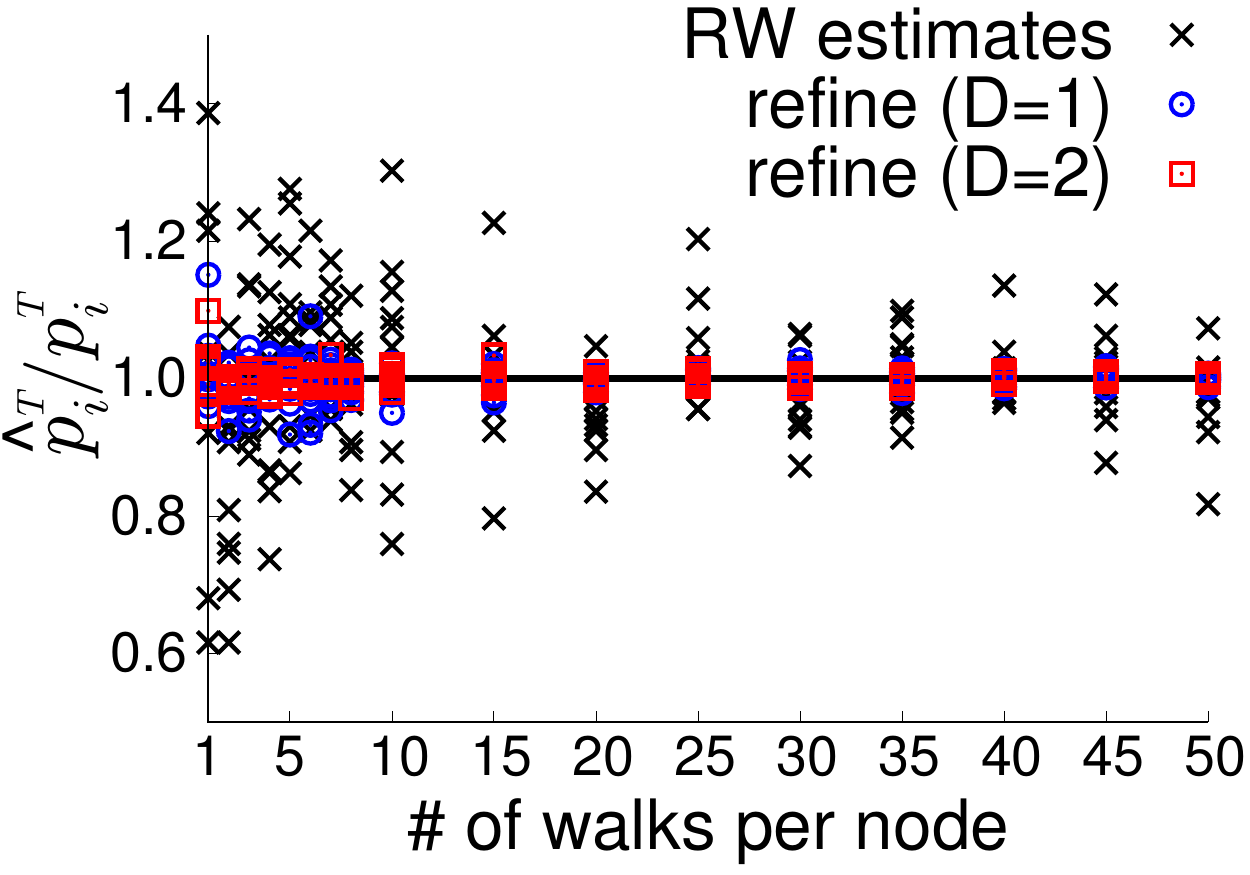}}
    }
  \end{tabular}
  \caption{Estimates of $p_i^T$ on three graphs.
    Each scatter is an estimate for a node sample.
    The low (or high) degree nodes refer to nodes with degree smaller (or larger)
    than the average degree in the graph. ($T=10$)}
  \label{fig:scatters_p}
\end{figure}

\begin{figure}[t]
  \small
  \centering
  \begin{tabular}{CCC}
    \hline
    HepTh & Enron & Gowalla \\
    \hline
    \multicolumn{3}{@{}c@{}}{
    \subfloat[$\hat{h}_i^T/h_i^T$ for low degree nodes $i$]{%
      \includegraphics[width=.33\linewidth]{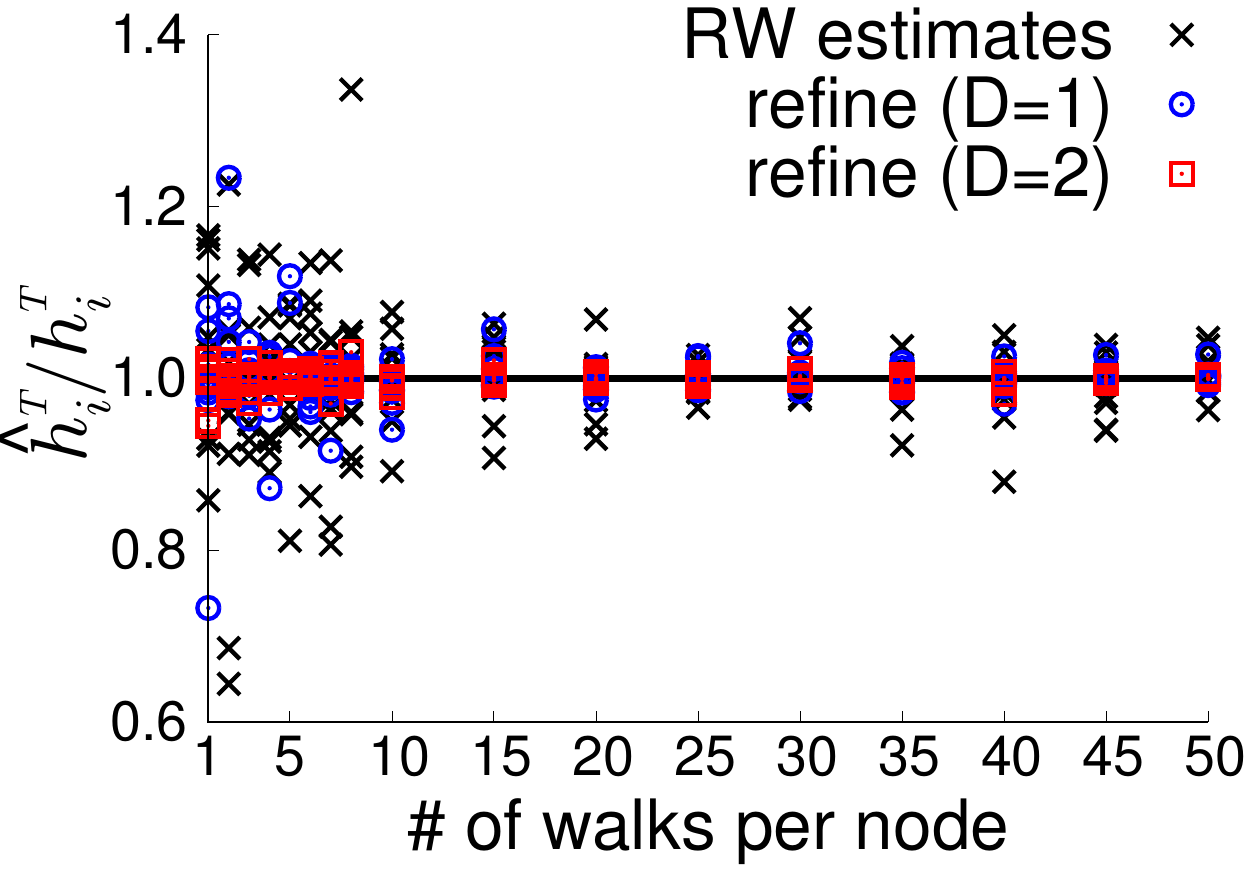}%
      \includegraphics[width=.33\linewidth]{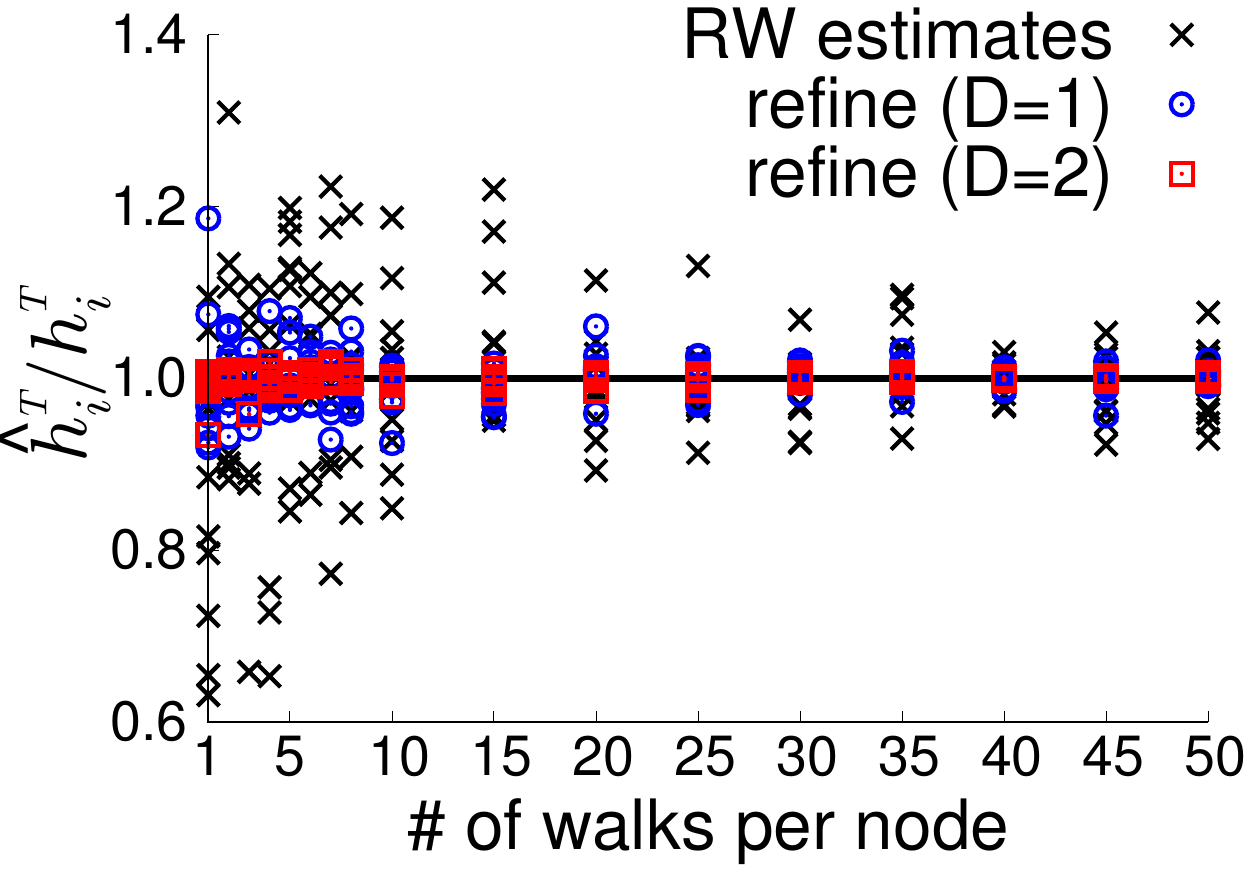}%
      \includegraphics[width=.33\linewidth]{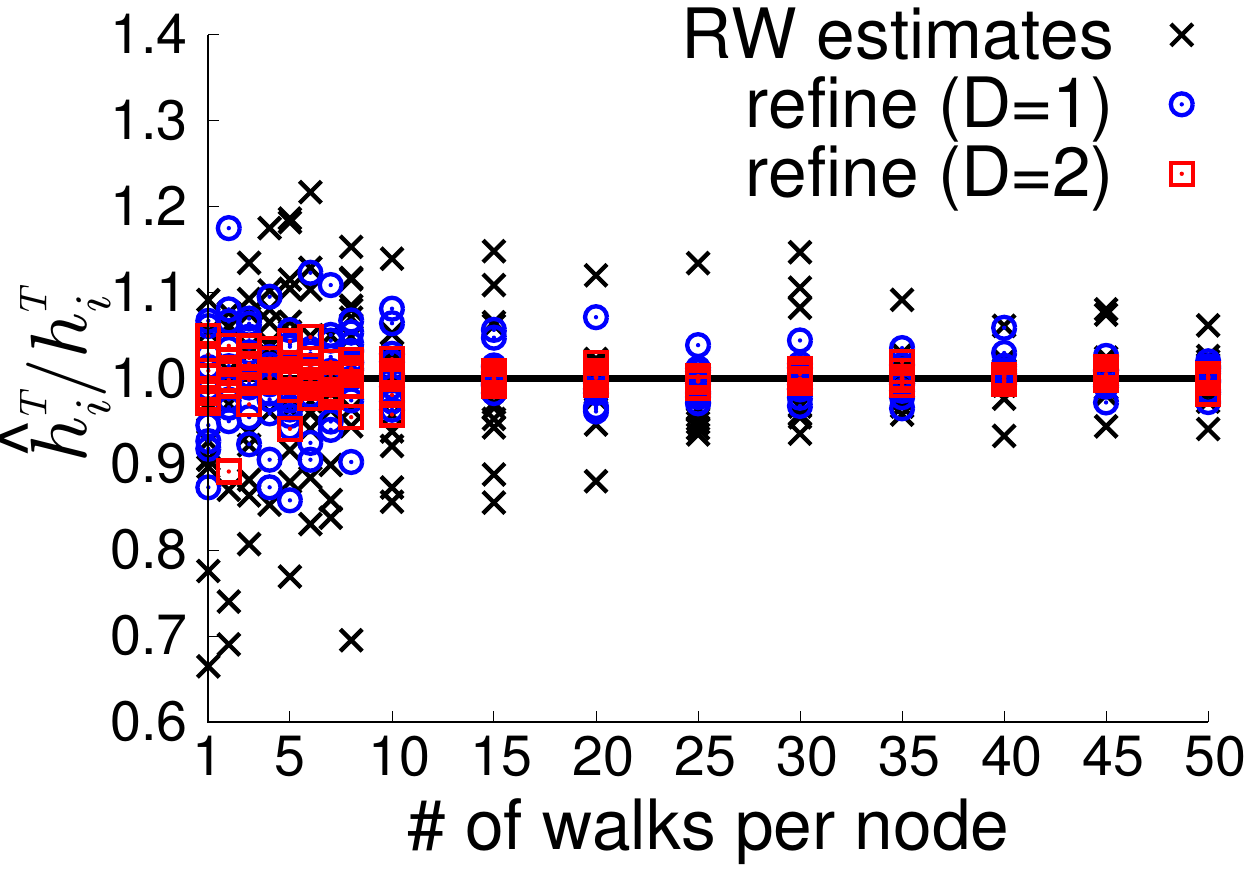}}
    } \\
    \multicolumn{3}{@{}c@{}}{
    \subfloat[$\hat{h}_i^T/h_i^T$ for high degree nodes $i$]{%
      \includegraphics[width=.33\linewidth]{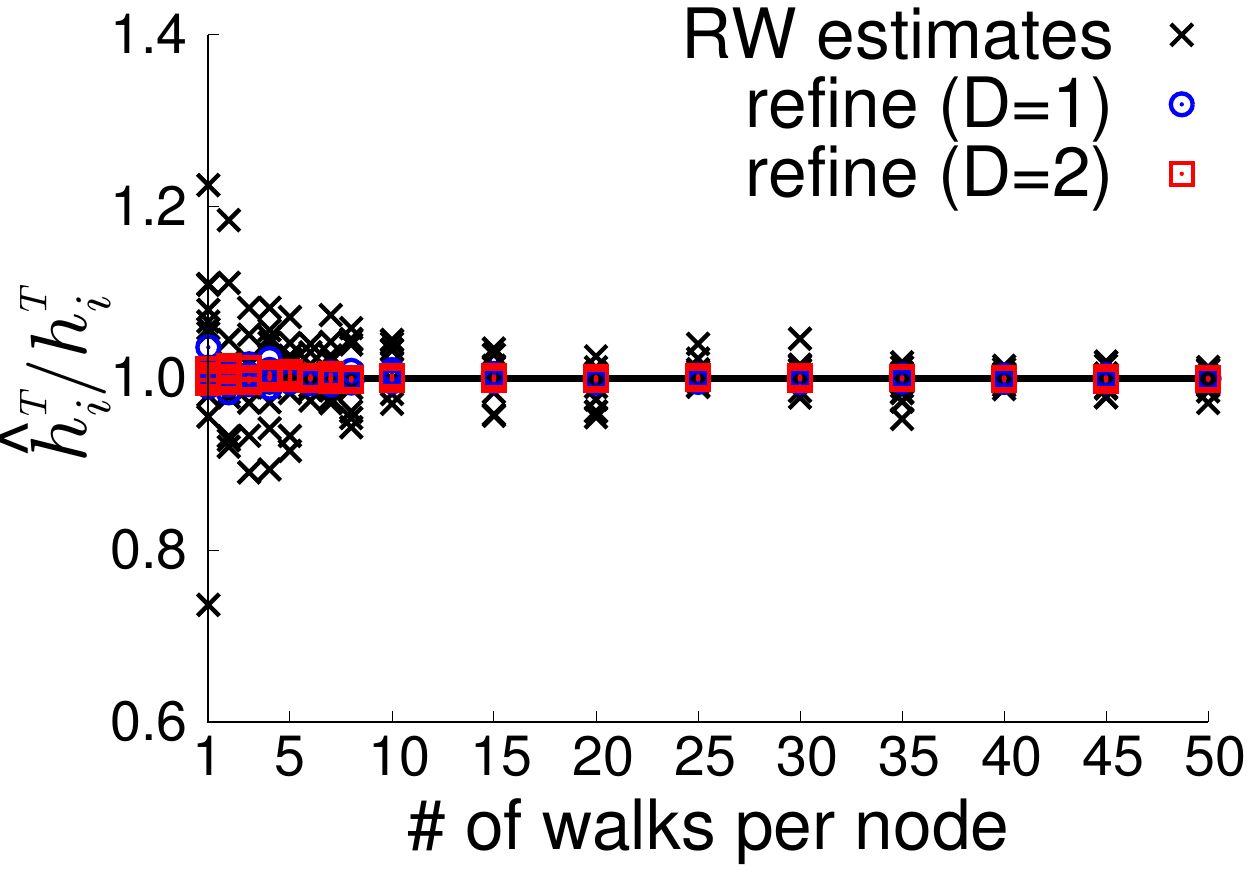}%
      \includegraphics[width=.33\linewidth]{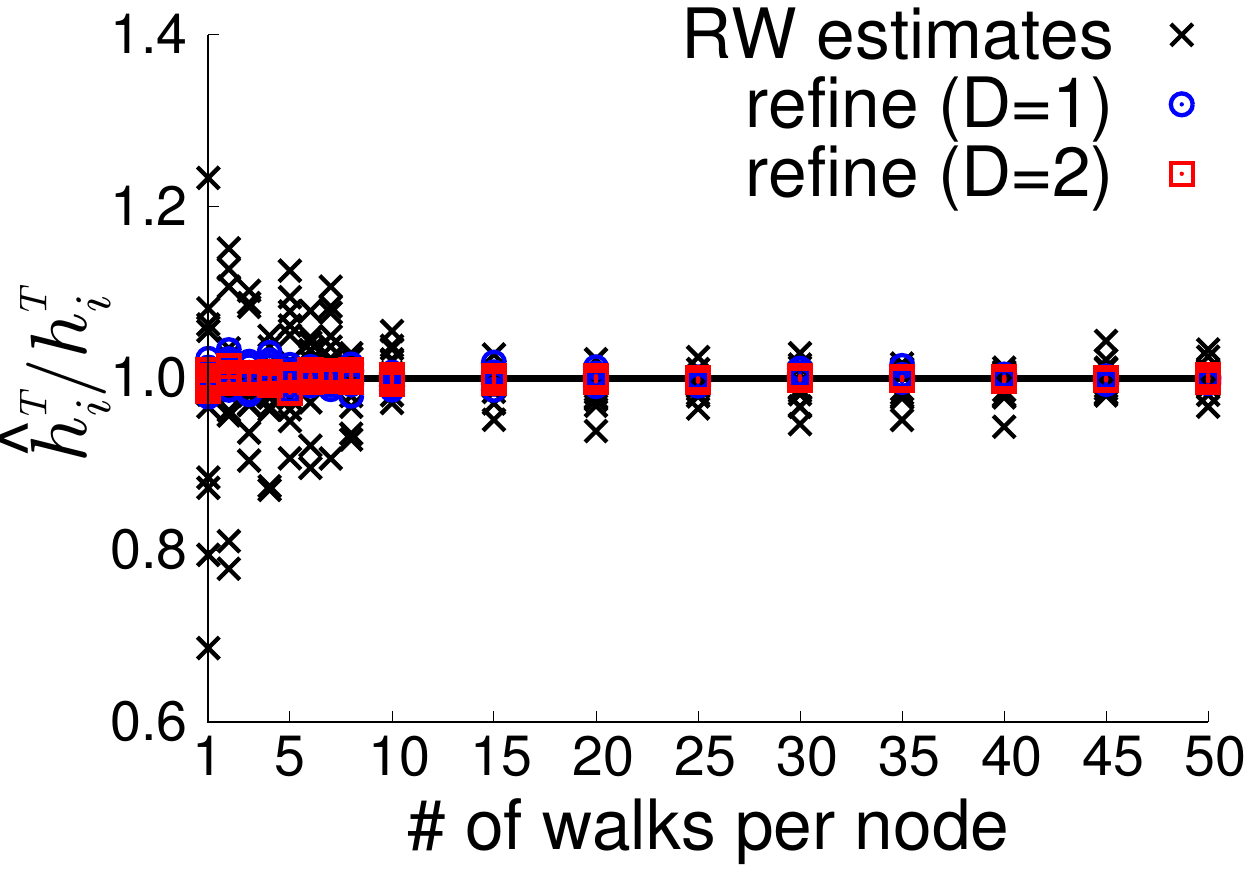}%
      \includegraphics[width=.33\linewidth]{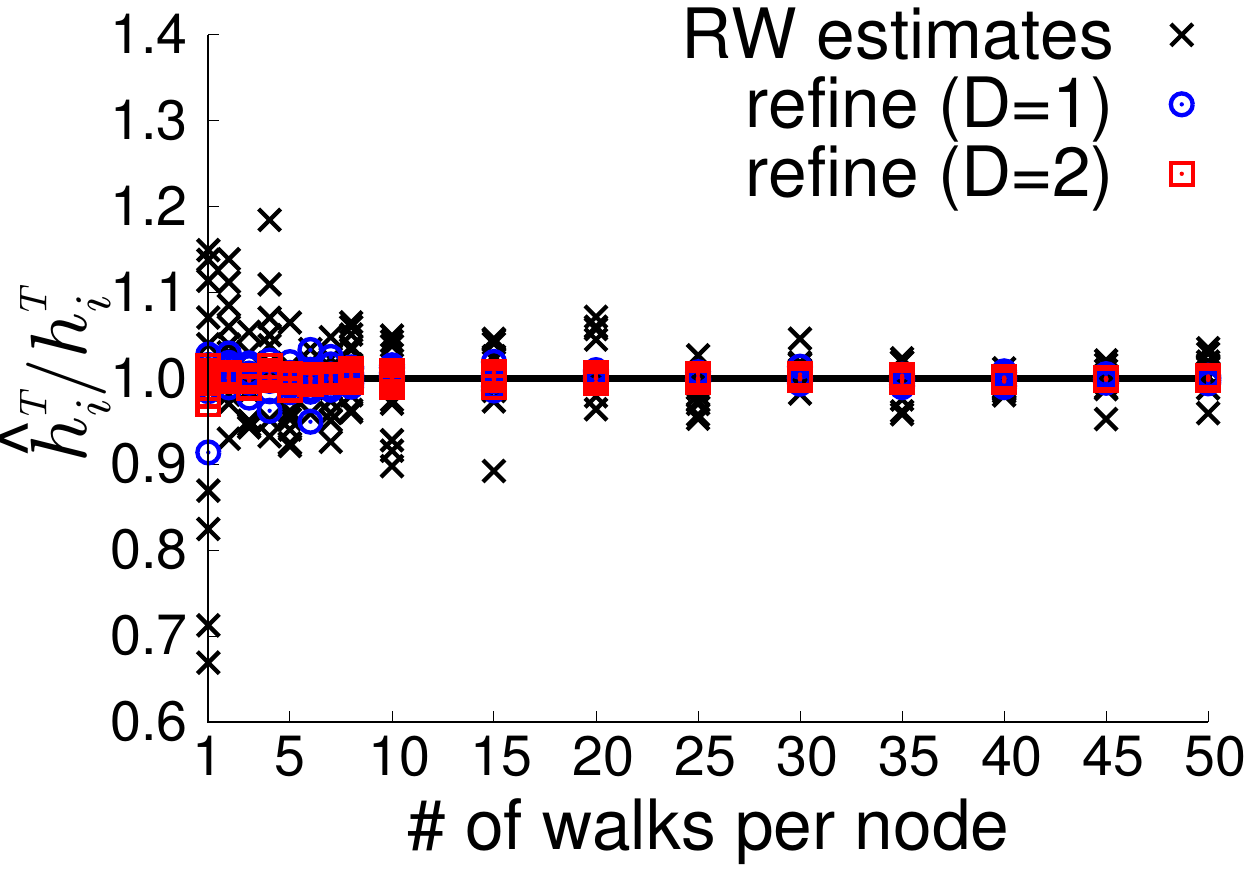}}
    }
  \end{tabular}
  \caption{Estimates of $h_i^T$ on three graphs. ($T=10$)}
  \label{fig:scatters_h}
\end{figure}

Another way to evaluate the estimation accuracy of an estimator is to study its
normalized rooted mean squared error (NRMSE).
NRMSE of an estimator $\hat\theta$ given groundtruth $\theta$ is defined by
NRMSE$(\hat\theta)\triangleq \sqrt{\E(\hat\theta - \theta)^2}/\theta$, and the
smaller the NRMSE, the more accurate an estimator is.
In our setting, we propose to quantify the estimation accuracy by the averaged
normalized rooted mean squared error (AVG-NRMSE), i.e.,
\begin{align*}
  \text{AVG-NRMSE}(\{\hat{p}_i^T\}_{i\in V'})
  &\triangleq \frac{1}{|V'|}\sum_{i\in V'}\text{NRMSE}(\hat{p}_i^T), \\
  \text{AVG-NRMSE}(\{\hat{h}_i^T\}_{i\in V'})
  &\triangleq \frac{1}{|V'|}\sum_{i\in V'}\text{NRMSE}(\hat{h}_i^T),
\end{align*}
where $V'\subseteq V$ is a subset of nodes to evaluate, and we set $V'=V$.
We depict these results in Figures~\ref{fig:fap_nrmse} and~\ref{fig:fht_nrmse}.
The NRMSE curves clearly show the difference of performance of the two methods and
with different parameter settings.
First, we observe that when the number of walks per node increases, the estimation
error of each method decreases, indicating that the estimates become more
accurate.
Second, the estimation-and-refinement approach can provide even more accurate
estimates than the RW estimation approach.
When the refinement depth $D$ increases, we could obtain even more accurate
estimates.
These observations are coincide with the previous experiment.

\begin{figure}[htp]
  \small
  \centering
  \begin{tabular}{CCC}
    \hline
    HepTh & Enron & Gowalla \\
    \hline
    \multicolumn{3}{@{}c@{}}{
    \subfloat[estimation error vs. $R$ ($T=10$)]{%
      \includegraphics[width=.33\linewidth]{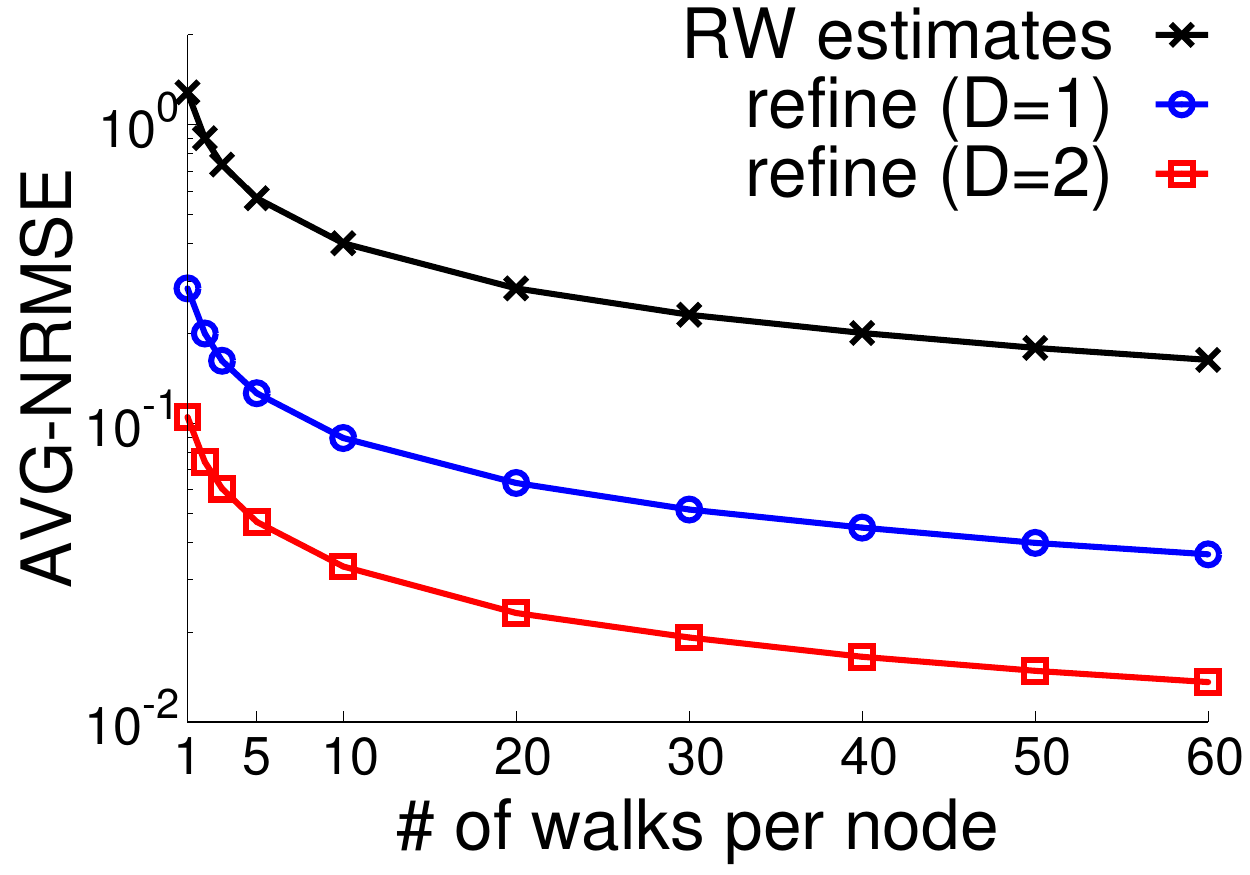}%
      \includegraphics[width=.33\linewidth]{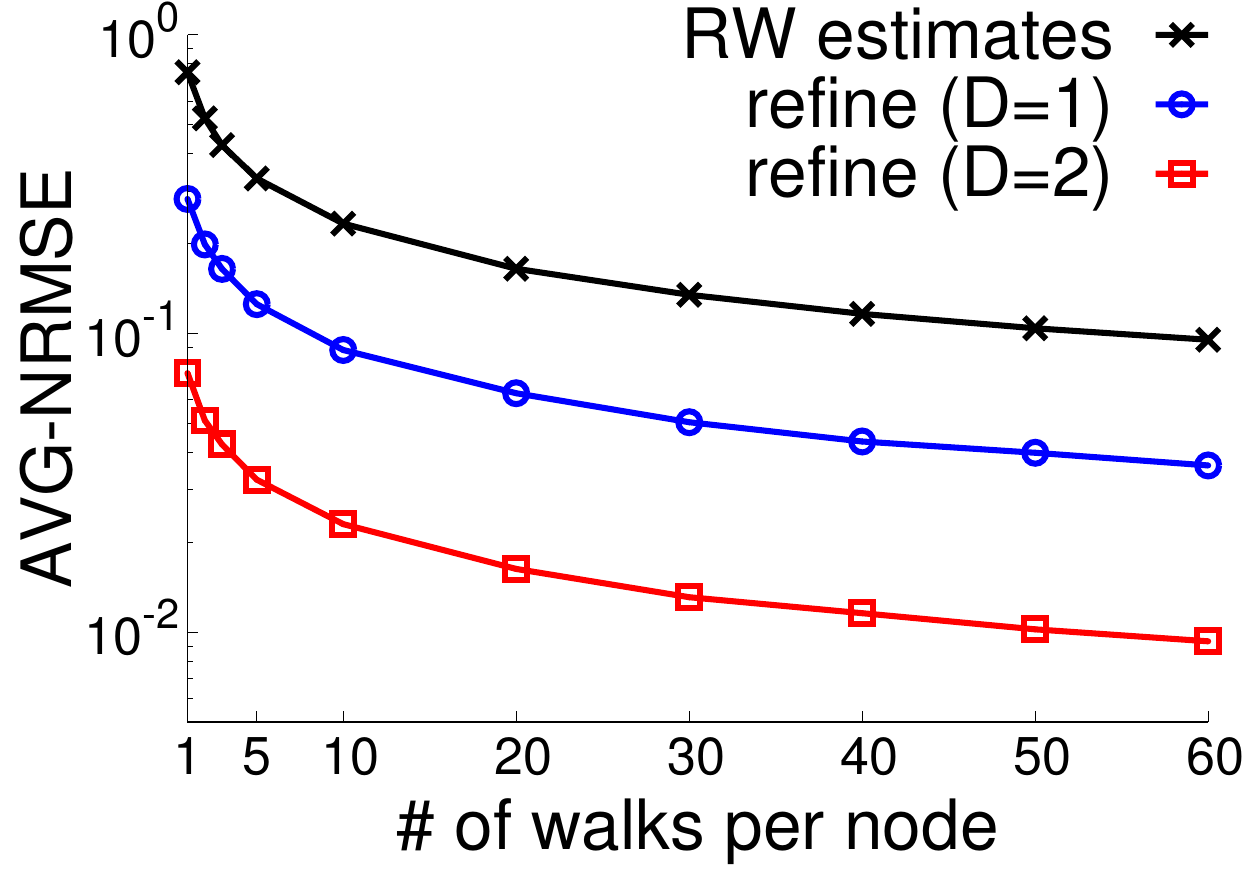}%
      \includegraphics[width=.33\linewidth]{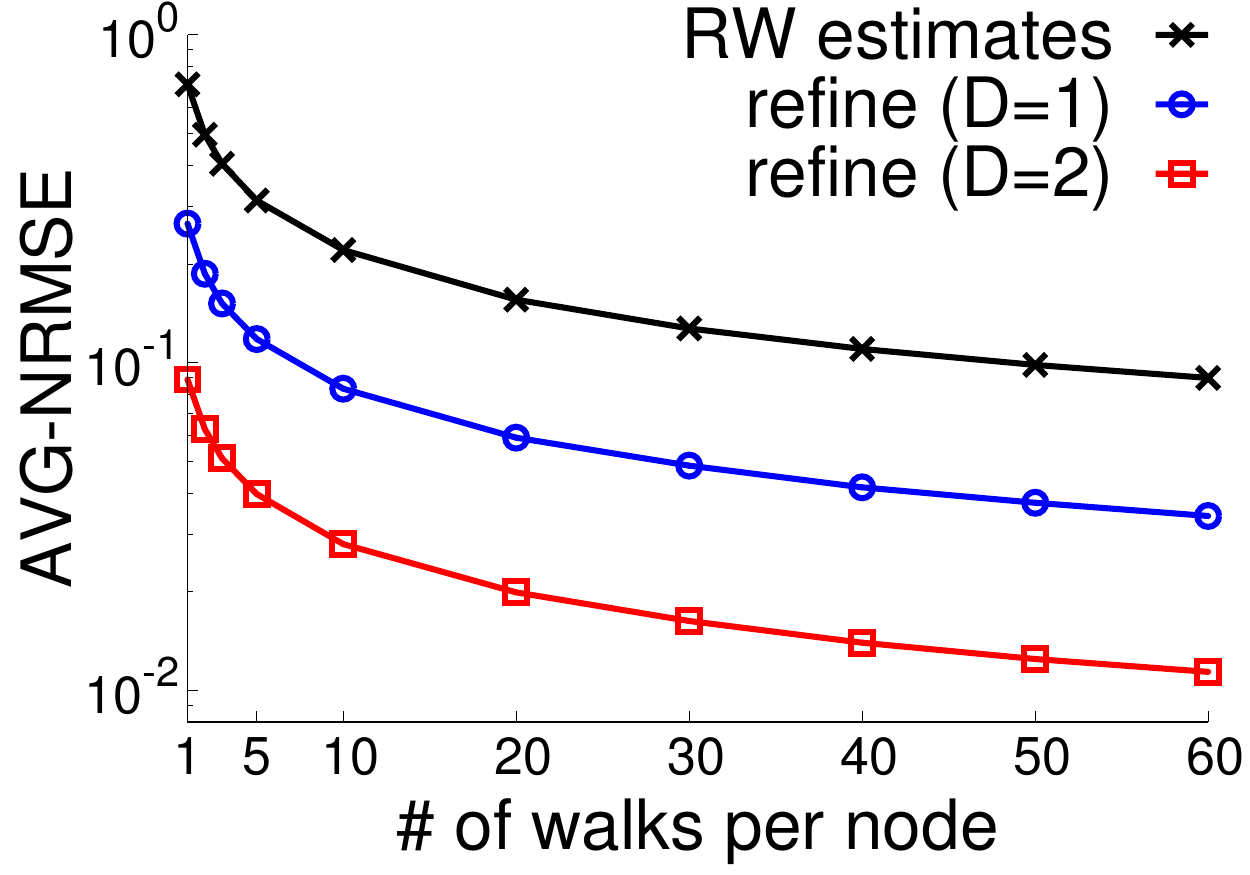}}
    } \\
    \multicolumn{3}{@{}c@{}}{
    \subfloat[estimation error vs. $T$ ($R=10$)\label{sf:nrmse_p_T}]{%
      \includegraphics[width=.33\linewidth]{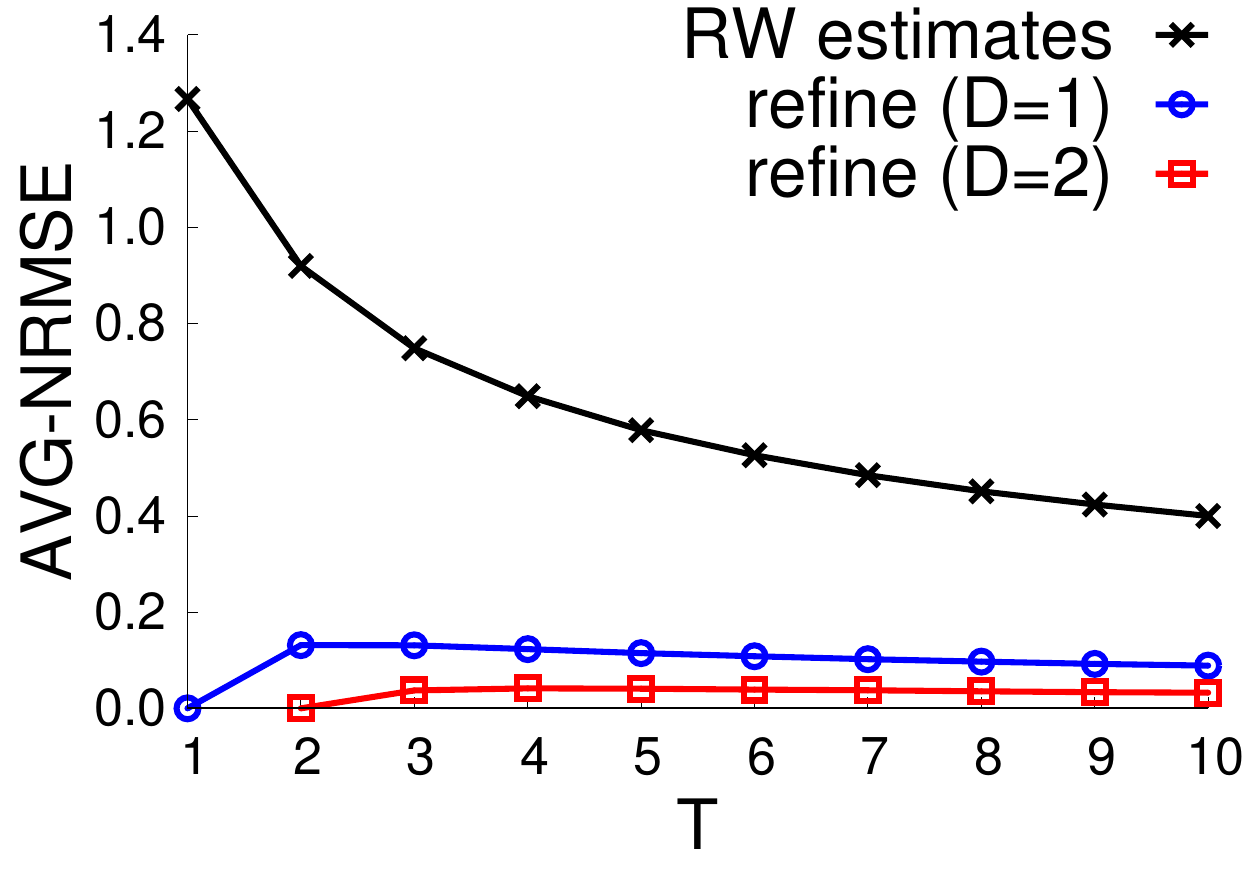}%
      \includegraphics[width=.33\linewidth]{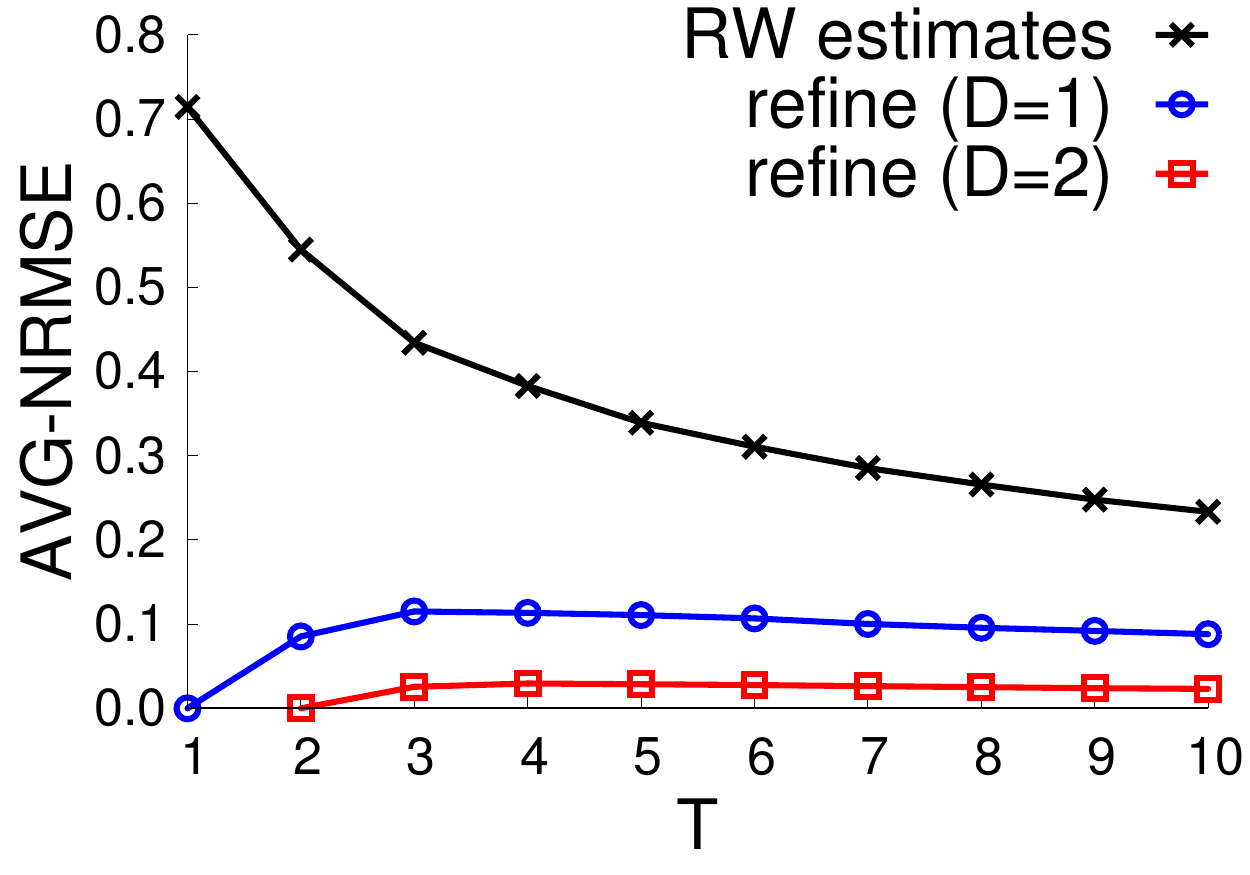}%
      \includegraphics[width=.33\linewidth]{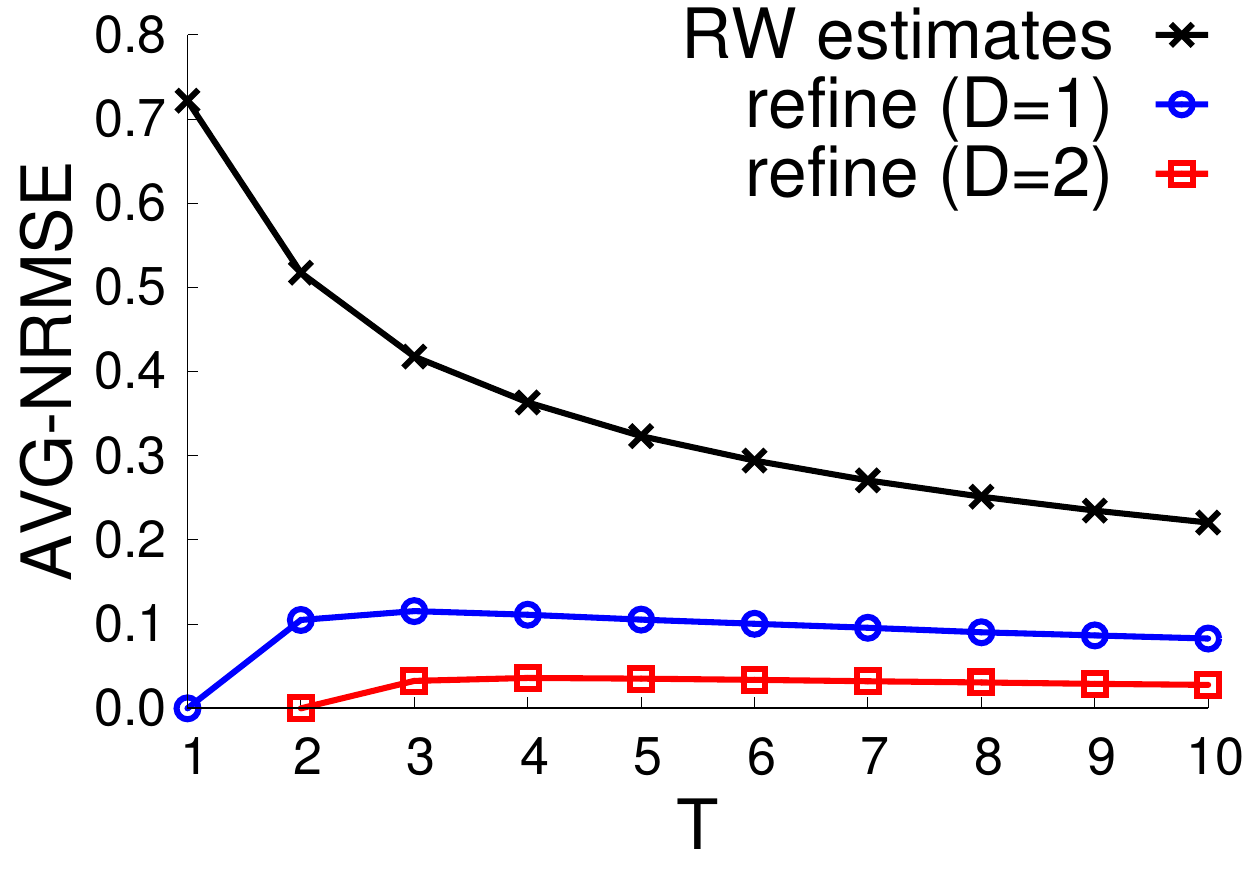}}
    }
  \end{tabular}
  \caption{$p_i^T$ estimation accuracy on three networks.}
    \label{fig:fap_nrmse}
\end{figure}

\begin{figure}[htp]
  \small
  \centering
  \begin{tabular}{CCC}
    \hline
    HepTh & Enron & Gowalla \\
    \hline
    \multicolumn{3}{@{}c@{}}{
    \subfloat[estimation error vs. $R$ ($T=10$)]{%
      \includegraphics[width=.33\linewidth]{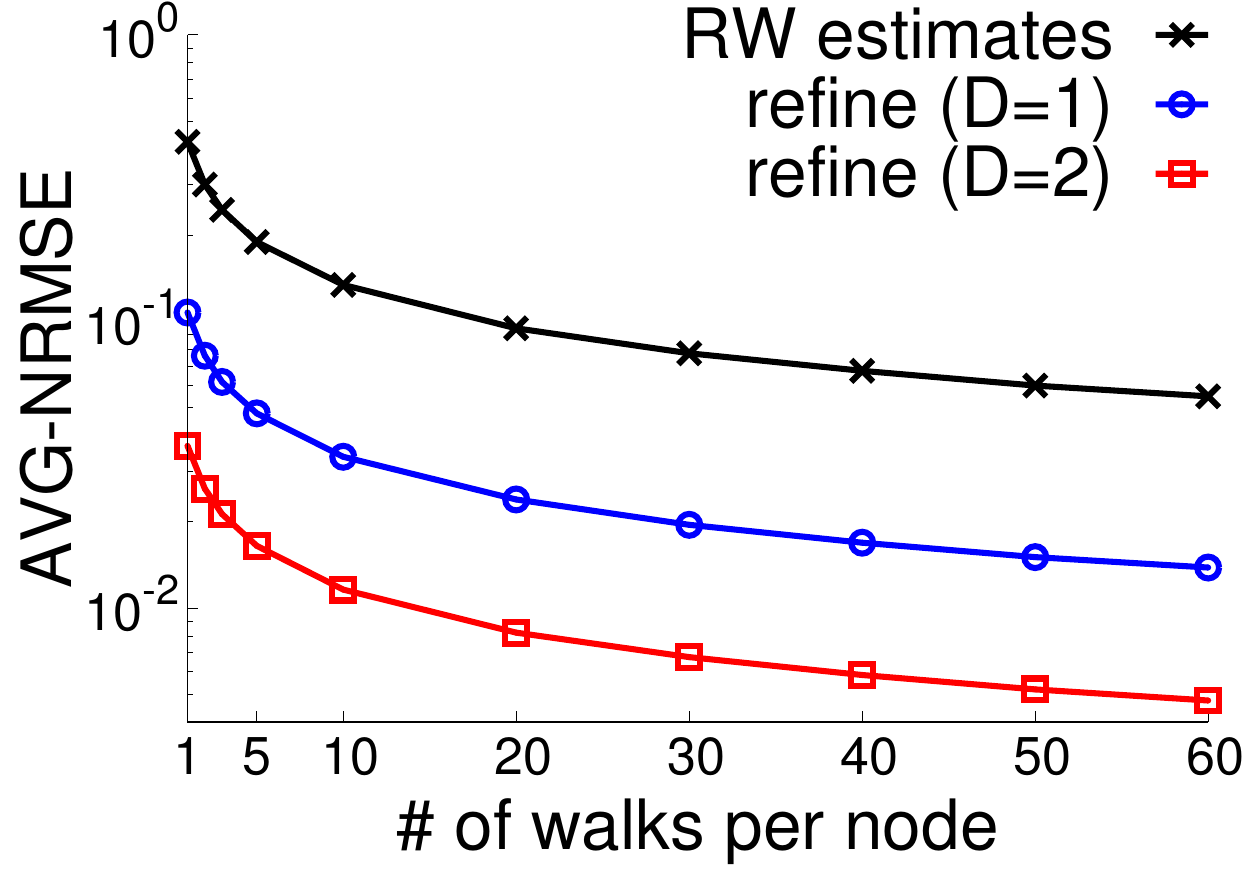}%
      \includegraphics[width=.33\linewidth]{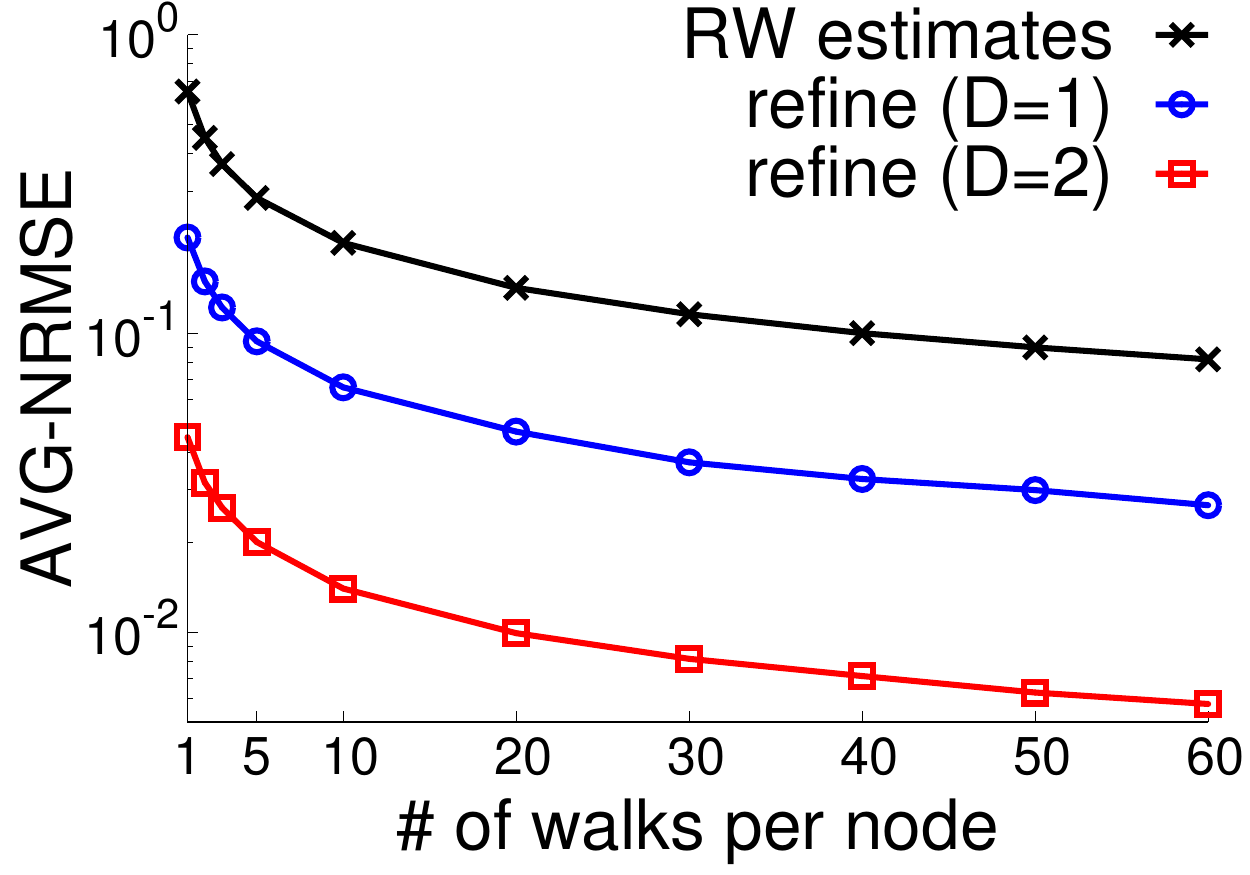}%
      \includegraphics[width=.33\linewidth]{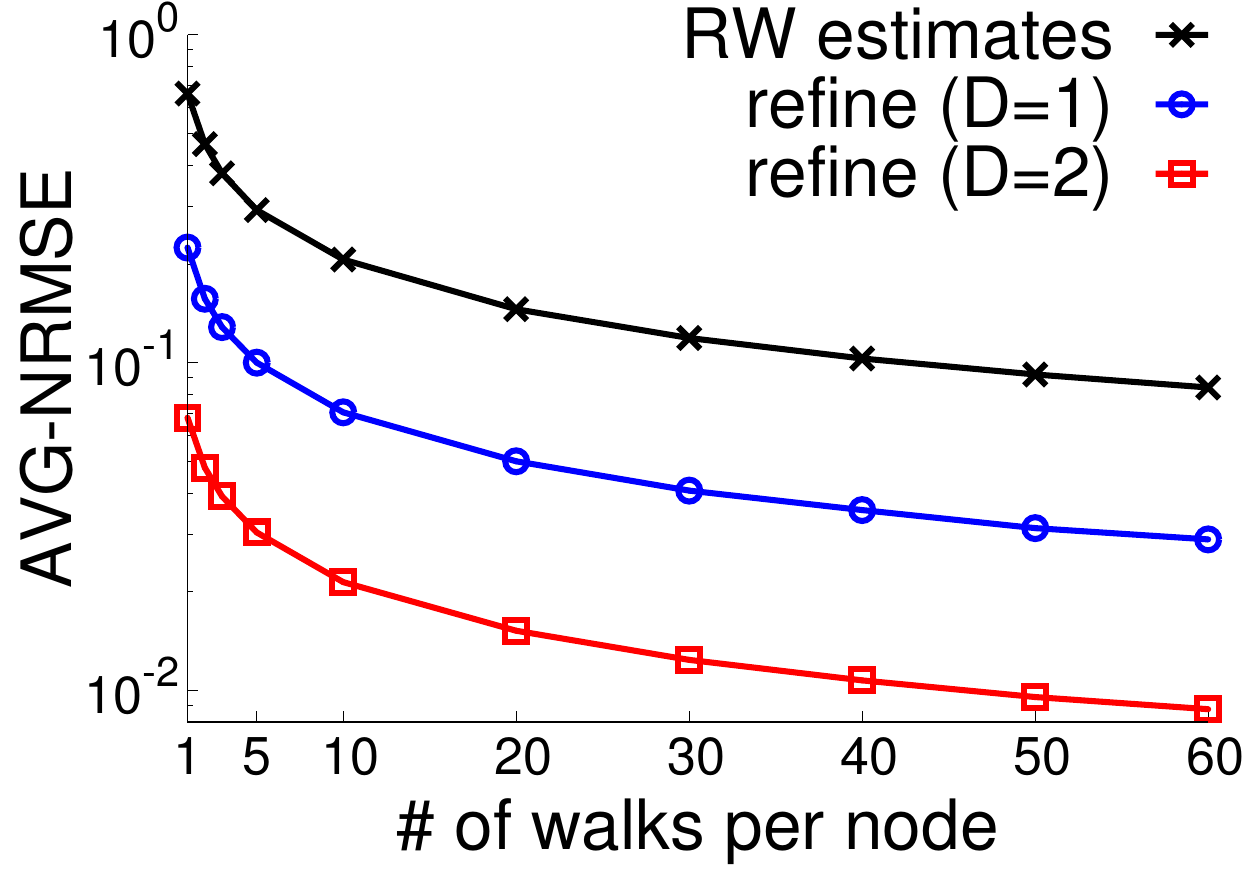}}
    } \\
    \multicolumn{3}{@{}c@{}}{
    \subfloat[estimation error vs. $T$ ($R=10$)\label{sf:nrmse_h_T}]{%
      \includegraphics[width=.33\linewidth]{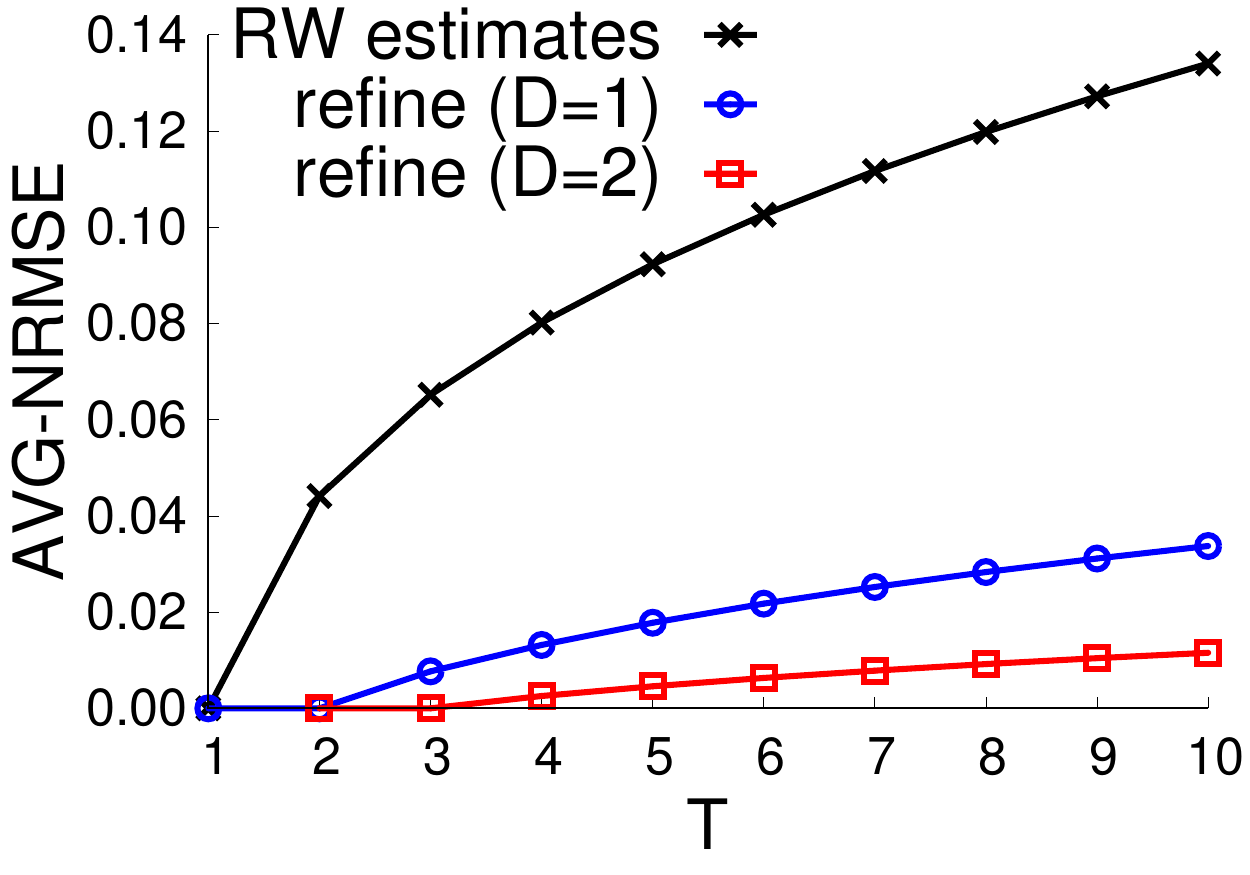}%
      \includegraphics[width=.33\linewidth]{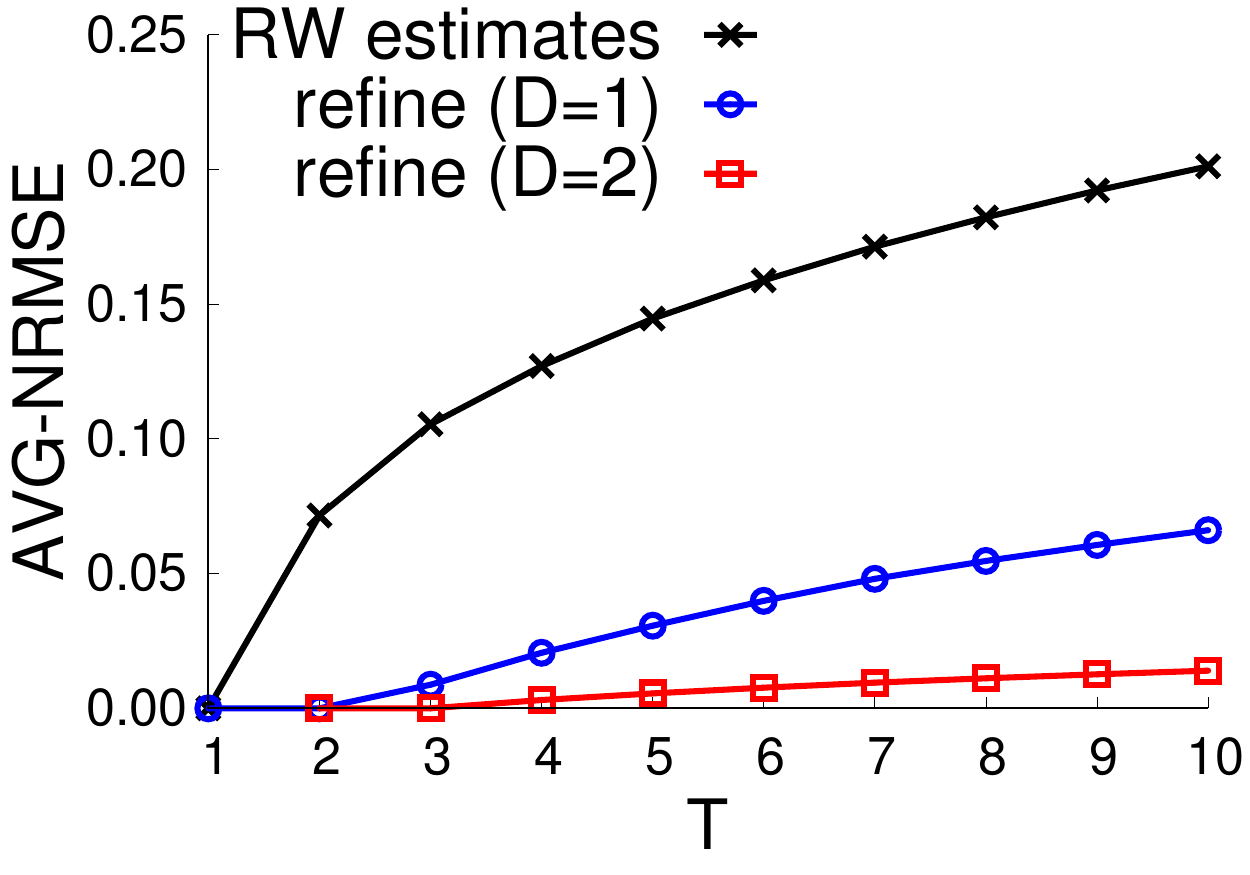}%
      \includegraphics[width=.33\linewidth]{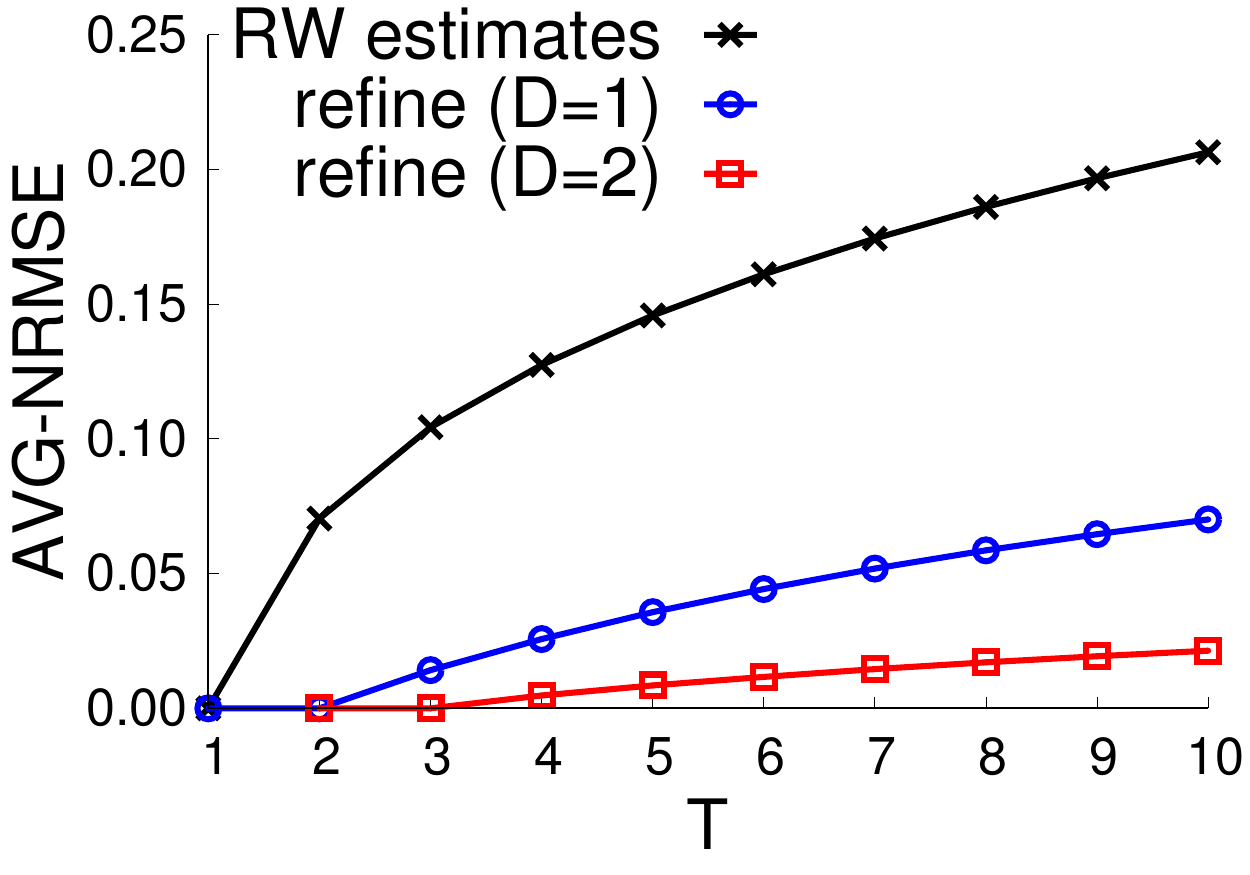}}
    }
  \end{tabular}
  \caption{$h_i^T$ estimation accuracy on three networks.}
    \label{fig:fht_nrmse}
\end{figure}

We also study how random walk length $T$ affects the estimation accuracy.
From Figures~\ref{sf:nrmse_p_T} and~\ref{sf:nrmse_h_T}, we observe that, using the
same amount of RWs, e.g., $R=10$, when $T$ increases, it actually becomes easier
to estimate $p_i^T$ as NRMSE decreases, and more difficult to estimate $h_i^T$ as
NRMSE increases.
For both cases, the estimation-and-refinement approach can obtain smaller NRMSE,
and when refinement depth $D$ increases, the NRMSE further decreases.
In conclusion, these results demonstrate that the estimation-and-refinement
approach can provide more accurate estimates than the RW estimation approach.

\subsection{Evaluating Oracle Call Accuracy and Efficiency}

In the second experiment, we evaluate the oracle call accuracy and efficiency
implemented by different methods.
Because we cannot afford to calculate the groundtruth of marginal gain for each
node, we randomly sample $100$ nodes from each graph, and calculate their marginal
gain groundtruth using DP with $S=\emptyset$.
Here, oracle call accuracy is measured by AVG-NRMSE, and oracle call efficiency is
measured by speedup, which is defined by
\[
  \text{speedup of a method}\triangleq
  \frac{\text{time cost of DP}}{\text{time cost of the method}}.
\]
The results of NRMSE and speedup of different methods on three graphs HepTh,
Enron, and gowalla, are depicted in Figures~\ref{fig:gain_AP}
and~\ref{fig:gain_HT}.

\begin{figure}[htp]
  \small
  \begin{tabular}{CCC}
    \hline
    HepTh & Enron & Gowalla \\
    \hline
    \multicolumn{3}{@{}c@{}}{%
    \subfloat[NRMSE]{%
      \includegraphics[width=.33\linewidth]{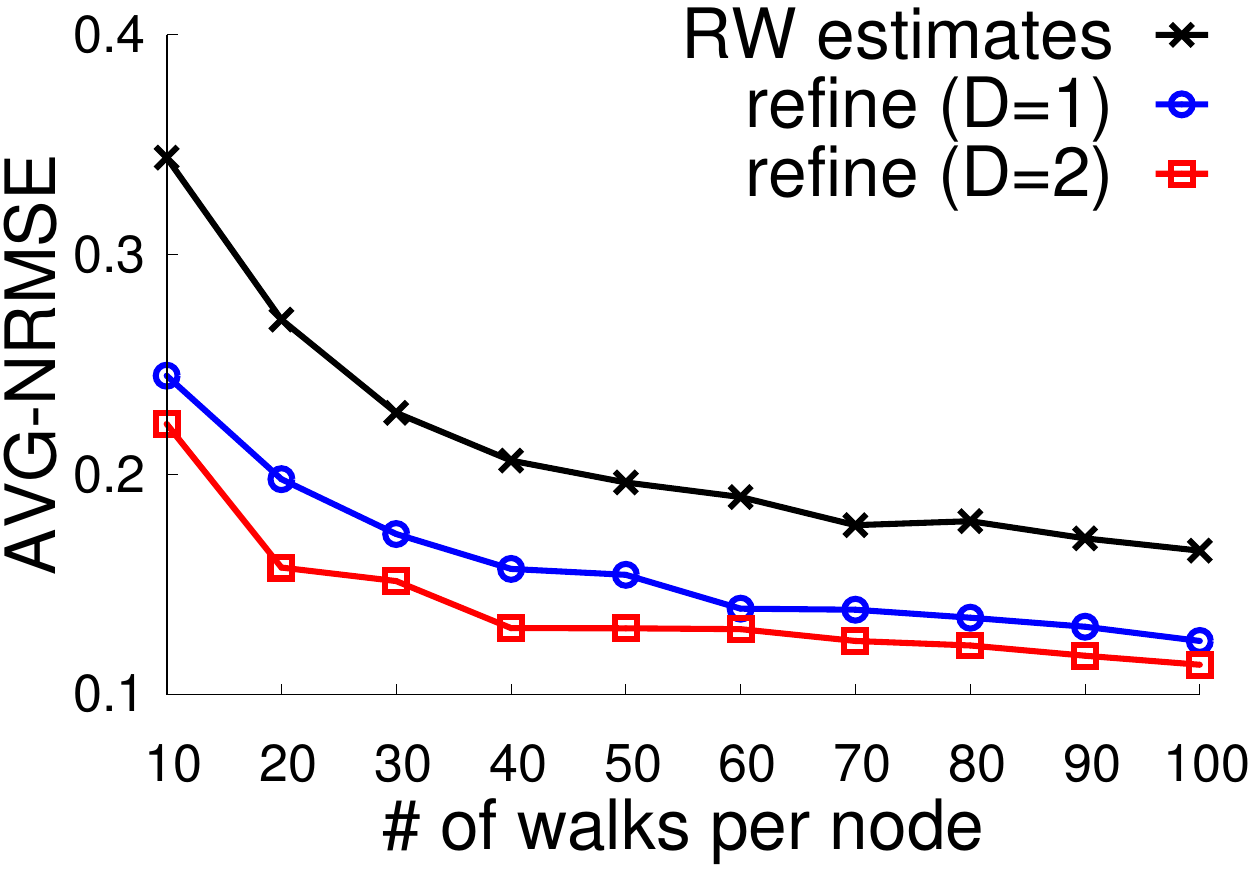}%
      \includegraphics[width=.33\linewidth]{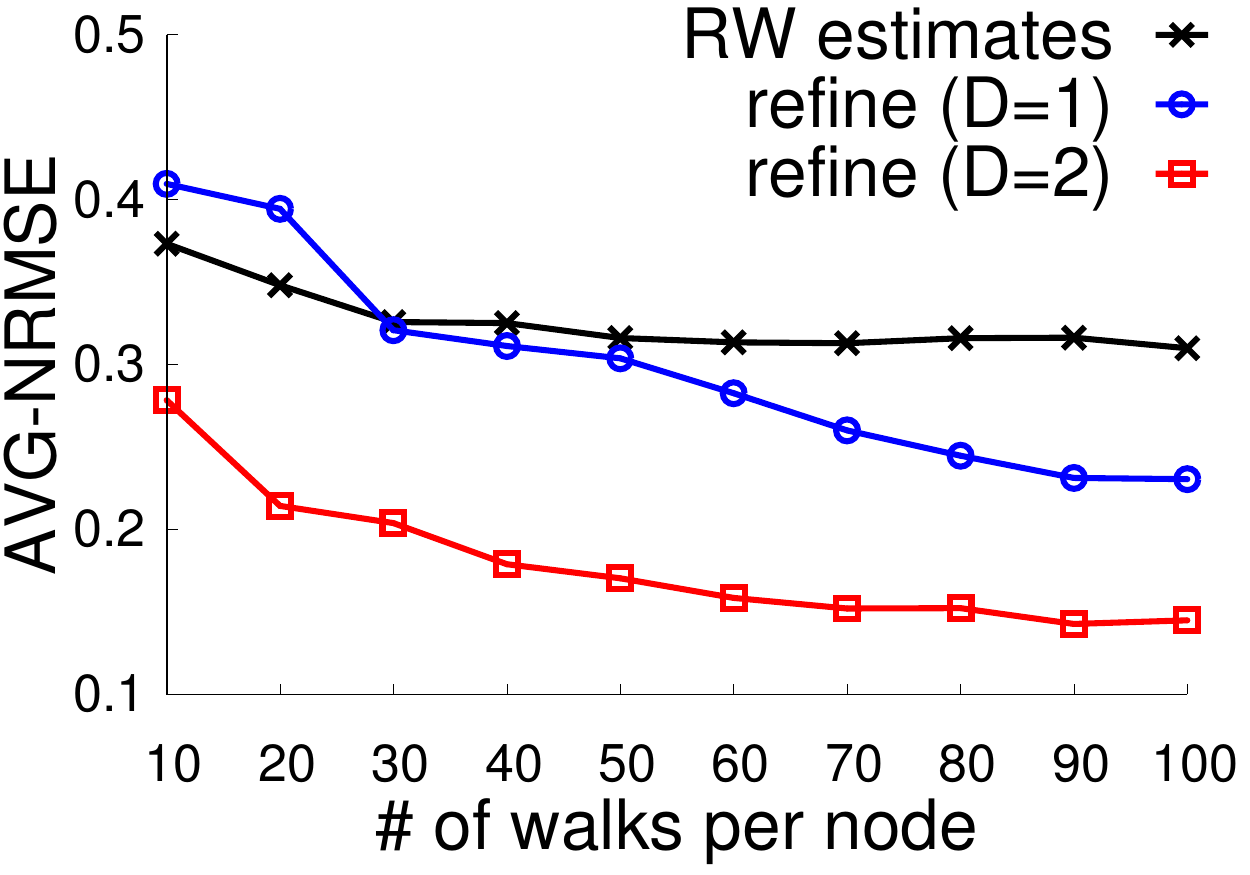}%
      \includegraphics[width=.33\linewidth]{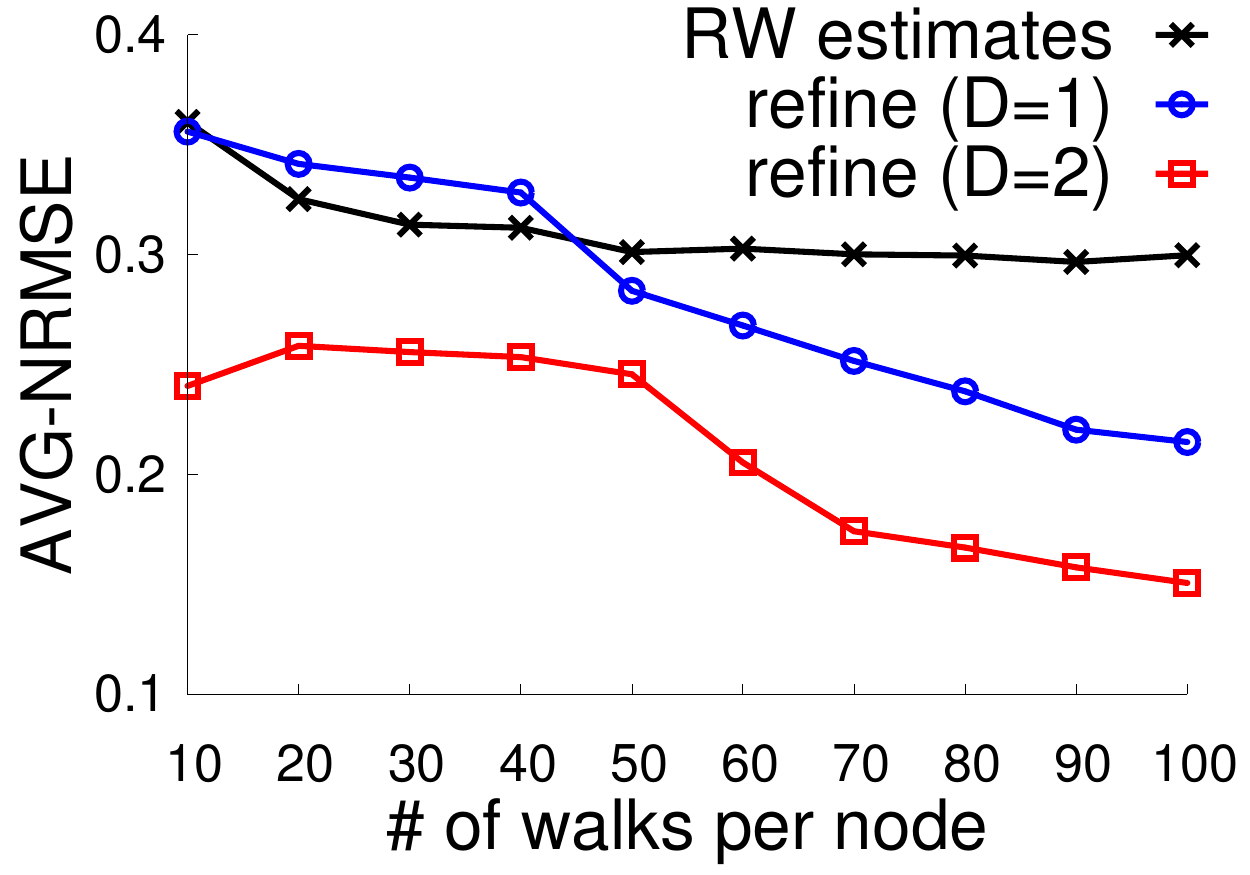}}
    } \\
    \multicolumn{3}{@{}c@{}}{
    \subfloat[Speedup]{%
      \includegraphics[width=.33\linewidth]{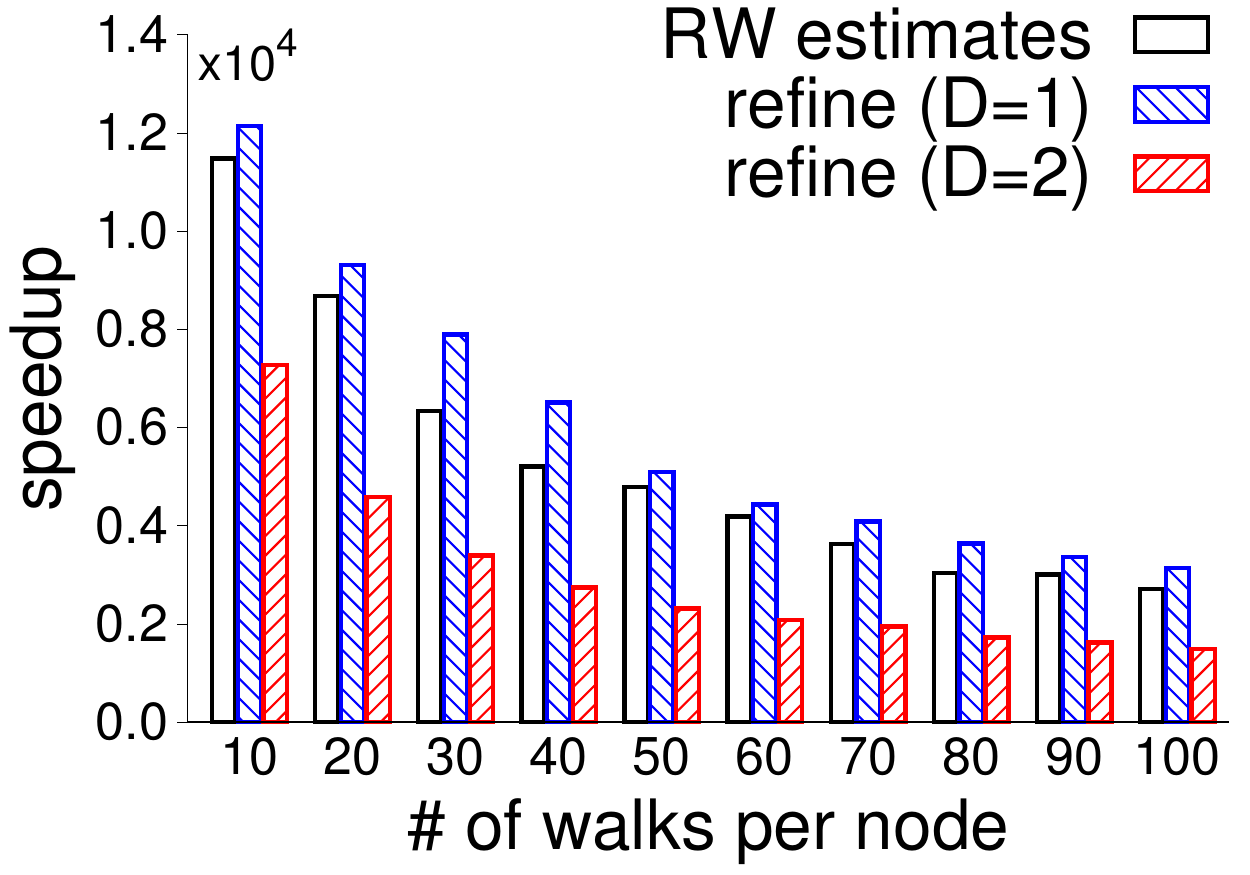}%
      \includegraphics[width=.33\linewidth]{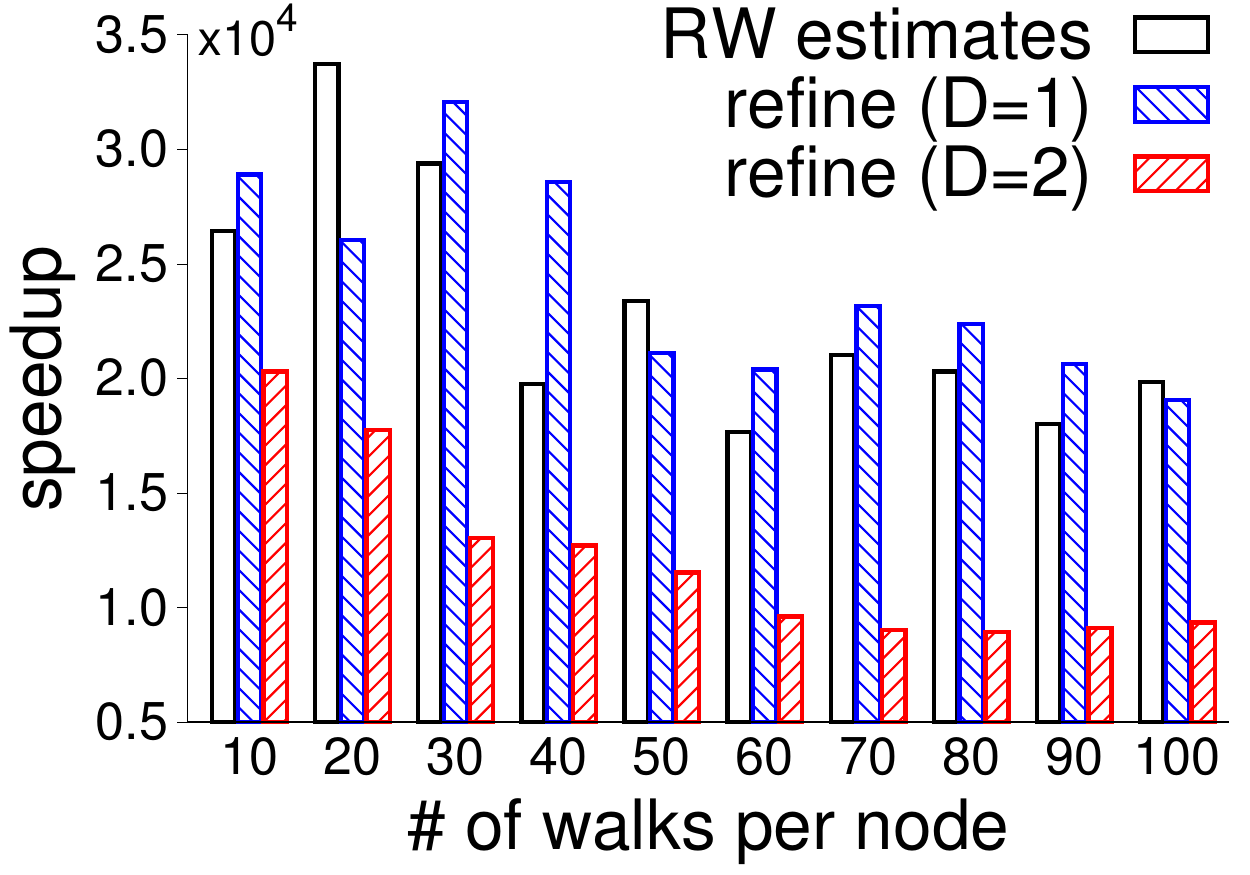}%
      \includegraphics[width=.33\linewidth]{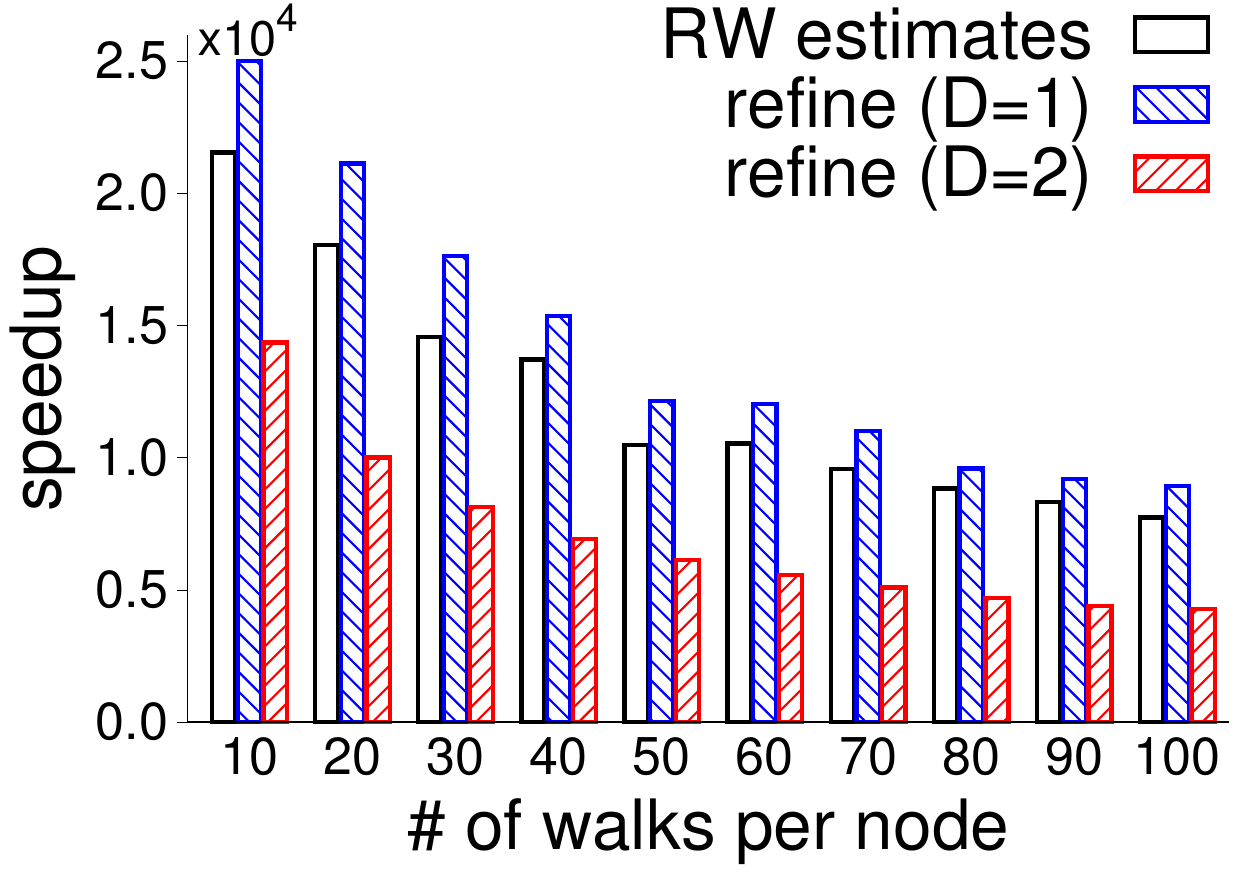}}
    }
  \end{tabular}
  \caption{Absorbing probability oracle call accuracy and efficiency ($T=10$)}
    \label{fig:gain_AP}
\end{figure}

\begin{figure}[htp]
  \small
  \begin{tabular}{CCC}
    \hline
    HepTh & Enron & Gowalla \\
    \hline
    \multicolumn{3}{@{}c@{}}{%
    \subfloat[NRMSE]{%
      \includegraphics[width=.33\linewidth]{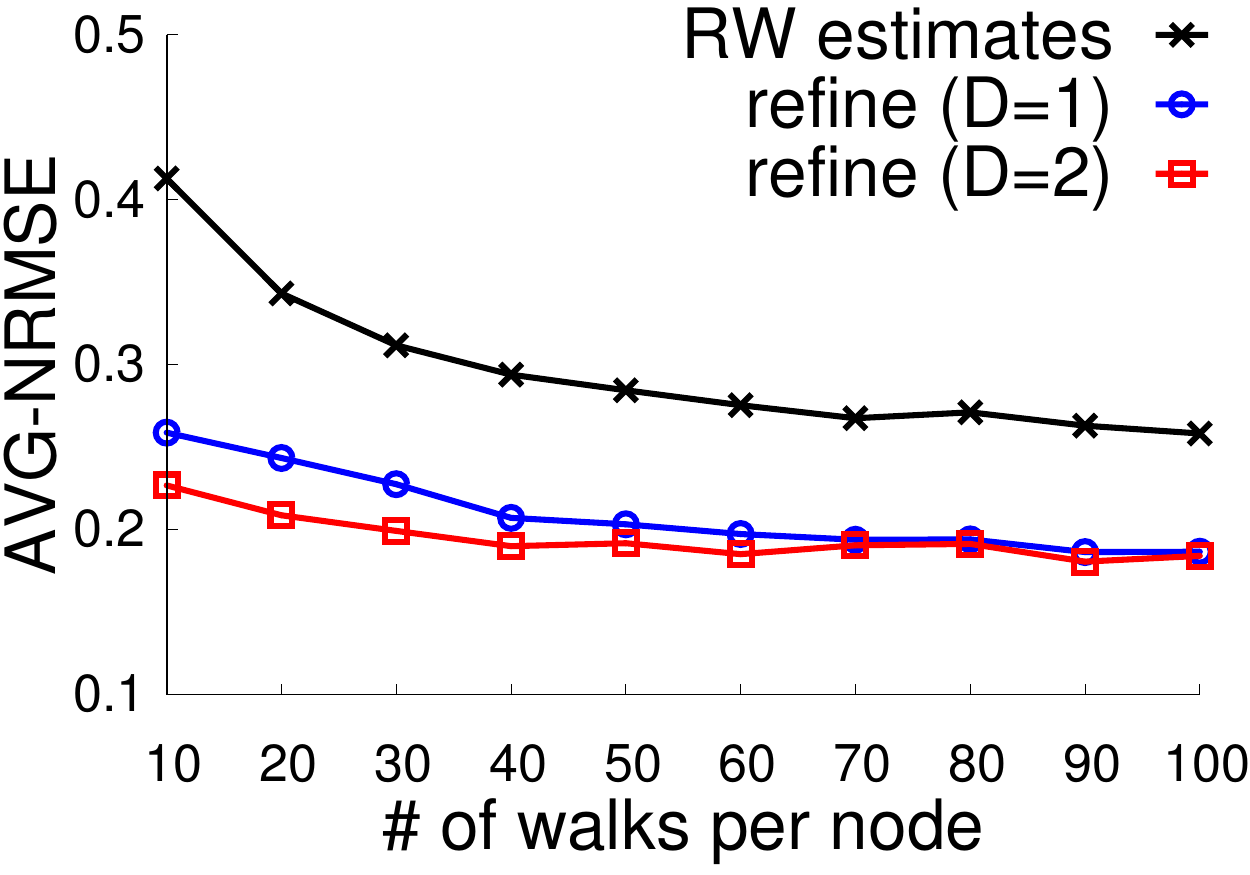}%
      \includegraphics[width=.33\linewidth]{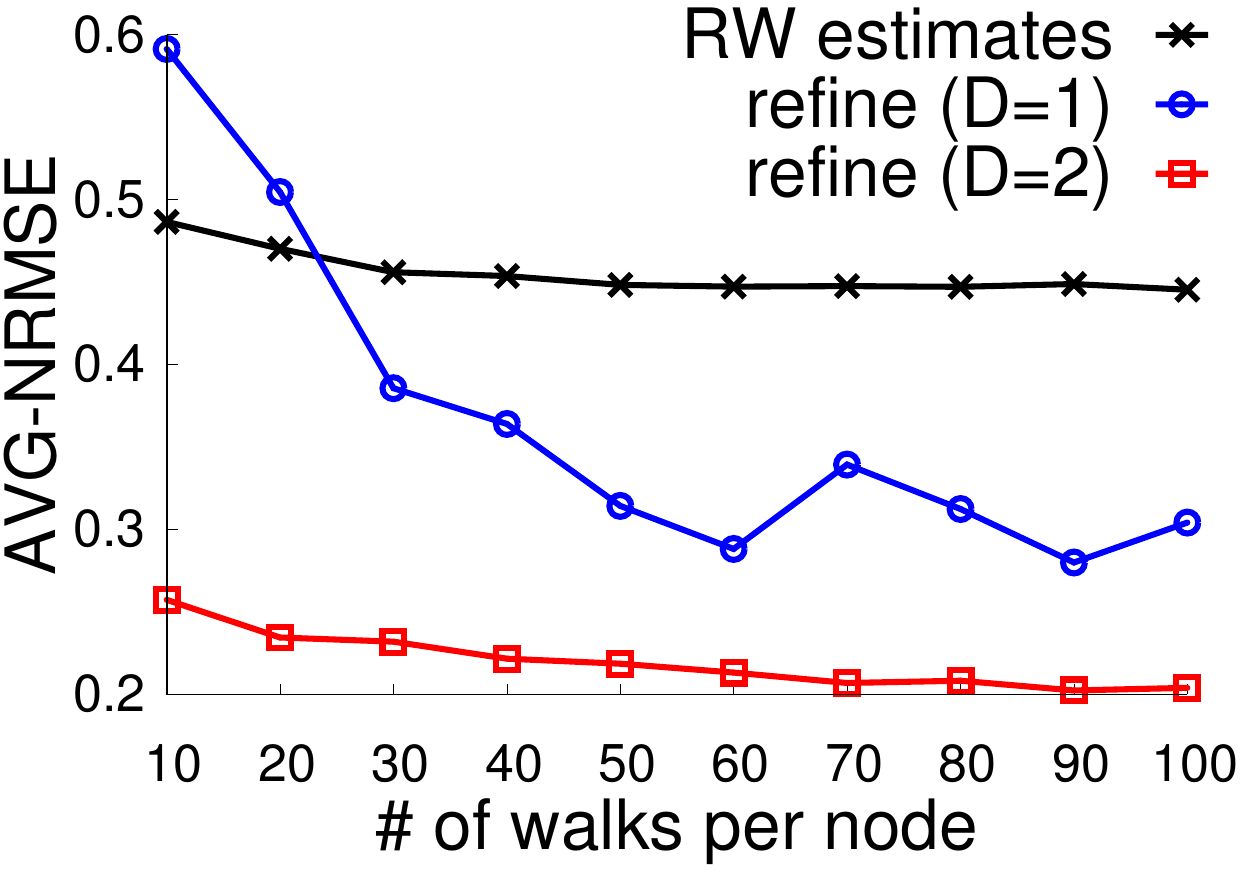}%
      \includegraphics[width=.33\linewidth]{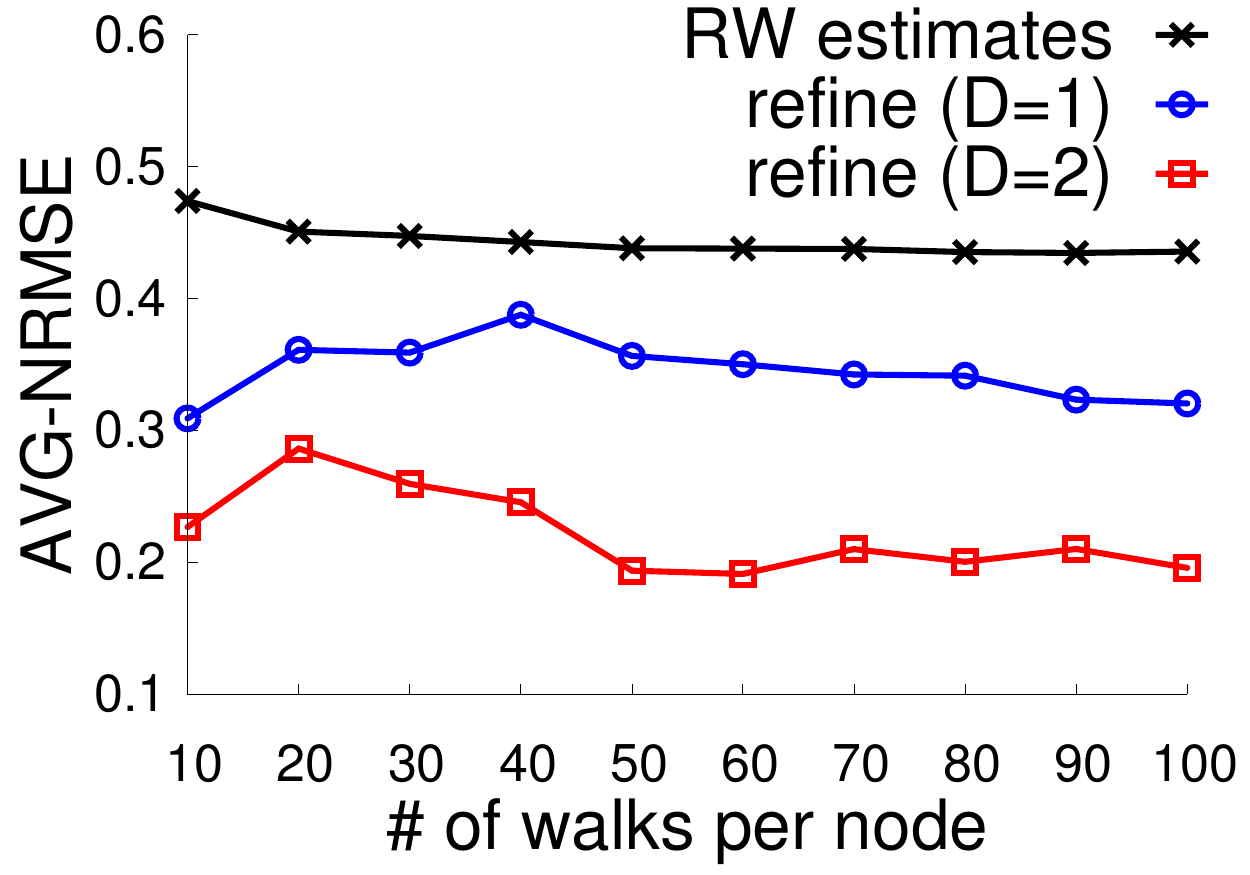}}
    } \\
    \multicolumn{3}{@{}c@{}}{
    \subfloat[Speedup]{%
      \includegraphics[width=.33\linewidth]{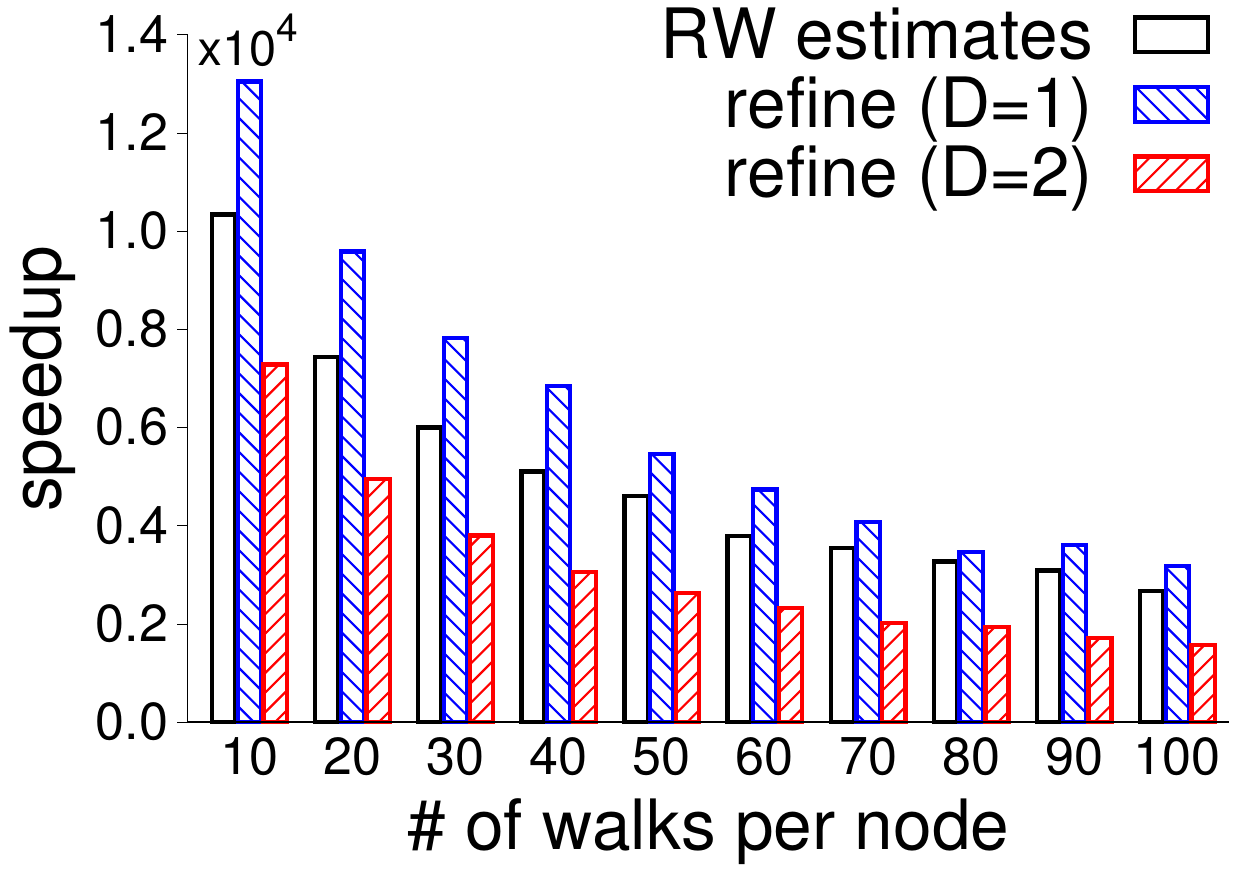}%
      \includegraphics[width=.33\linewidth]{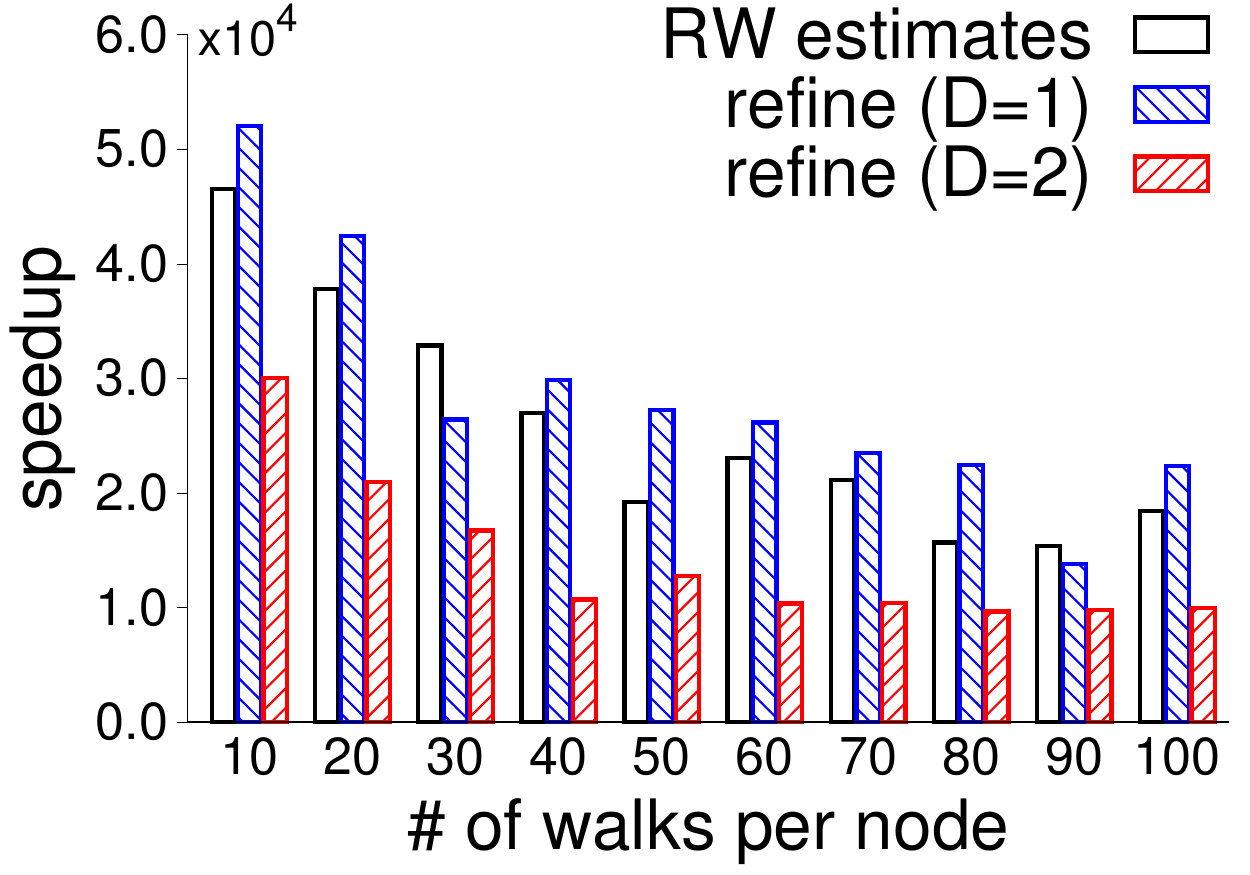}%
      \includegraphics[width=.33\linewidth]{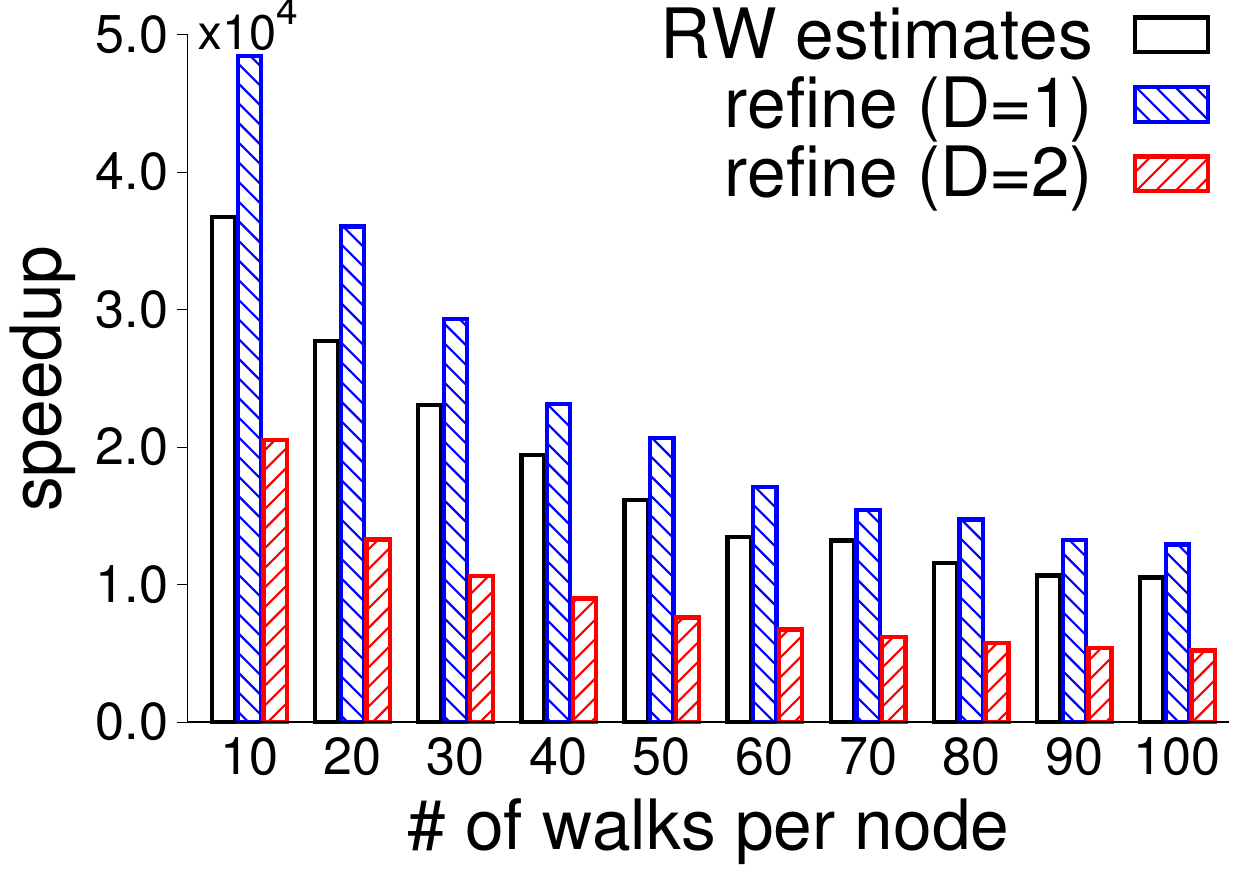}}
    }
  \end{tabular}
  \caption{Hitting time oracle call accuracy and efficiency ($T=10$)}
    \label{fig:gain_HT}
\end{figure}

From the NRMSE curves, we observe similar results as in the previous experiment:
In general, (1) when the number of walks per node increases, every method obtains
more accurate estimates; (2) the estimation-and-refinement approach can obtain
more accurate estimates than the RW estimation approach, and the estimation
accuracy improves when refinement depth $D$ increases.
Note that we also observe some exceptions, e.g., on some graphs, the
estimation-and-refinement method with $D=1$ exhibits larger NRMSE, however, for
$D\geq 2$ or with larger number of walks, the estimation-and-refinement approach
is significantly more accurate than the RW approach.

From the speedup curves, we can observe that both the RW estimation approach and
the estimation-and-refinement approach are significantly more efficient than DP.
On average, the two estimation approaches are at least thousands of times faster
than DP.
We also observe something interesting: When we increase the refinement depth, the
oracle call efficiency decreases in general, as expected; however, we observe that
the estimation-and-refinement approach with $D=1$ is actually more efficient than
the RW estimation approach.
This is because that when we use the estimation-and-refinement approach, we
simulate shorter walks, and this could slightly improve the oracle call
efficiency.
As we further increase refinement depth to $D=2$, because we need to explore a
large part of a node's neighborhood, the estimation-and-refinement approach
becomes slower than the RW estimation method.

\subsection{Comparing Greedy Algorithm with Baseline Methods}
\label{ss:greedy}

Equipped with the verified oracle call implementations, we are now ready to solve
the node discoverability optimization problem using the greedy algorithm.
In the third experiment, we run the greedy algorithm on each graphs, and choose a
subset of connection sources $S$ to optimize the target node's discoverability,
i.e., maximizing D-AP, and minimizing D-HT.
For each graph, we simulate $100$ walks from each node, and we use the
estimation-and-refinement approach with $D=2$ to implement the oracle call.
We set edge weight $w_{sn}=10$ if node $s$ is chosen to connect to target node $n$.
We also set $c_s\equiv 1$.
To better understand the performance of the greedy algorithm, we compare the
results with two baseline methods:
\begin{itemize}
\item \textbf{Random}: randomly pick nodes from the graph as connection sources;
\item \textbf{Degree}: always choose the top-$K$ largest degree nodes from the
  graph as connection sources.
\end{itemize}
The random approach is expected to have the poorest performance, and the
performance improvement of a method against the random approach reflects the
advantage of the method.
The performance of the degree approach is not clear.
One may think that nodes with large degrees represent high discoverability nodes
of a network, and connecting to high discoverability nodes could improve the
discoverability of target node.
We will study its performance through experiments.
The results are depicted in Figures~\ref{fig:greedy_ap} and~\ref{fig:greedy_ht}.

\begin{figure}[htp]
  \centering
  \subfloat[HepTh]{\includegraphics[width=.333\linewidth]{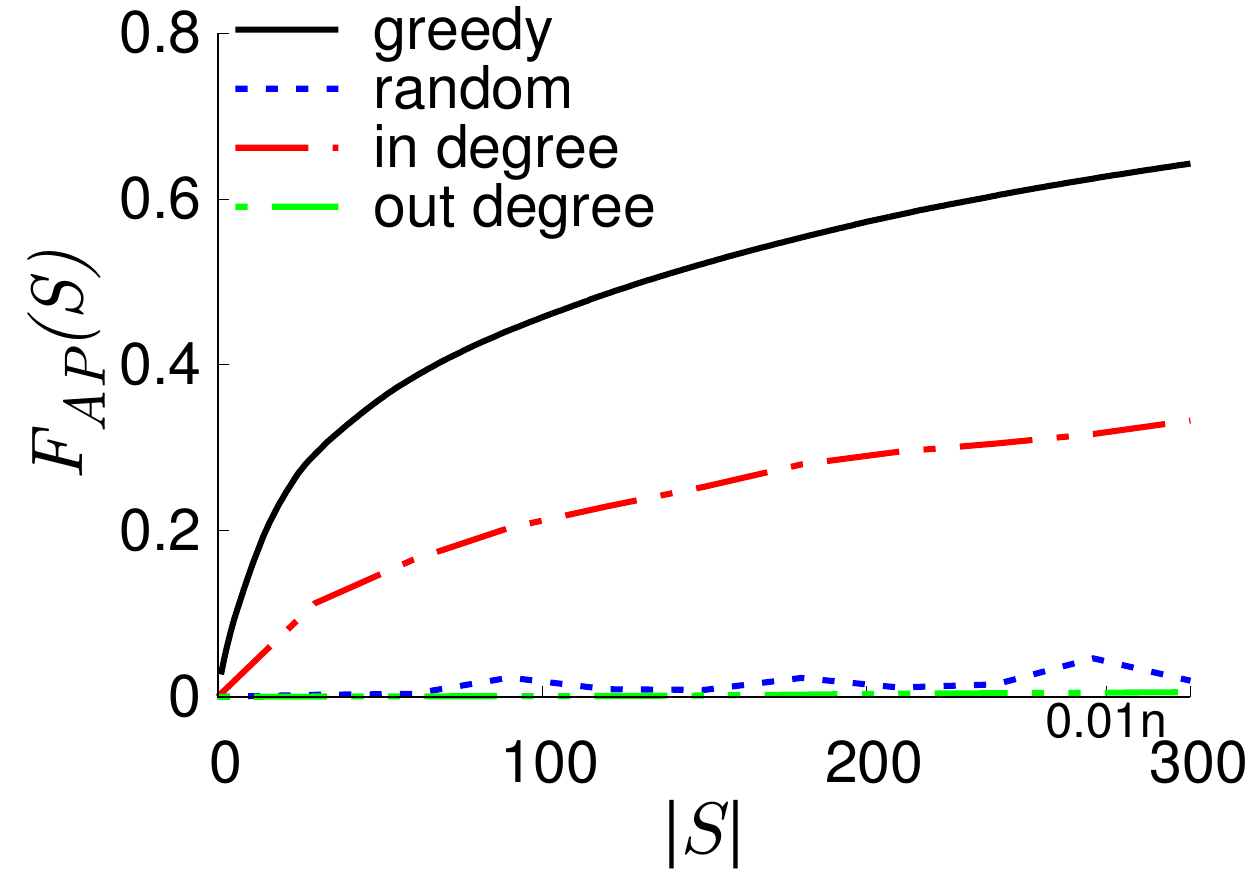}}
  \subfloat[Gowalla]{\includegraphics[width=.333\linewidth]{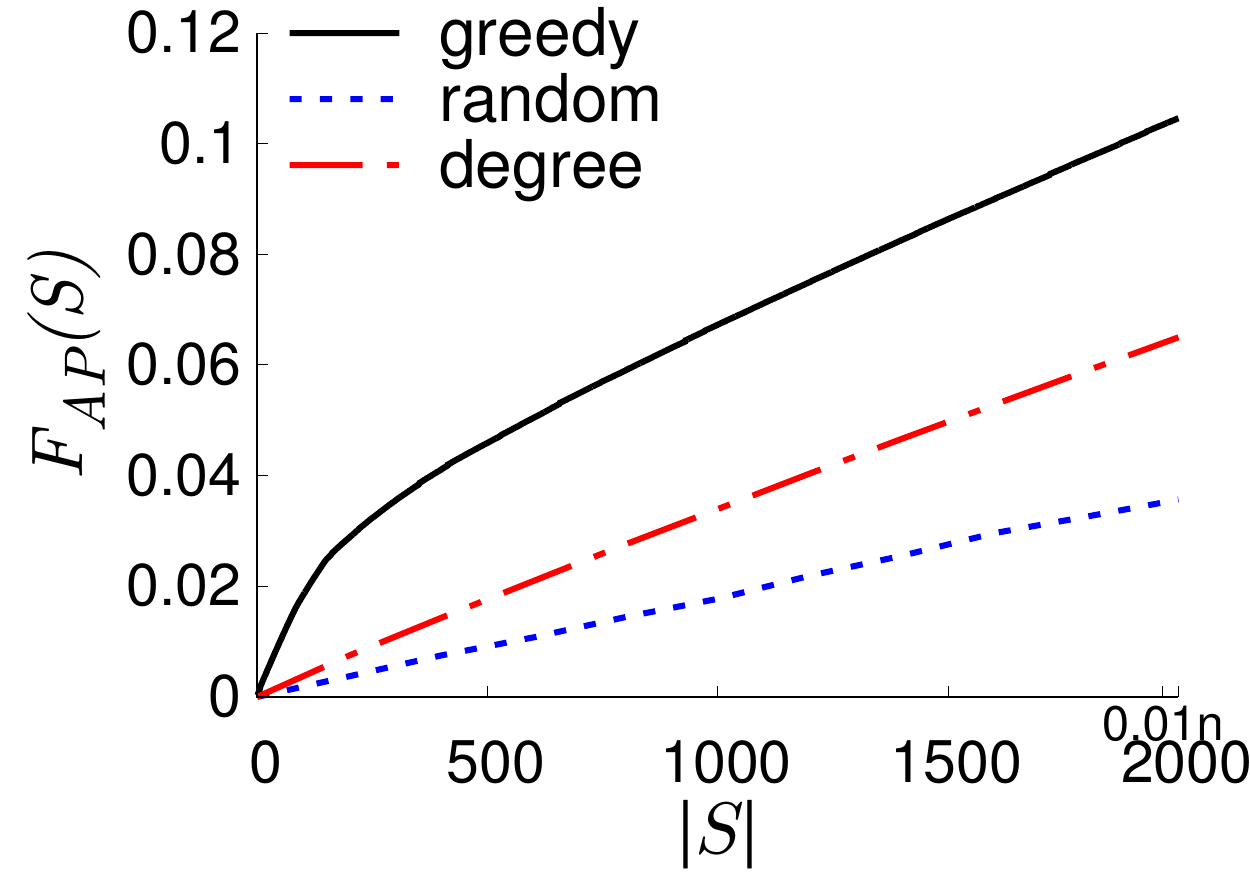}}
  \subfloat[DBLP]{\includegraphics[width=.333\linewidth]{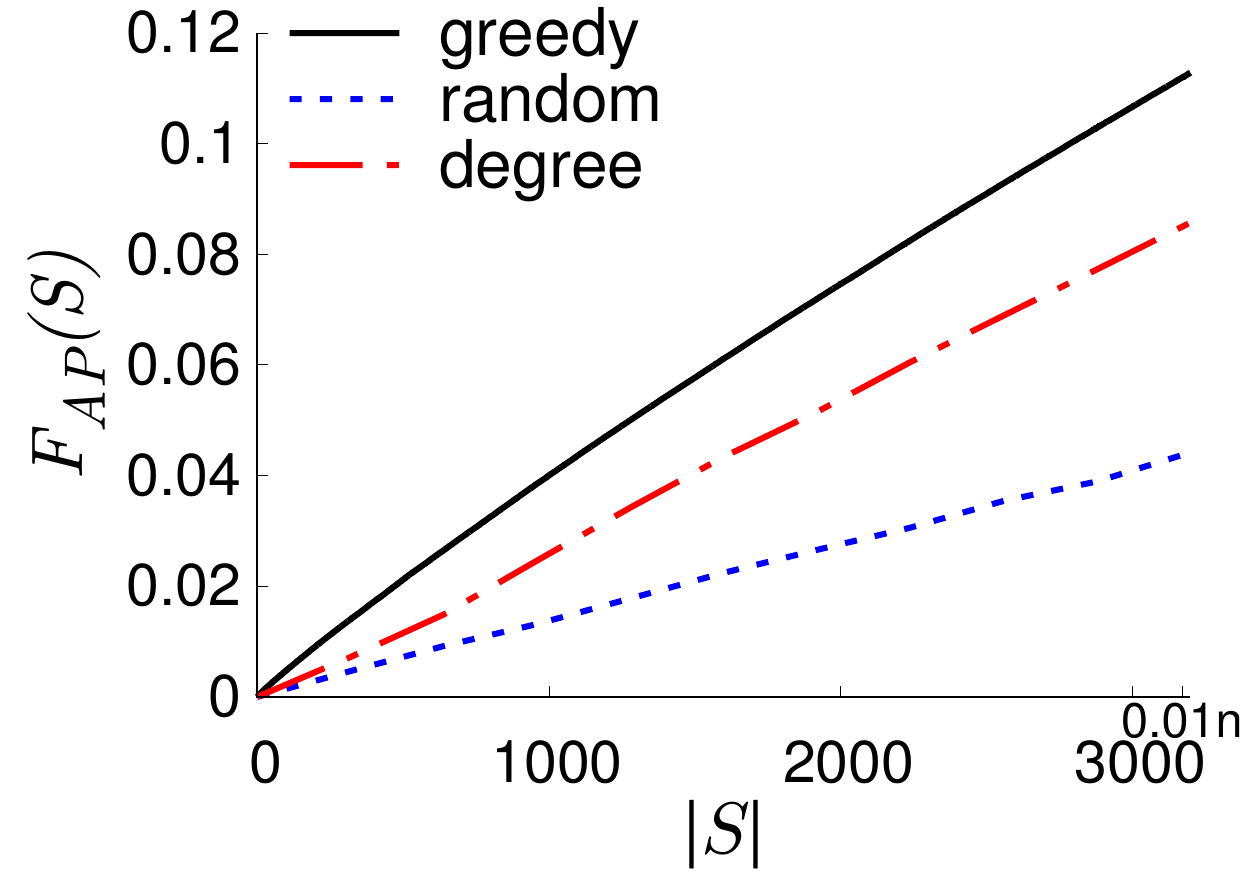}} \\
  \subfloat[Amazon]{\includegraphics[width=.333\linewidth]{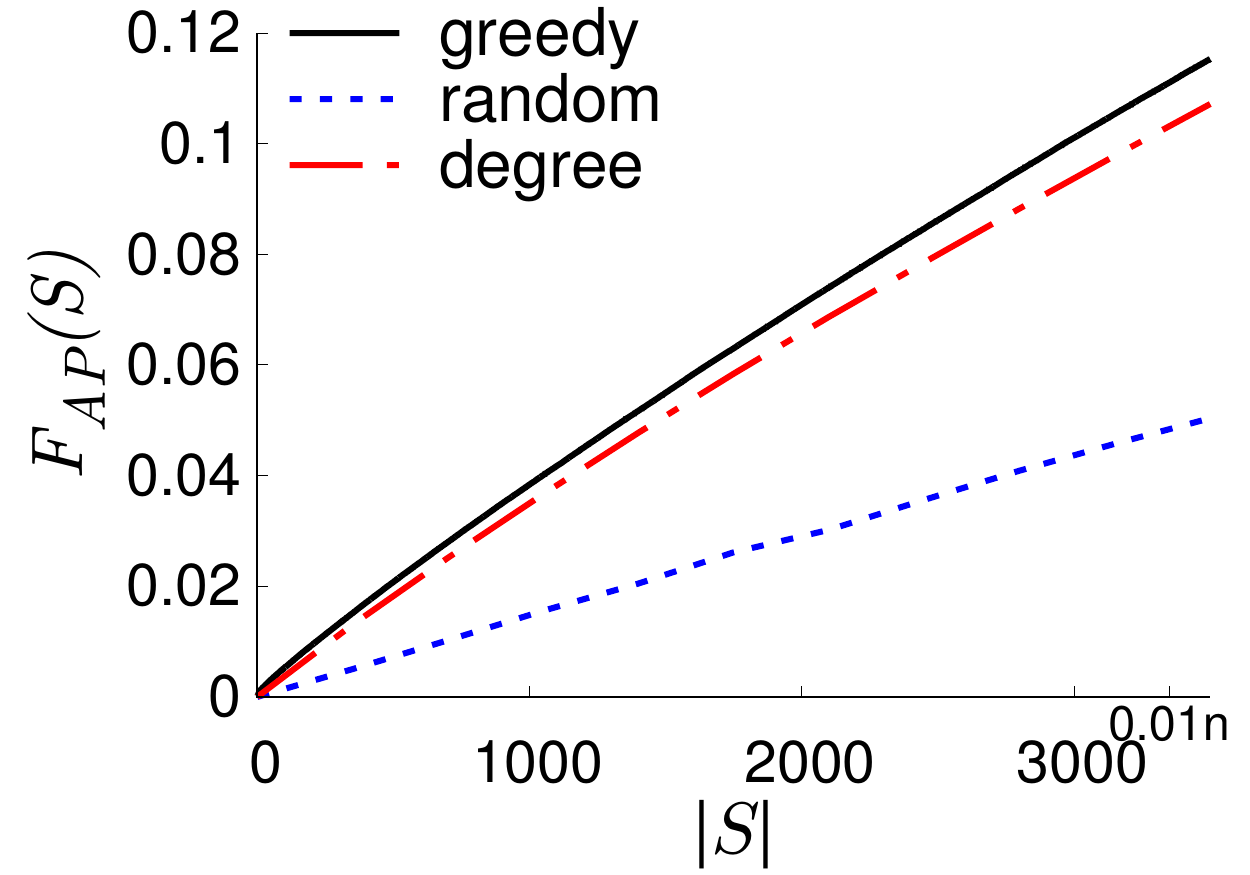}}
  \subfloat[YouTube]{\includegraphics[width=.333\linewidth]{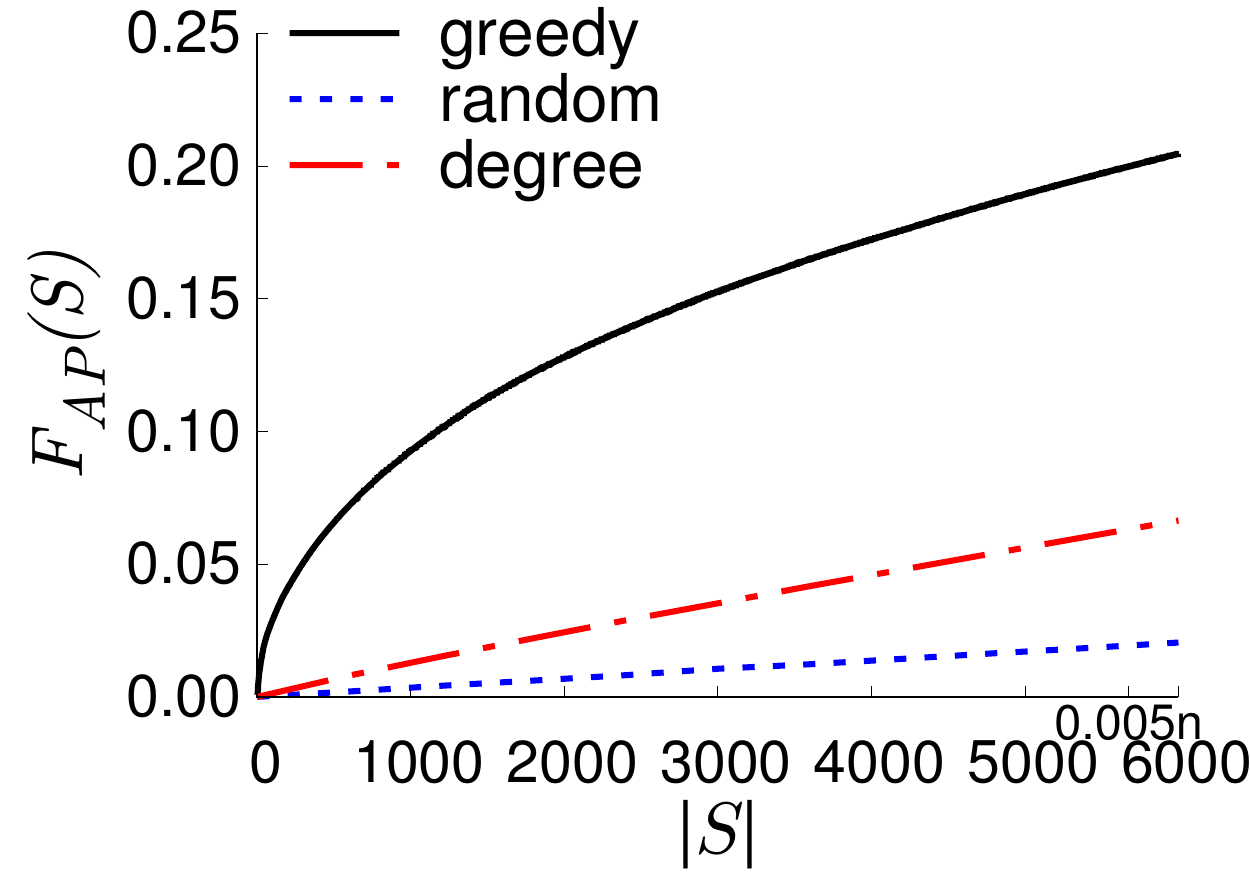}}
  \subfloat[Patents]{\includegraphics[width=.333\linewidth]{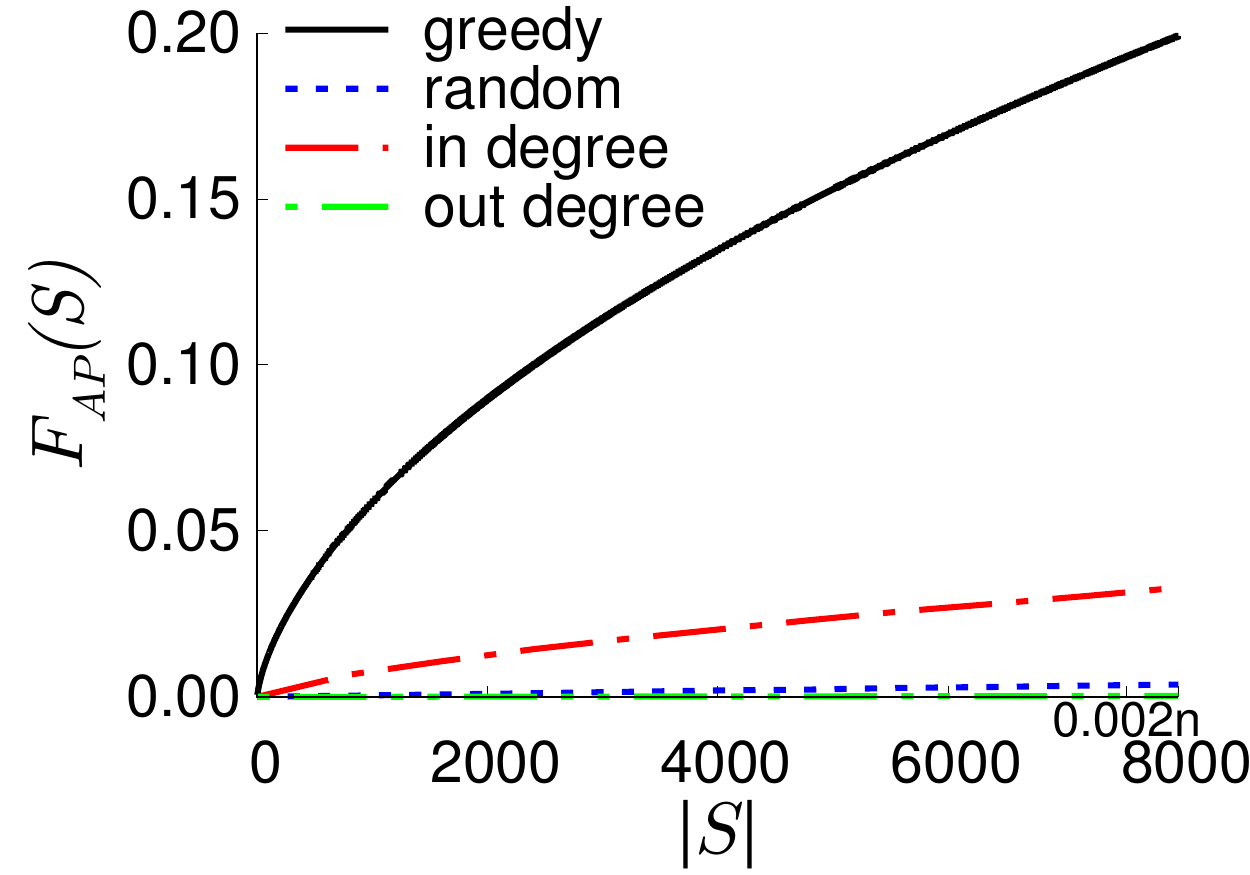}}
  \caption{D-AP maximization ($T=10$)}
  \label{fig:greedy_ap}
\end{figure}

\begin{figure}[htp]
  \centering
  \subfloat[HepTh]{\includegraphics[width=.333\linewidth]{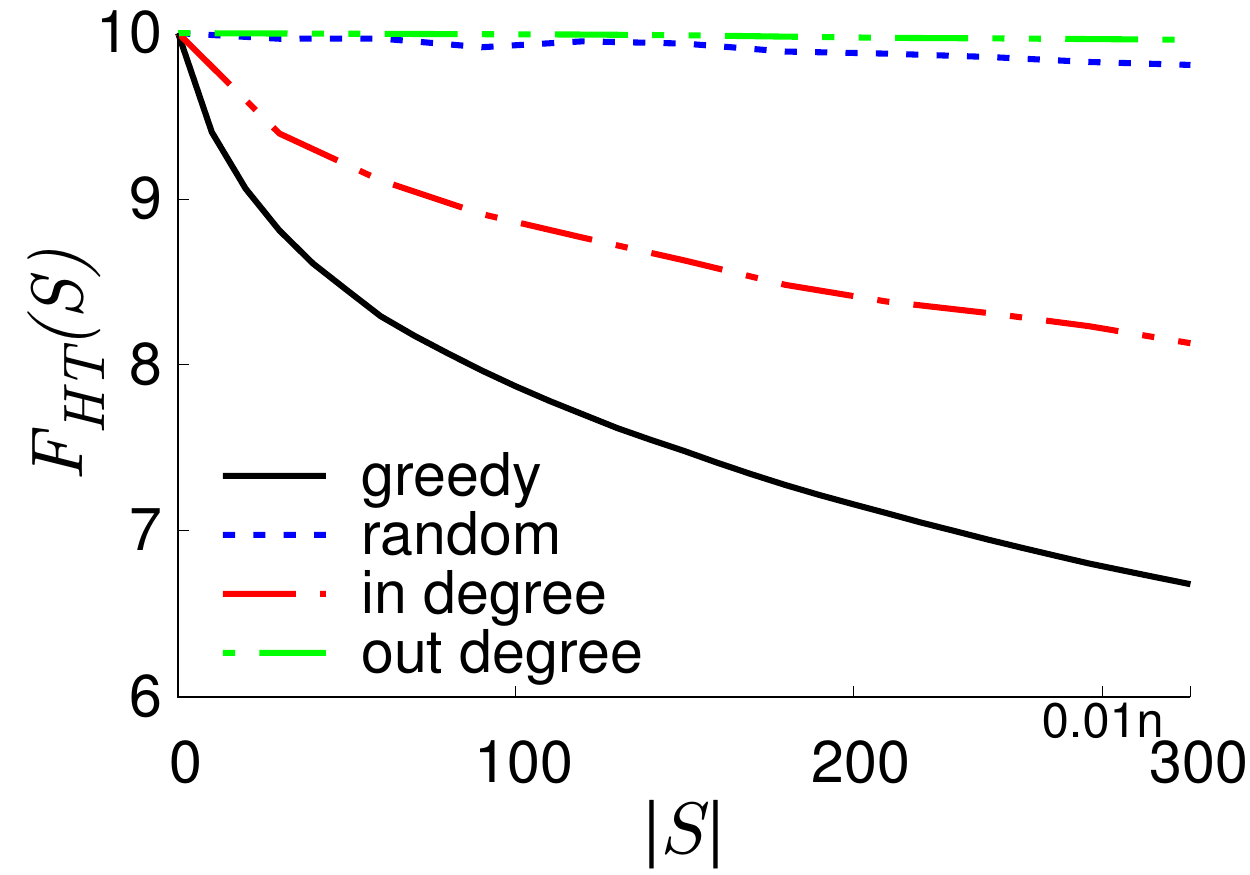}}
  \subfloat[Gowalla]{\includegraphics[width=.333\linewidth]{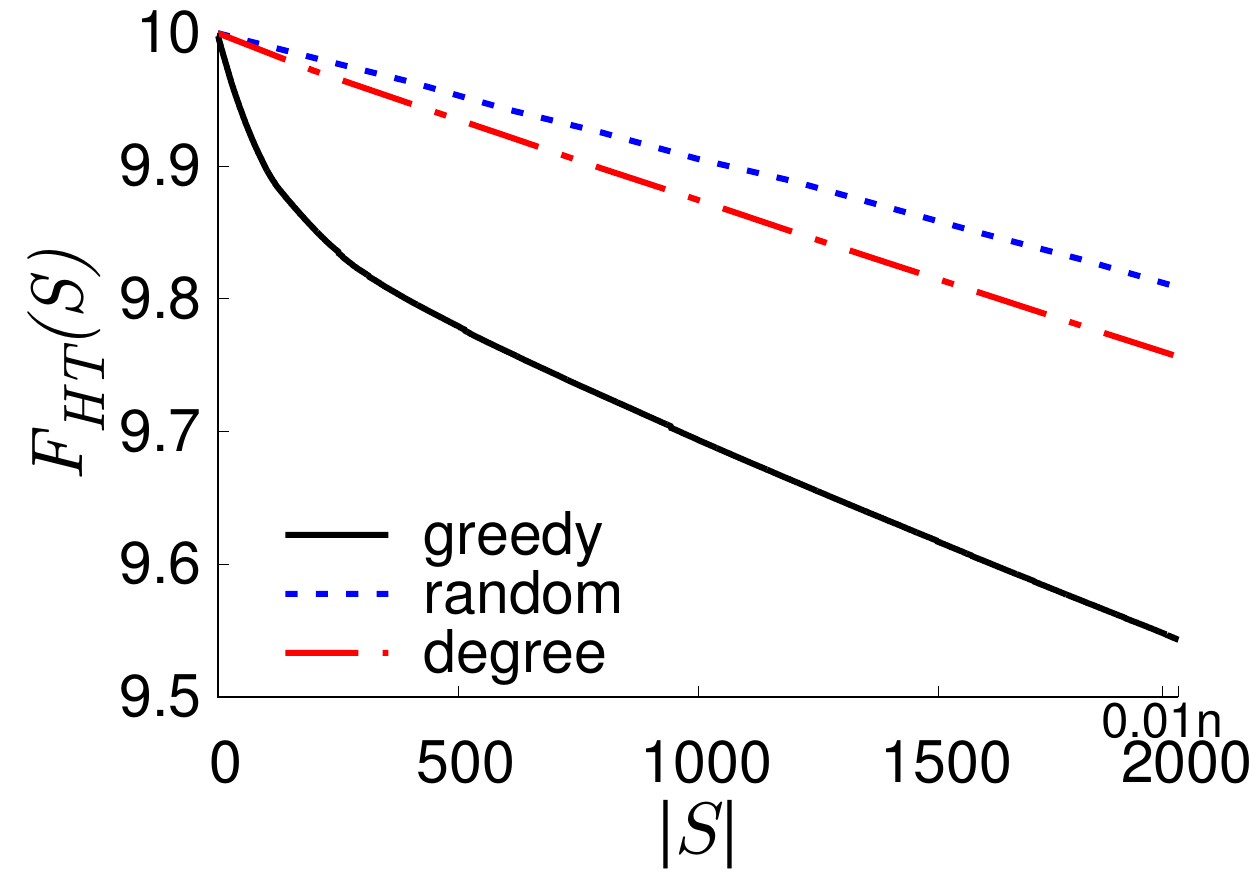}}
  \subfloat[DBLP]{\includegraphics[width=.333\linewidth]{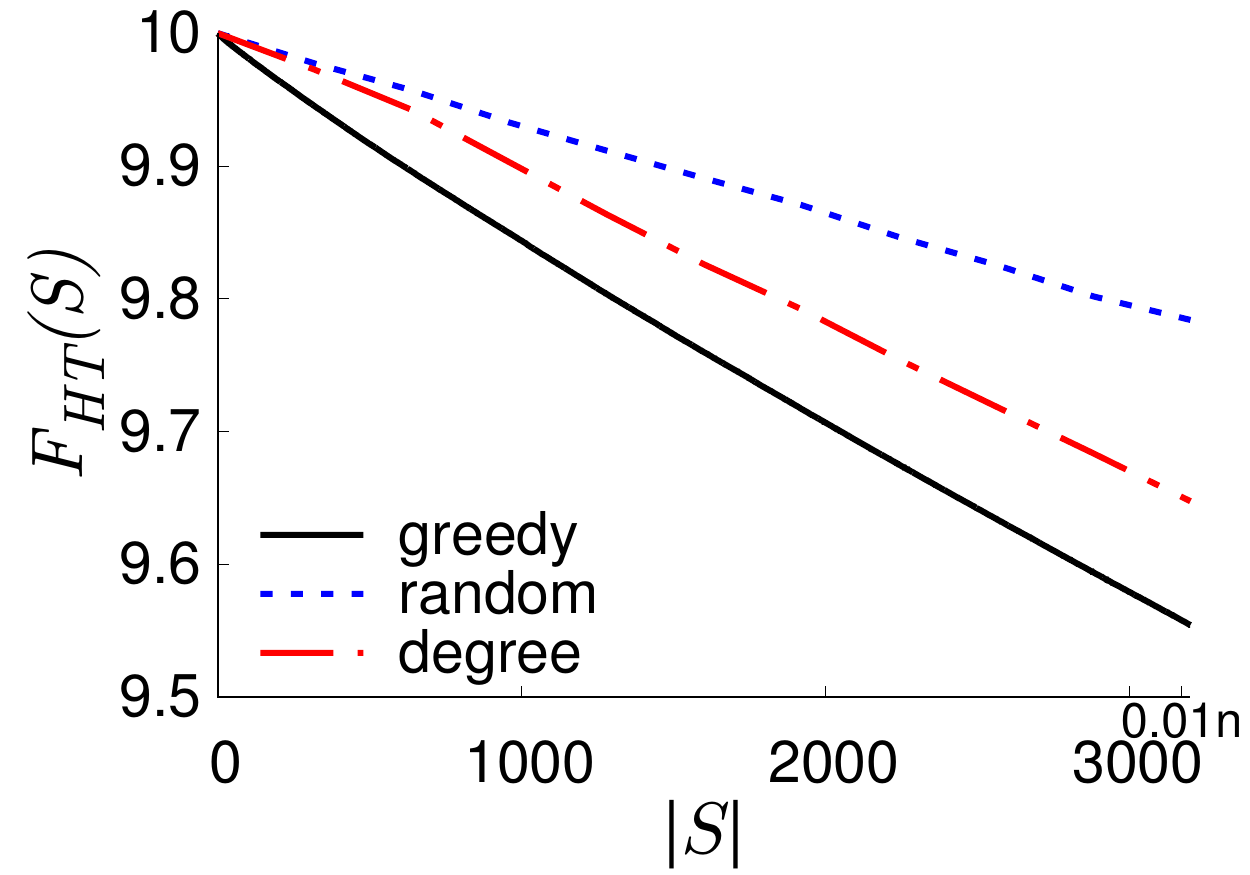}} \\
  \subfloat[Amazon]{\includegraphics[width=.333\linewidth]{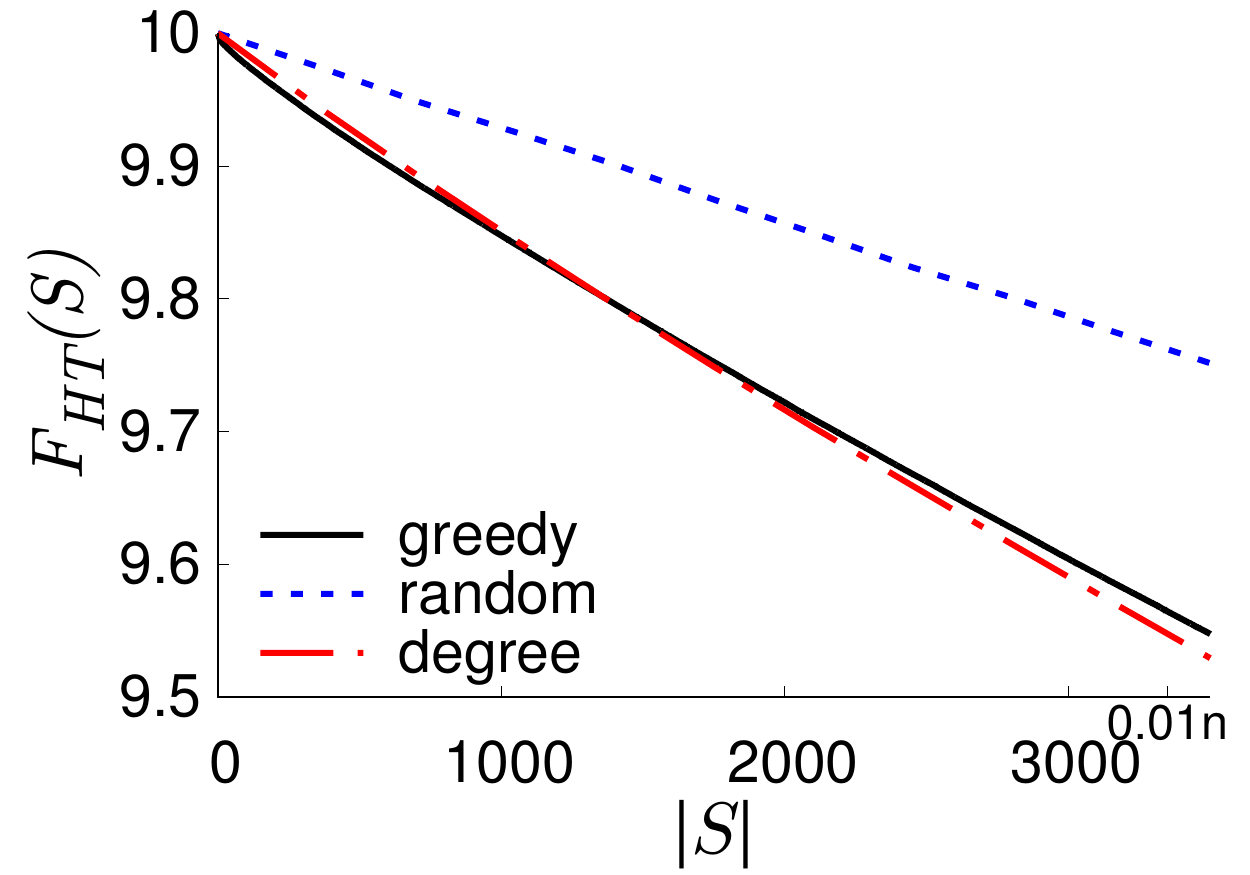}}
  \subfloat[YouTube]{\includegraphics[width=.333\linewidth]{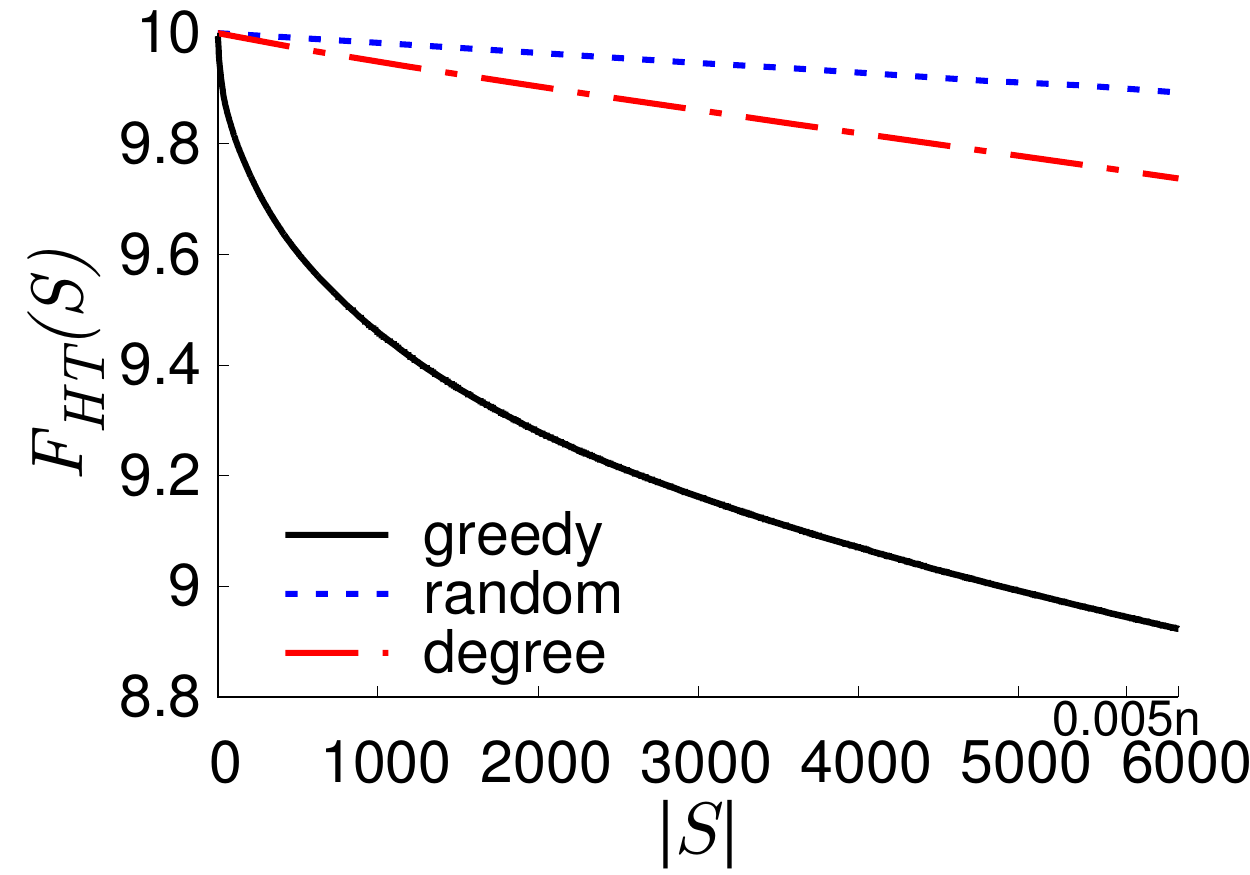}}
  \subfloat[Patents]{\includegraphics[width=.333\linewidth]{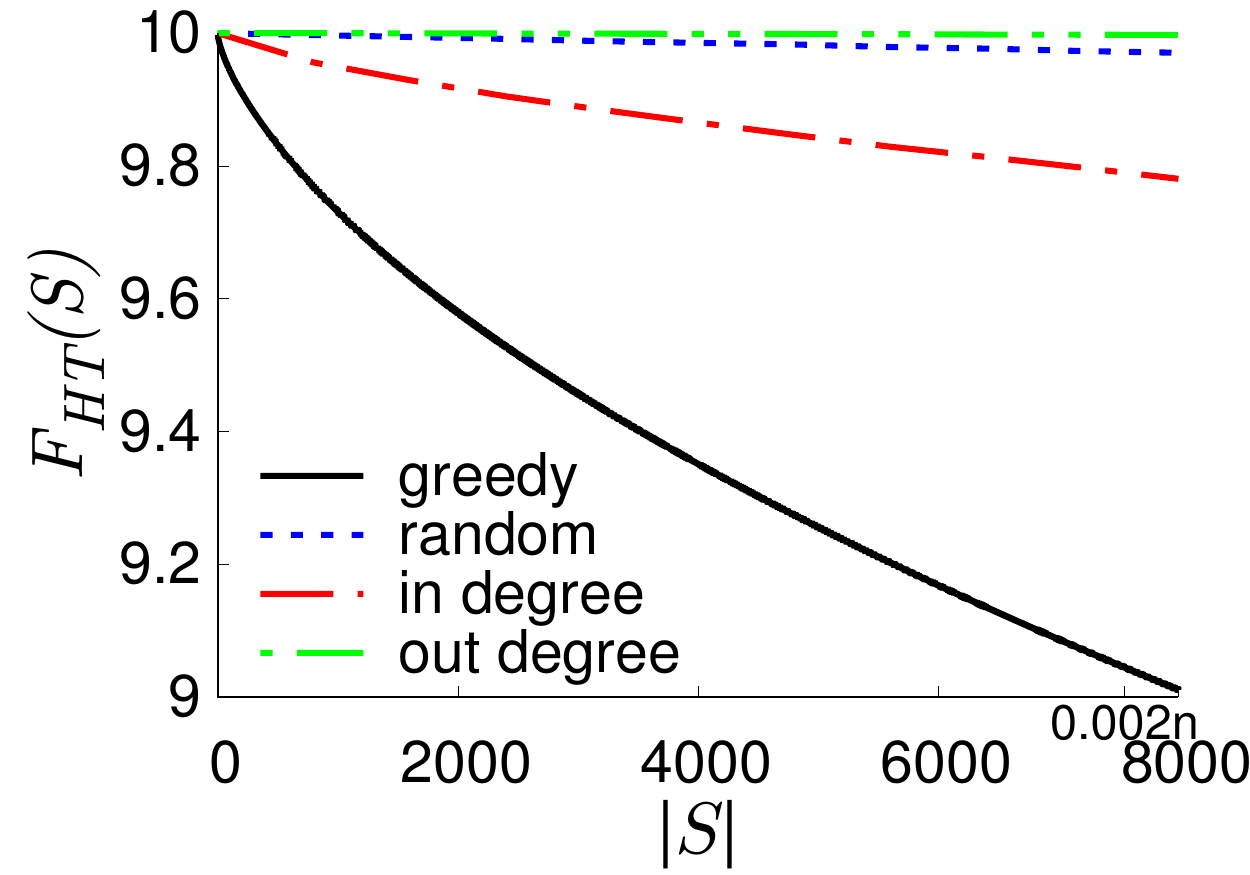}}
  \caption{D-HT minimization ($T=10$)}
  \label{fig:greedy_ht}
\end{figure}

We can clearly see that the greedy algorithm indeed performs much better than the
two baseline methods on all the graphs: the greedy algorithm could choose
connection sources with larger D-AP, and smaller D-HT.
We also note that on the Amazon product network, the greedy algorithm and degree
approach have competitive performance when minimizing D-HT.
In general, the degree approach is better than random approach.
However, on directed graphs HepTh and Patents, the random approach is actually
slightly better than choosing connections by top largest out-degrees.
These results hence show that choosing connection sources using the greedy
approach is more stable than the other baseline methods.

\section{Applications}
\label{sec:applications}

In this section, we study the node discoverability optimization in some real-world
applications and show some interesting observations of the patterns of nodes
maximizing D-AP and minimizing D-HT.

\subsection{Measurements and Observations on Real Networks}

People may argue that nodes maximizing D-AP may also minimize D-HT simultaneously.
Indeed, if this hypothesis is true, then it is not necessary to differentiate the
D-AP maximization problem and D-HT minimization problem, and studying any one of
them is enough.
To investigate this problem, we calculate the overlap of the two sets of nodes
maximizing D-AP and minimization D-HT respectively, under the same cardinality
constraint (using the same settings as in~\autoref{ss:greedy}).
The results are depicted in Figure~\ref{fig:overlap}.

\begin{figure}[htp]
  \centering
  \subfloat[HepTh]{\includegraphics[width=.33\linewidth]{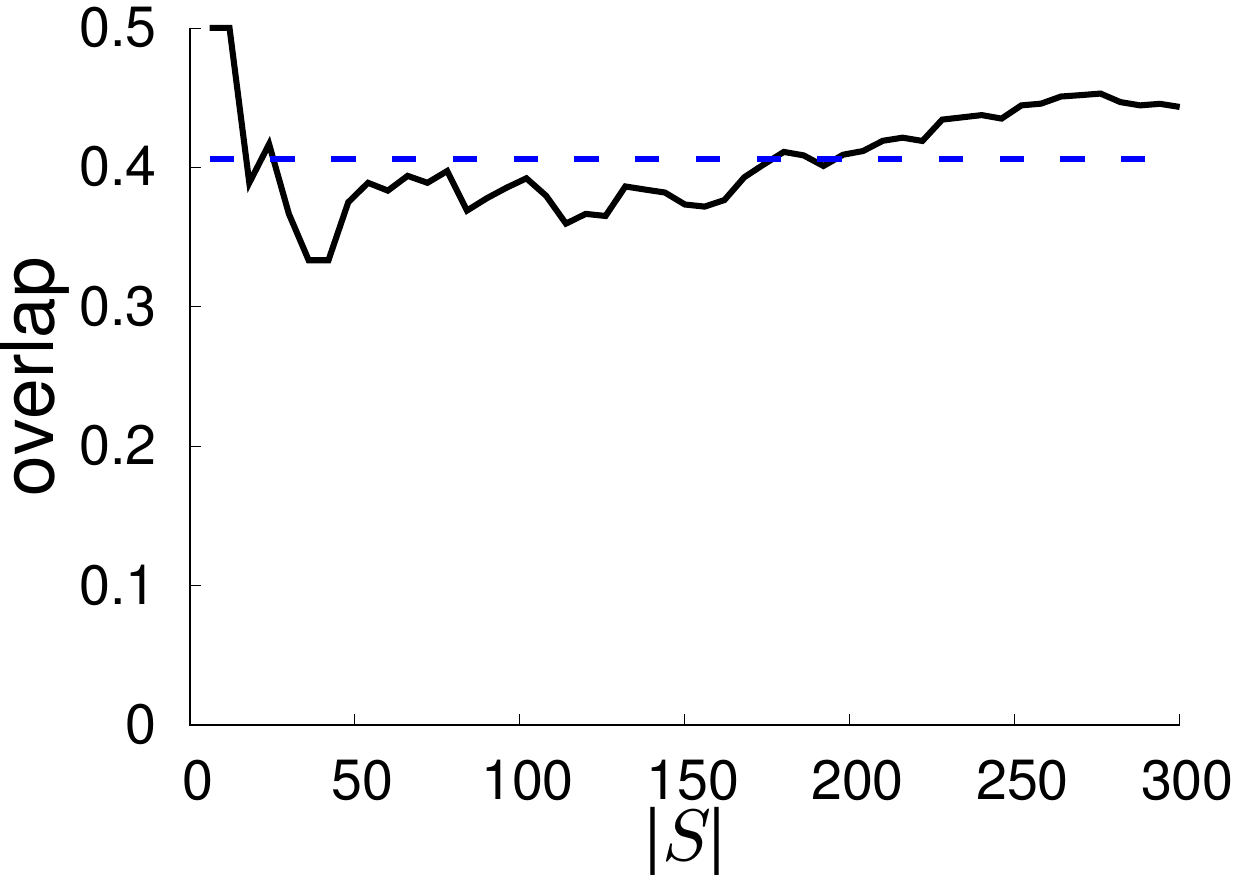}}
  \subfloat[Gowalla]{\includegraphics[width=.33\linewidth]{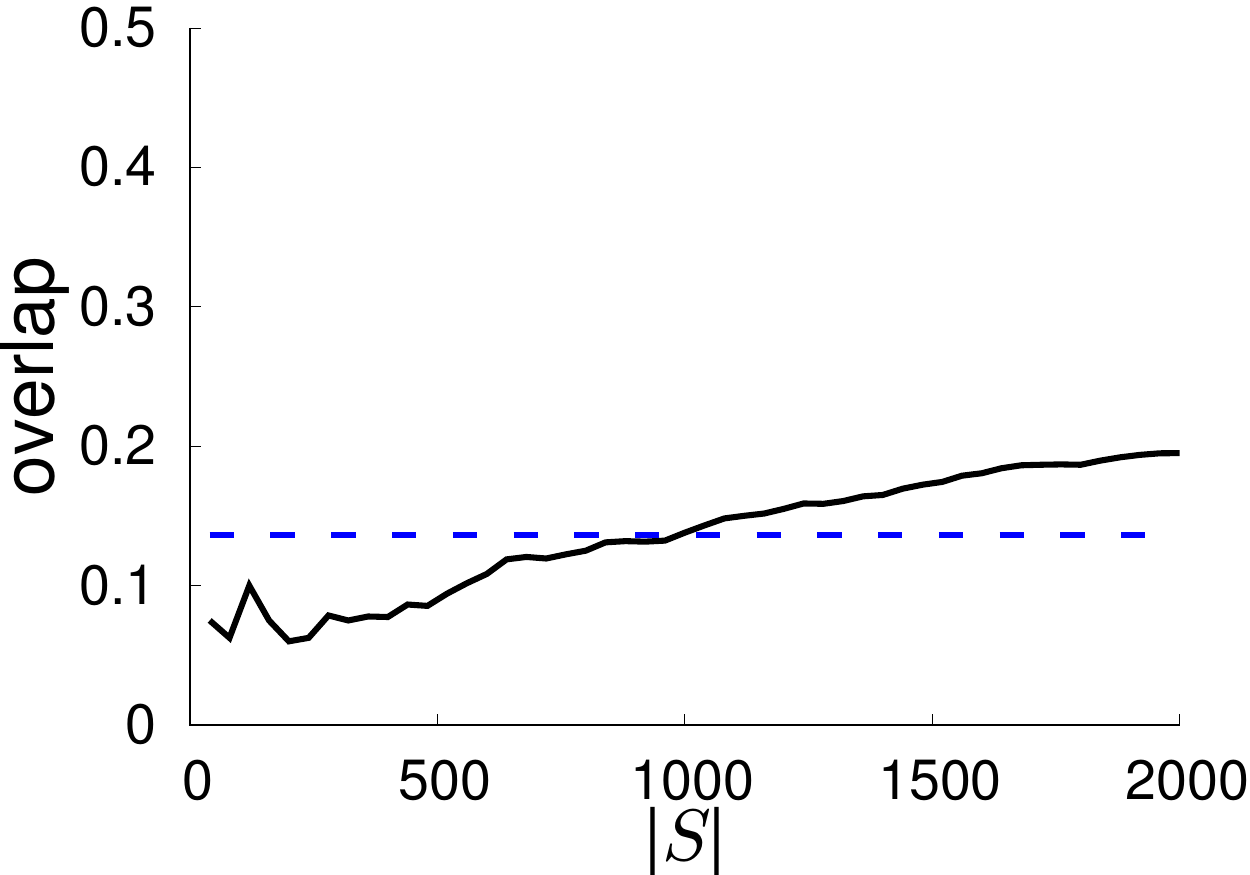}}
  \subfloat[DBLP]{\includegraphics[width=.33\linewidth]{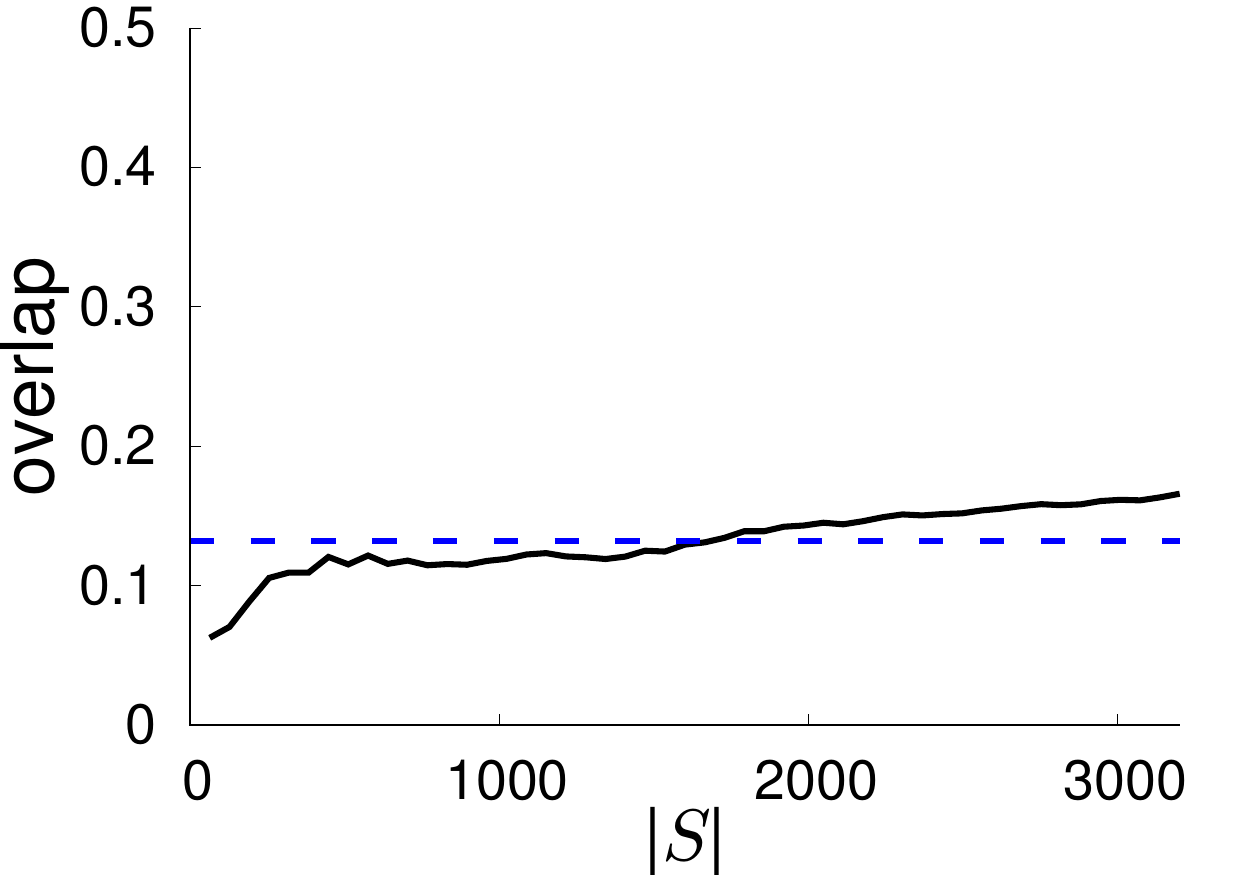}} \\
  \subfloat[Amazon]{\includegraphics[width=.33\linewidth]{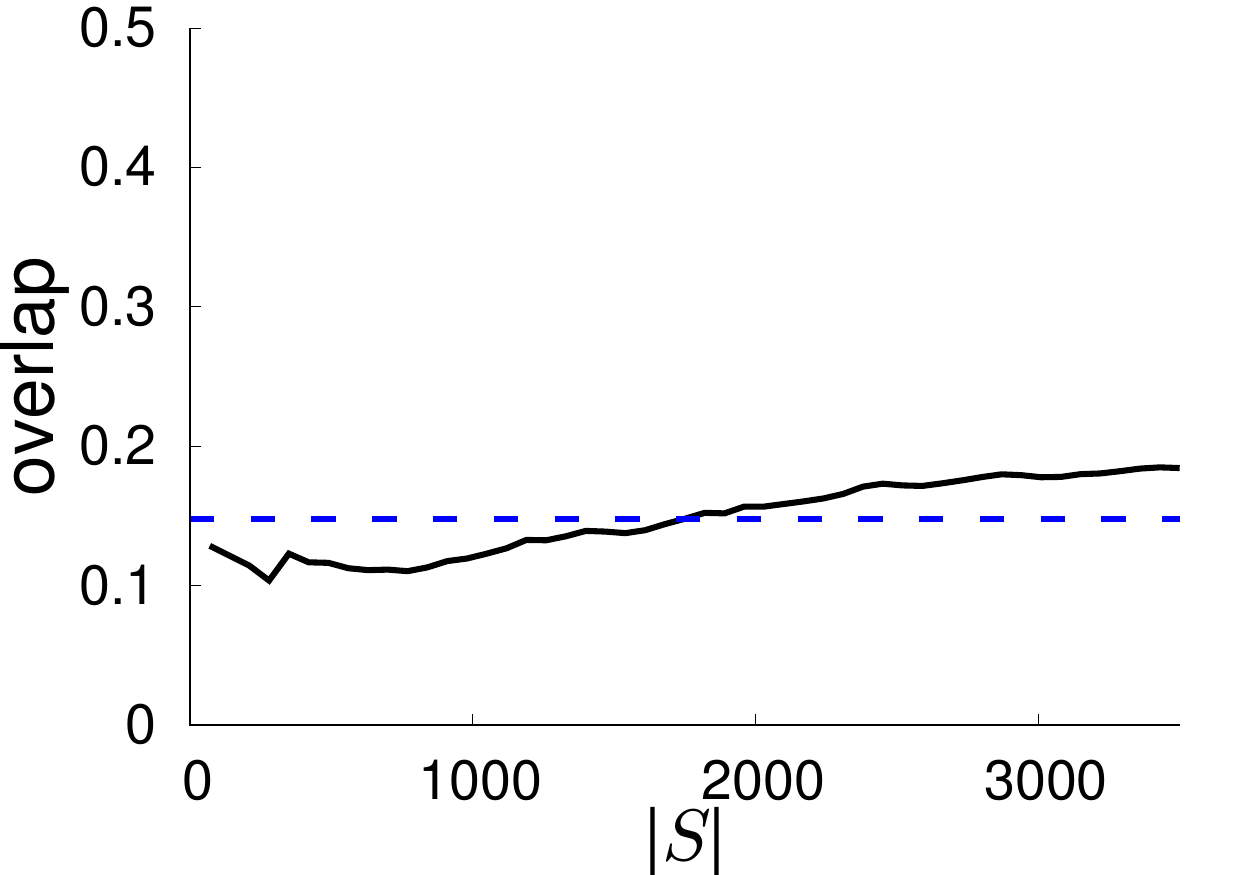}}
  \subfloat[YouTube]{\includegraphics[width=.33\linewidth]{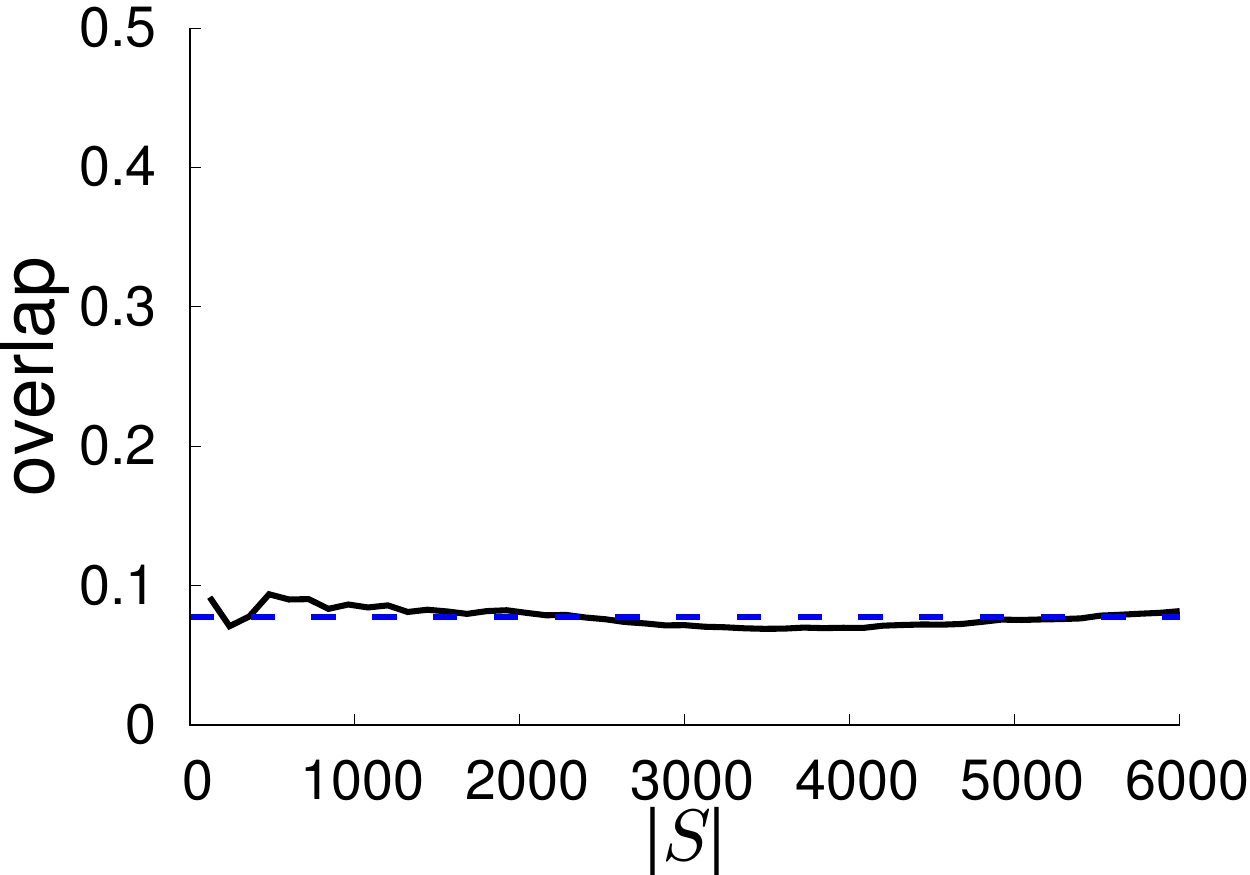}}
  \subfloat[Patents]{\includegraphics[width=.33\linewidth]{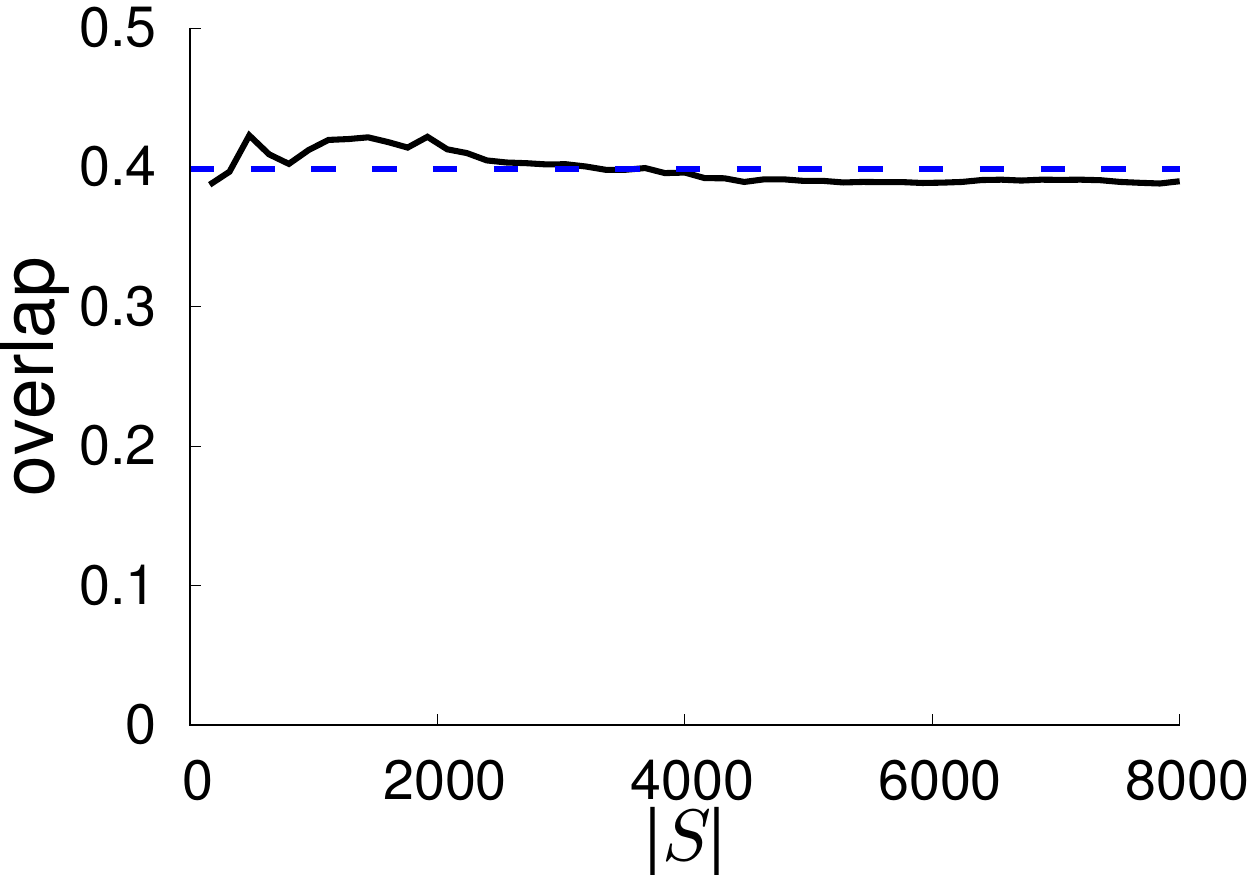}}
  \caption{Overlap between two sets of nodes maximizing D-AP and minimizing D-HT}
  \label{fig:overlap}
\end{figure}

We observe that the overlap is actually small.
On all of these tested graphs, the overlap is less than $50\%$, and on some
graphs, e.g., YouTube, the overlap could be as low as less than $10\%$.
Hence, we demonstrate that the previous hypothesis is actually not true, and it is
necessary to study the two problems separately.
It also makes sense to study the composite optimization
problem~\eqref{eq:composite_opt} as we discussed in~\autoref{ss:objectives}.

\subsection{Cascades Detection on Real Follower Networks}

We next show the usefulness of node discoverability optimization problem in
cascades detection.
The cascades detection problem has been extensively studied in the
literature~\cite{Leskovec2007a,Wilson2009a,Christakis2010,Garcia-Herranz2014,Sun2014,Zhao2014a,Mahmoody2016},
and the goal is to pick a few nodes as sensors from a network so that these
sensors can detect information diffusions in the network as many as possible and
also with time delay as small as possible.
In practice, the cascades detection problem has application in recommending users
(or information sources) that a new user should follow so that the new user will
have maximum information coverage and also minimum time delay of receiving
information in a follower network such as Twitter and Sina Weibo.
As we discussed in Introduction, this problem can also be formulated as a node
discoverability optimization problem.
In the following discussion, we evaluate the quality of nodes obtained by solving
node discoverability optimization from the perspective of maximizing information
coverage and minimizing time delay.

We use two real-world follower networks from Weibo and Douban, which are two
popular OSNs in China, and the graph statistics are summarized in
Table~\ref{tab:data}.
In a follower network, an edge has direction from a user to another user she
follows (i.e., from a follower to its followee).
However, the the direction of information diffusion on a follower network is in a
reverse direction, i.e., from a followee to its followers.
Hence, we actually need to solve the node discoverability problem on a reversed
follower network where each edge direction is reversed.

We consider two types of information diffusion on a follower network:
\begin{itemize}
\item \textbf{Random walk (RW) diffusion:} A piece of information spreads on a
  follower network in the way of random walk.
  That is, at each step of diffusion, the information cascade randomly picks a
  neighbor of current resident node to infect.
  The RW diffusion model is inspired from the letter forwarding process in
  Milgram's experiment~\cite{Travers1969}.
\item \textbf{Independent cascade (IC) diffusion:} Each information cascade starts
  from a seed node.
  When a node $i$ first becomes active at step $t$, it is given only one chance to
  infect each of its neighbors $j$ with success probability $p_{ij}$.
  If a neighbor $j$ is infected at $t$, then $j$ becomes active at next step
  $t+1$; but whether $i$ succeeds in infecting its neighbors at step $t$, it
  cannot make any further attempts to infect its
  neighbors~\cite{Bikhchandani1992}.
\end{itemize}

We simulate $100,000$ and $200,000$ cascades on Weibo and Douban respectively, and
measure the fraction of cascades detected by a set of nodes (referred to as the
{\em coverage}), and also the average minimum time delay of detecting a cascade
(referred to as the {\em delay}).
We set cardinality budgets to be $0.1\%,0.2\%$ and $0.3\%$ of graph size, and
depict the performance of different sets of nodes in
Figures~\ref{fig:weibo_cascades} and~\ref{fig:douban_cascades}.

\begin{figure}[h]
  \centering
  \subfloat[RW]{\includegraphics[width=.4\linewidth]{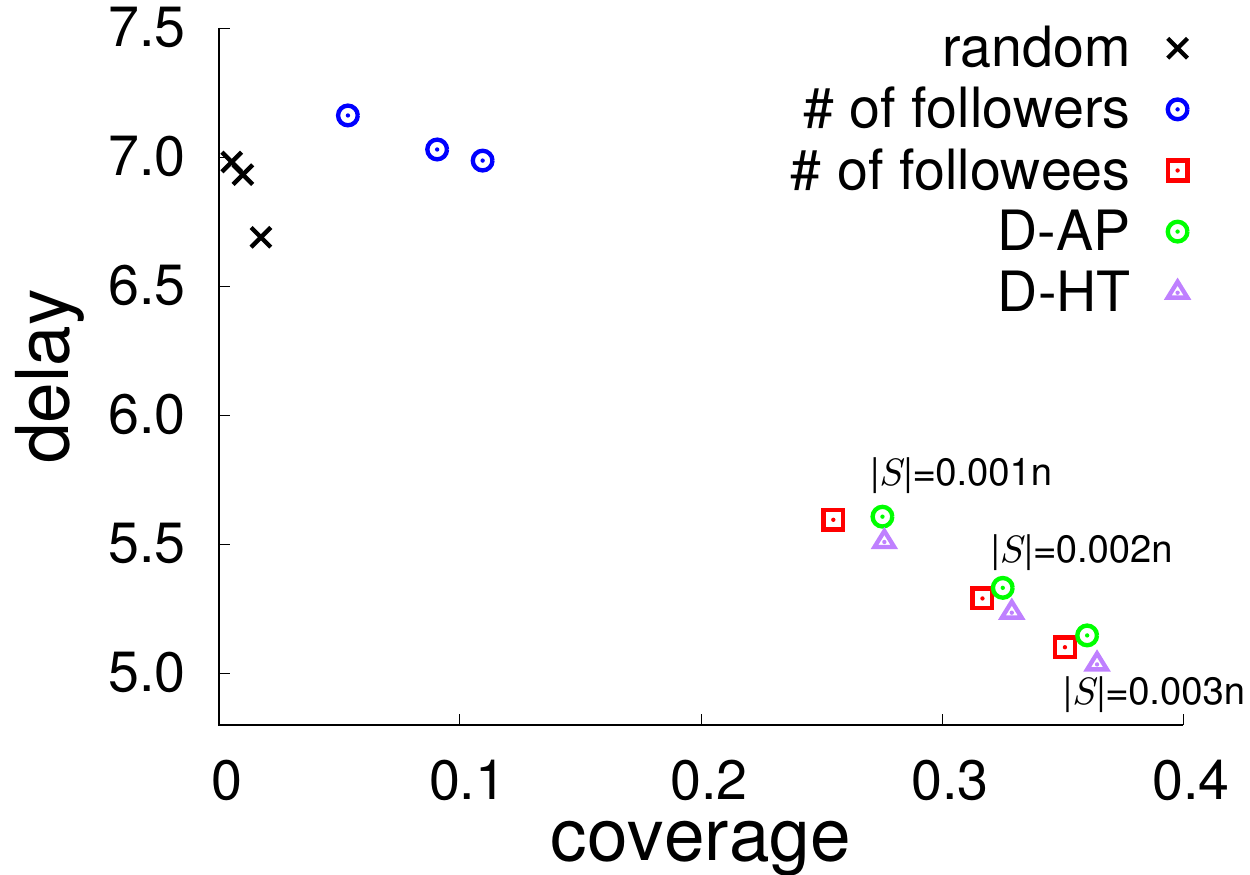}}\quad
  \subfloat[IC]{\includegraphics[width=.4\linewidth]{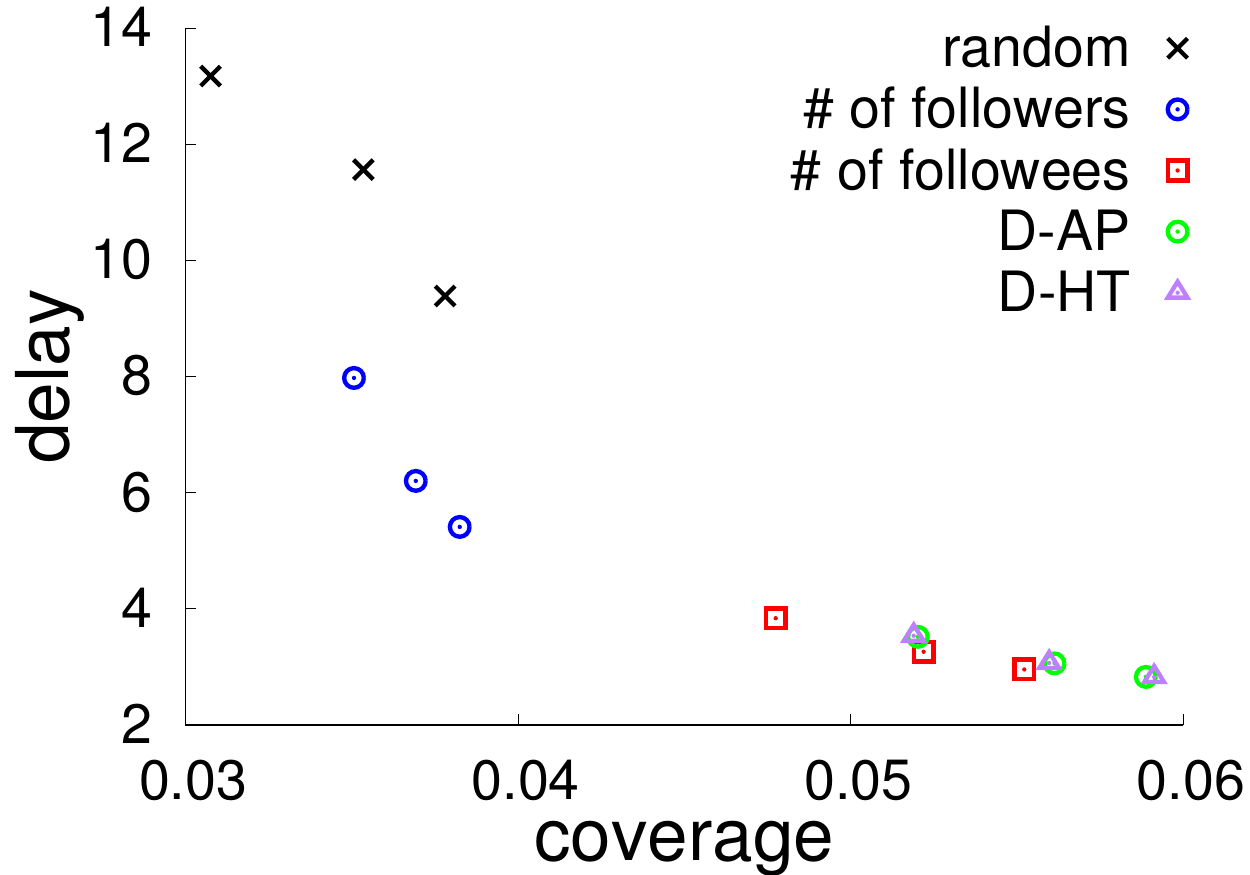}}
  \caption{Cascades detection on Weibo}
  \label{fig:weibo_cascades}
\end{figure}

\begin{figure}[h]
  \centering
  \subfloat[RW]{\includegraphics[width=.4\linewidth]{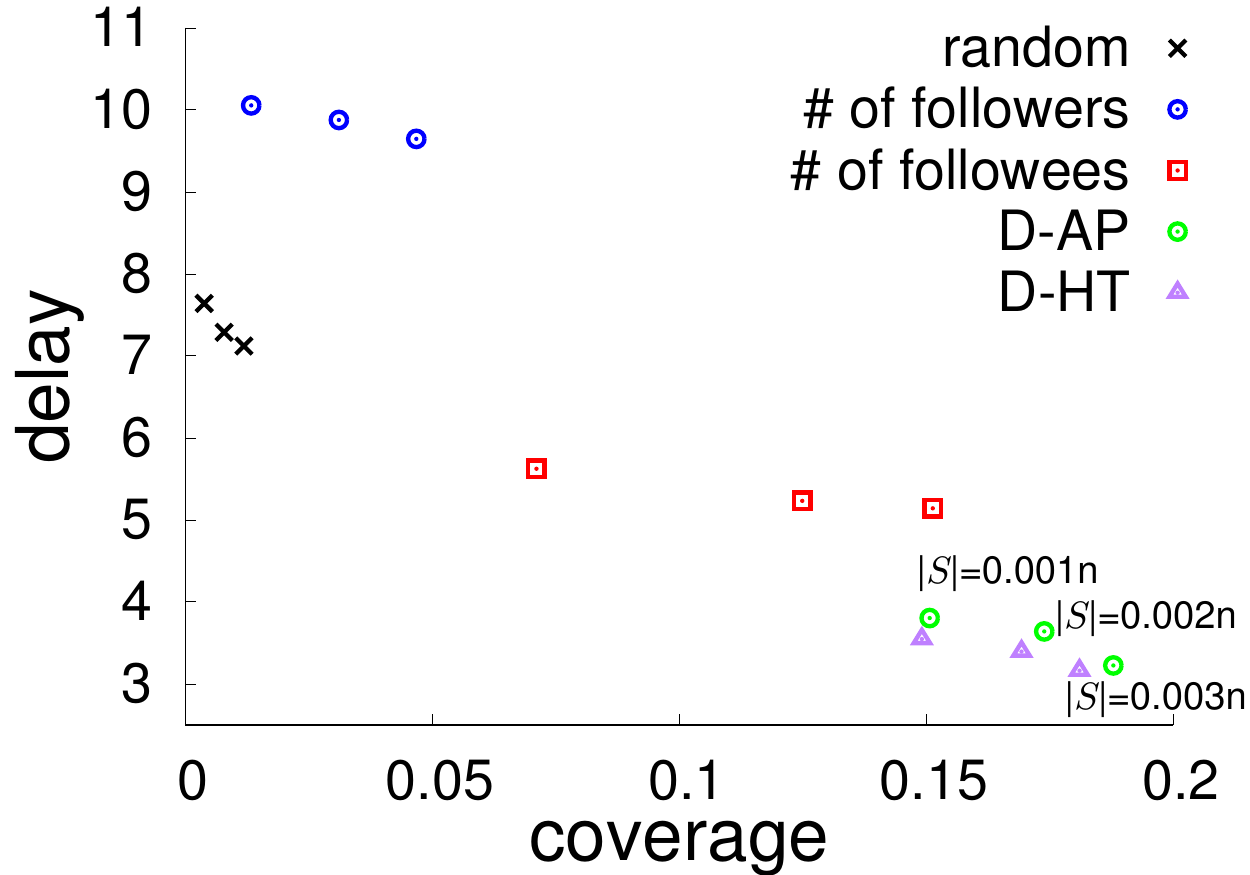}}\quad
  \subfloat[IC]{\includegraphics[width=.4\linewidth]{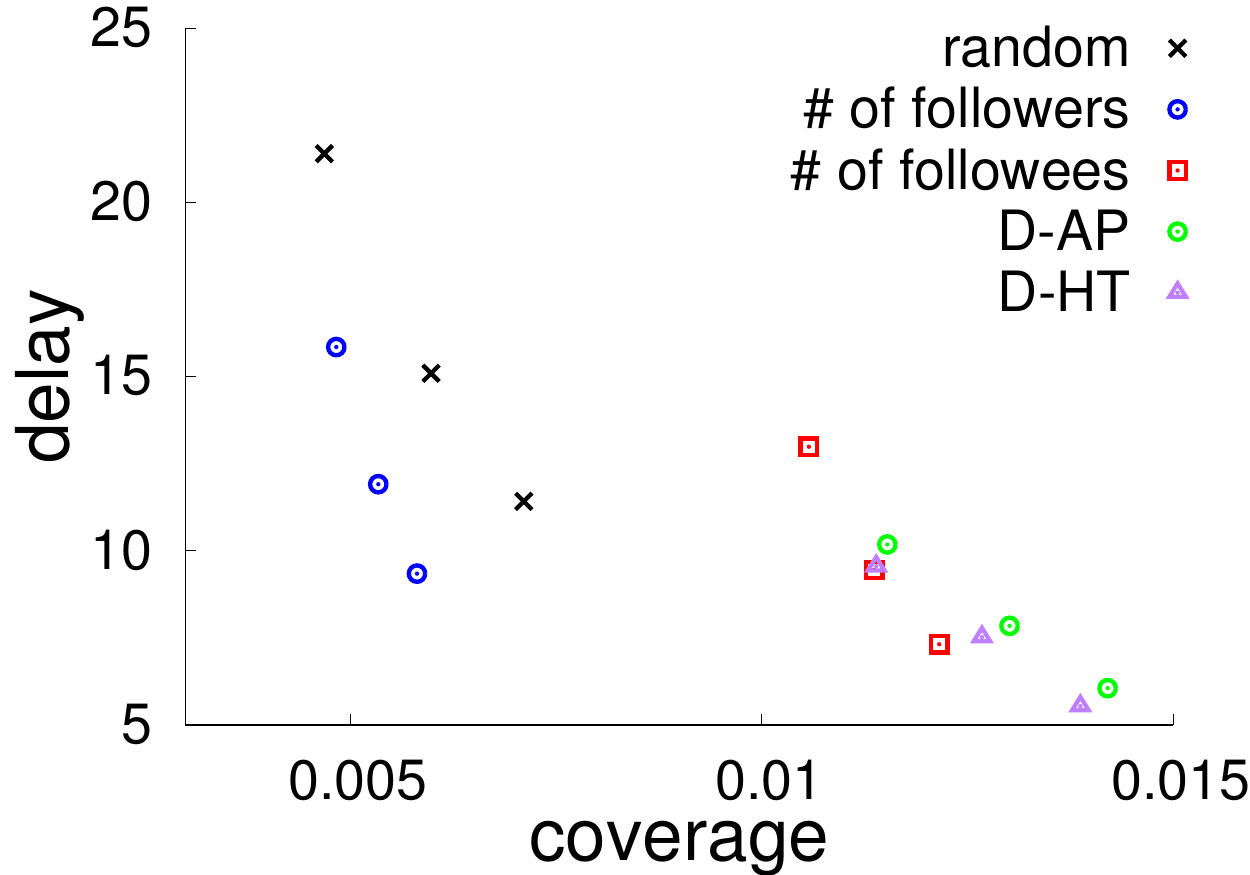}}
  \caption{Cascades detection on Douban}
  \label{fig:douban_cascades}
\end{figure}

In the plot, points lay on the bottom right corner imply good performance as these
nodes detect cascades with large coverage and small delay; while points lay on the
top left corner imply poor performance as these nodes detect cascades with small
coverage and large delay.
We observe that, for both diffusion models, nodes obtained by solving node
discoverability optimization problems are close to the bottom right corner,
indicating good performance; nodes obtained by the other methods, e.g., random and
top largest number of followers, are close to the top left corner, indicating the
poor performance.
We also observe that nodes minimizing D-HT usually have smaller delay than nodes
maximizing D-AP, except the case of IC model on Weibo which is indistinguishable.
In conclusion, the results show the usefulness of node discoverability
optimization problem on cascades detection.

\section{Related Work}
\label{sec:relatedwork}

This section is devoted to review some related literature.

Node discoverability is related to the concept of node
centrality~\cite{Freeman1978,Boldi2014}, which captures the importance of a node
in analyzing complex networks, such as closeness~\cite{Cohen2014} and
betweenness~\cite{MahmoodyEtAl2016}.
The classic closeness centrality~\cite{Cohen2014} characterizes how close a node
is to other nodes in a graph, and can be easily modified to measure how close the
other nodes to the target node.
If we use this modified closeness centrality to measure the target node's
discoverability, we will bear the burden of solving the shortest path problem,
which is a notorious difficulty on large scale weighted graphs.
So it is not scalable to use closeness or other shortest path based centrality
measures to quantify a node's discoverability.

Two recent work~\cite{Antikacioglu2015} and~\cite{Rosenfeld2016} shed some light
on defining proper node discoverability.
Antikacioglu et al.~\cite{Antikacioglu2015} study the web discovery optimization
problem in an e-commerce website, and their goal is to add links from a small set
of popular pages to new pages to make as many new pages discoverable as possible.
They define a page is discoverable if the page has at least $a\geq 1$ links from
popular pages.
However, such a definition may be too strict, as it actually assumes that a user
only browses a site for one hop.
In fact, a user could browse the site for several hops, and finally discover a
page, even though the page may have no link from popular pages.
Rosenfeld and Globerson~\cite{Rosenfeld2016} study the optimal tagging problem in
a network consisting of tags and items, and their goal is to pick $k$ tags for a
new item in order to maximize the new item's incoming traffic.
This problem is formulated as maximizing the absorbing probability of the new item
in an absorbing Markov chain.
Measuring a node's discoverability by absorbing probability relieves the
restriction of~\cite{Antikacioglu2015}, but it implicitly assumes that a user has
infinite amount of time or patience to browse the network to discover an item,
which is, however, not the case~\cite{Simon1971,Scaria2014}.
We avoid the two extremes by taking a Middle Way, and propose two orthogonal
definitions of node discoverability based on finite length random walks.

Our proposed node discoverability definitions D-AP and D-HT leverage the theory of
absorbing Markov chains~\cite{Doyle1984,Trivedi2016}.
Recently, Mavroforakis et al.~\cite{Mavroforakis2015} propose the absorbing random
walk centrality to measure a node's importance in a graph.
Golnari et al.~\cite{Golnari2015} propose several measures based on hitting time
to measure node reachability in communication networks.
Hitting time is also used in measuring node
similarity~\cite{Sarkar2008,Sarkar2010a} in large graphs, and finding dominating
sets of a graph~\cite{Li2014}.

From the algorithmic point of view, our method leverages submodularity and
supermodularity of the defined discoverability measures, and uses the greedy
heuristic~\cite{Nemhauser1978,Krause2014} to solve the optimization problem.
There has been rich literature in scaling up the greedy algorithm in different
applications, e.g., solving the set cover problem for data residing on
disk~\cite{Cormode2010}, solving the max-$k$ cover problem using
MapReduce~\cite{Chierichetti2010}, calculating group closeness centrality by
exploiting the properties of submodular set functions~\cite{Zhao2017}, etc.
In contrast, we design an ``estimation-and-refinement'' approach for implementing
an efficient oracle call in the greedy algorithm, built on top of the contemporary
efficient random walk simulation systems~\cite{Fogaras2005,Kyrol2013,Liu2016c}.

\section{Conclusion}
\label{sec:conclusion}

This work considers a general problem of node discoverability optimization problem
on networks, that appears in a wide range of applications.
We propose two definitions of node discoverability, namely, D-AP based on the
truncated absorbing probability, and D-HT based on the truncated hitting time.
Although optimizing a target node's discoverability with regard to the two
measures is NP-hard, we find that the two measures satisfy submodularity and
supermodularity, respectively.
This enables us to use the greedy algorithm to find provably near-optimal
solutions for the optimization problem.
To scale up the greedy algorithm for handling large networks, we propose an
efficient estimation-and-refinement implementation of the oracle call.
Experiments conducted on real graphs demonstrate that our method provides a good
trade-off between estimation accuracy and computational efficiency, and it
achieves thousands of times faster than the method using dynamic programming.

\bibliographystyle{model5-names}

\appendix

\section*{Proof of Theorem~\ref{thm:np-hard}}
\label{app:proof-np-hard}

\begin{proof}
  The D-AP maximization problem can be easily reduced from the optimal tagging
  problem~\cite{Rosenfeld2016}, which has been proved to be NP-hard.
  Hence, the D-AP maximization problem is NP-hard.
  We only need to prove the NP-hardness of D-HT minimization problem.

  We prove that the decision problem of D-HT minimization problem is NP-complete
  by a reduction from the vertex cover problem.
  The decision problem asks: Given a graph $G$ and some threshold $J$, does there
  exist a solution $S$ such that $\FHT(S)\leq J$?
  We will prove that, given threshold $J(k)$, there exists a solution $S$ for the
  decision problem iff a vertex cover problem has a cover $S$ of size at most $k$.

  The vertex cover problem is defined on an undirected graph $H=(V,E)$, where
  $V=\{0,\ldots,n-1\}$, and $E\subseteq V\times V$.
  Let $S\subseteq V$ denote a subset of vertices of size $k$.
  We construct an instance of the D-HT minimization problem on directed graph
  $G=(V', E')$, where $V'=V\cup\{m,n\}$ and edge set $E'$ includes both $(i,j)$
  and $(j,i)$ for each edge $(i,j)\in E$.
  $E'$ contains additional edges: For each $i\in V$, we add an edge $(i,m)$ with
  proper weight to make the transition probabilities $p_{im}=\epsilon$; we add
  self-loop edges to vertices $m$ and $n$, and thus $m$ and $n$ become two
  absorbing vertices, i.e., transition probabilities $p_{mm}=p_{nn}=1$.
  For this particular instance of D-HT minimization problem, we need to choose
  connection sources $S$ from $V$; once a source $s$ is selected, we set
  transition probability $p_{sn}=1$, which is equivalent to set edge weight
  $w_{sn}=\infty$.

  Assume $S$ is a vertex cover on graph $H$.
  Then, for each vertex $i\in S$, a walker starting from $i$ hits $n$ using one
  step with probability $1$.
  For each vertex $i\in V\backslash S$, a walker starting from $i$ hits $m$ and
  becomes absorbed on $m$ with probability $\epsilon$ (the corresponding hitting
  time is $T$); the walker passes a neighbor in $V$, which must be in $S$, and
  then hits $n$, with probability $1-\epsilon$ (the corresponding hitting time is
  $2$).
  This achieves the minimum D-HT, denoted by $J(k)\triangleq \FHT(S)=\frac{k}{n} +
  \frac{n-k}{n}[2(1-\epsilon) + T\epsilon]$.

  If a solution $S$ satisfies $\FHT(S)\leq J(k)$ on graph $G$, then $S$ must be a
  vertex cover on graph $H$.
  Otherwise, assume $S$ is not a vertex cover on graph $H$.
  Then there must be an edge $(i,j)$ such that $i,j\notin S$.
  The probability that a walker starting from $i$ and becoming absorbed at vertex
  $m$ will be strictly larger than $\epsilon$, and becomes absorbed at vertex $n$
  using two steps will be strictly smaller than $1-\epsilon$.
  As a result, the hitting time from $i$ will be strictly larger than
  $2(1-\epsilon)+T\epsilon$ whenever $T\geq 3$.
  Thus, $\FHT(S)>J(k)$.

  The above analysis indicates that if there exists an efficient algorithm for
  deciding whether there exists a set $S, |S|=k$ such that $\FHT(S) \geq J(k)$ on
  graph $G$, we could use the algorithm to decide whether graph $H$ has a vertex
  cover of size at most $k$, thereby demonstrating the NP-hardness of the D-HT
  minimization problem.
\end{proof}

\section*{Proof of Theorem~\ref{thm:FAP}}
\label{app:proof-submodular}

\begin{proof}
  The monotonicity and submodularity of a set function is both closed under
  non-negative linear combinations.
  Hence, for $\FAP(S)=1/n\sum_{i\in V}p_i^T(S)$, we only need to prove that
  $p_i^T(S)$ is non-decreasing and submodular.

  \header{Monotonicity.}
  To show that $p_i^T(S)$ is non-decreasing $\forall i\in V$, we use induction on
  $T$.
  Let $S_1\subseteq S_2\subseteq V$, and $i\in V$.
  For $T=0$, it holds that $p_i^0(S_1)=p_i^0(S_2)=0$.
  (Also notice that $p_n^t(S)\equiv 1,\forall S,\forall t$.)

  Assume the conclusion holds for $T=t$, i.e., $p_i^t(S_1)\leq p_i^t(S_2)$.
  Consider the case when $T=t+1$,
  \begin{align*}
    p_i^{t+1}(S_1) - p_i^{t+1}(S_2)
    &= \sum_k\bigl[ p_{ik}(S_1)p_k^t(S_1) - p_{ik}(S_2)p_k^t(S_2) \bigr] \\
    &\leq \sum_k\bigl[ p_{ik}(S_1) - p_{ik}(S_2) \bigr]p_k^t(S_2) \\
    &= \sum_{k\neq n}\bigl[ p_{ik}(S_1) - p_{ik}(S_2) \bigr]p_k^t(S_2) \\
    &\hspace{30pt} + \bigl[ p_{in}(S_1) - p_{in}(S_2) \bigr]p_n^t(S_2) \\
    &\leq\sum_{k\neq n}\bigl[ p_{ik}(S_1) - p_{ik}(S_2) \bigr]
      + p_{in}(S_1) - p_{in}(S_2) \\
    &= \sum_k\bigl[ p_{ik}(S_1) - p_{ik}(S_2) \bigr] \\
    &= 0.
  \end{align*}
  The first inequality holds due to the induction assumption, and the last
  inequality holds because $p_{ik}(S_1)\geq p_{ik}(S_2)$ for $k\neq n$,
  $p_k^t(S_2)\leq 1$, and $p_n^t(S_2)=1$.
  Thus, by induction, we conclude that $p_i^T(S)$ is non-decreasing.

  \header{Submodularity.}
  To show that $p_i^T(S)$ is submodular $\forall i\in V$, we also use induction.
  Let $S_1\subseteq S_2\subseteq V, s\in V\backslash S_2, S_1'\triangleq
  S_1\cup\{s\}, S_2'\triangleq S_2\cup\{s\}$, and $\delta_i^t(s;S)\triangleq
  p_i^t(S\cup\{s\})-p_i^t(S)$.
  Notice that $\delta_n^t(s;S)\equiv 0, \forall S,\forall t$.
  For $T=0$, because $p_i^0(S)=0, \forall S\subseteq V$, then $\delta_i^0(s;S_1) =
  \delta_i^0(s;S_2)$.
  Assuming $\delta_i^t(s;S_1)\geq\delta_i^t(s;S_2)$ holds for $T=t$, we consider
  the case when $T=t+1$.

  \noindent\textbullet\, $i\in V\backslash S_2'\cup S_1$.
  In this case, probability transitions $\{p_{ik}\}_{k\in V}$ are all constants,
  i.e., $p_{ik}(S_1') = p_{ik}(S_1) = p_{ik}(S_2) = p_{ik}(S_2')\triangleq
  p_{ik}$.
  So,
  \begin{align*}
    \delta_i^{t+1}(s;S_1)
    &= \sum_kp_{ik}\bigl[ p_k^t(S_1')-p_k^t(S_1) \bigr] \\
    &= \sum_kp_{ik}\delta_k^t(s;S_1) \\
    &\geq \sum_kp_{ik}\delta_k^t(s;S_2) \\
    &= \delta_i^{t+1}(s;S_2).
  \end{align*}

  \noindent\textbullet\, $i\in S_1\backslash S_2$.
  In this case, probability transitions have relation $p_{ik}(S_1') =
  p_{ik}(S_1)\geq p_{ik}(S_2) = p_{ik}(S_2')$, for $k\neq n$.
  Hence,
  \begin{align*}
    \delta_i^{t+1}(s;S_1) - \delta_i^{t+1}(s;S_2)
    &= \sum_k\biggl\{ p_{ik}(S_1)
      \bigl[ p_k^t(S_1') - p_k^t(S_1) \bigr]\\
    &\hspace{40pt} -p_{ik}(S_2)
      \bigl[ p_k^t(S_2')-p_k^t(S_2) \bigr] \biggr\} \\
    &= \sum_{k\neq n}
      \bigl[ p_{ik}(S_1) \delta_k^t(s;S_1) - p_{ik}(S_2)\delta_k^t(s;S_2) \bigr] \\
    &\geq\sum_{k\neq n}p_{ik}(S_2)
      \bigl[ \delta_k^t(s;S_1)-\delta_k^t(s;S_2) \bigr] \\
    &\geq 0.
  \end{align*}

  \noindent\textbullet\, $i=s$.
  In this case, probability transitions have relation $p_{ik}(S_2') =
  p_{ik}(S_1')\leq p_{ik}(S_1) = p_{ik}(S_2)$, for $k\neq n$.
  So,
  \begin{align*}
    \delta_i^{t+1}(s;S_1) -\delta_i^{t+1}(s;S_2)
    &=\sum_k\biggl\{ p_{ik}(S_1) \bigl[ p_k^t(S_2) - p_k^t(S_1) \bigr]\\
    &\hspace{40pt}
      -p_{ik}(S_1') \bigl[ p_k^t(S_2') - p_k^t(S_1') \bigr]\biggr\} \\
    &\geq\sum_{k\neq n}p_{ik}(S_1)
      \bigl[ \delta_k^t(s;S_1) - \delta_k^t(s;S_2) \bigr]\\
    &\geq 0.
  \end{align*}
  The three cases above have covered each $i\in V$.
  By induction, we then conclude that $p_i^T(S)$ is a submodular set function, and
  this completes the proof of Theorem~\ref{thm:FAP}.
\end{proof}

\section*{Proof of Theorem~\ref{thm:FHT}}
\label{app:proof-supermodular}

\begin{proof}
  The monotonicity and supermodularity of a set function is both closed under
  non-negative linear combinations.
  Hence, for $\FHT(S)=1/n\sum_{i\in V}h_i^T(S)$, we only need to prove that
  $h_i^T(S)$ is non-increasing and supermodular.

  \header{Monotonicity.}
  To show that $h_i^T(S)$ in is non-increasing $\forall i\in V$, we use induction.
  Let $S_1\subseteq S_2\subseteq V$.
  According to the definition of hitting time given in Definition~\ref{def:ht}, we
  find that, for $T=0$, $h_i^0(S_1)=h_i^0(S_2)=0, \forall i\in V$.

  Now we assume that the conclusion holds for $T=t$, i.e., $h_i^t(S_1)\geq
  h_i^t(S_2)$ holds for every $i\in V$.
  (Notice that $h_n^t(S)\equiv 0, \forall S,\forall t$.)
  Consider the case when $T=t+1$,
  \begin{align*}
    h_i^{t+1}(S_1)
    &= 1+\sum_{k\neq n}p_{ik}(S_1)h_k^t(S_1) \\
    &\geq 1+\sum_{k\neq n}p_{ik}(S_2)h_k^t(S_2) \\
    &= 1+\sum_kp_{ik}(S_2)h_k^t(S_2) \\
    &= h_i^{t+1}(S_2).
  \end{align*}
  The inequality holds because $h_k^t(S_1)\geq h_k^t(S_2)$ and $p_{ik}(S_1)\geq
  p_{ik}(S_2)$ for $k\neq n$ both hold.
  The first holds due to the induction assumption, and the second holds because
  that the transition probability from a transit state $i$ to transit state $k$ is
  impossible to increase when more nodes in $S_2\backslash S_1$ are connected to
  the absorbing state $n$, i.e., $p_{ik}(S_1)\geq p_{ik}(S_2)$ for $k\neq n$.

  By induction, we conclude that $h_i^T(S)$ is non-increasing.

  \header{Supermodularity.}
  We use induction to show that $h_i^T(S)$ is a supermodular set function.
  Let $S_1'\triangleq S_1\cup\{s\}$ and $S_2'\triangleq S_2\cup\{s\}$, where $s\in
  V\backslash S_2$.
  Let $\delta_i^t(s;S)\triangleq h_i^t(S\cup\{s\})-h_i^t(S)\leq 0$ denote the
  marginal gain.
  (Notice that $\delta_n^t(s;S)\equiv 0, \forall S,\forall t$.)
  For $T=0$, $\delta_i^0(s;S_1) = \delta_i^0(s;S_2) = 0$.
  Assume the conclusion holds for $T=t$, i.e., $\delta_i^t(s;S_1)\leq
  \delta_i^t(s;S_2)$.
  To show that the conclusion holds for $T=t+1$, we need to consider three cases:

  \noindent\textbullet\, $i\in V\backslash S_2'\cup S_1$.
  In this case, probability transitions $\{p_{ik}\}_{k\in V}$ are constants, i.e.,
  $p_{ik}(S')=p_{ik}(S)=p_{ik}(S_2)=p_{ik}(T')\triangleq p_{ik}$, for $k\neq n$.
  So,
  \begin{align*}
    \delta_i^{t+1}(s;S_1)
    &= \sum_k p_{ik}\bigl[ h_k^t(S_1')-h_k^t(S_1) \bigr] \\
    &= \sum_{k\neq n} p_{ik}\delta_k^t(s;S_1) \\
    &\leq \sum_{k\neq n} p_{ik}\delta_k^t(s;S_2) \\
    &= \delta_i^{t+1}(s;S_2).
  \end{align*}

  \noindent\textbullet\, $i\in S_2\backslash S_1$.
  In this case, probability transitions satisfy relation $p_{ik}(S_1') =
  p_{ik}(S_1)\geq p_{ik}(S_2) = p_{ik}(S_2')$.
  So,
  \begin{align*}
    \delta_i^{t+1}(s;S_1)-\delta_i^{t+1}(s;S_2)
    &= \sum_k\bigl[
      p_{ik}(S_1)\delta_k^t(s;S_1) - p_{ik}(S_2)\delta_k^t(s;S_2) \bigr] \\
    &\leq \sum_kp_{ik}(S_2)\bigl[ \delta_k^t(s;S_1) - \delta_k^t(s;S_2) \bigr] \\
    &\leq 0.
  \end{align*}
  (Note that $\delta_k^t(s;S_1)\leq 0$ due to monotonicity.)

  \noindent\textbullet\, $i=s$.
  In this case, probability transitions have relation $p_{ik}(T') = p_{ik}(S')\leq
  p_{ik}(S) = p_{ik}(S_2)$, for $k\neq n$.
  So,
  \begin{align*}
    \delta_i^{t+1}(s;S_1) -\delta_i^{t+1}(s;S_2)
    &=\sum_k\biggl\{p_{ik}(S_1') \bigl[h_k^t(S_1') - h_k^t(S_2')\bigr]\\
    &\hspace{40pt}
      -p_{ik}(S_1) \bigl[ h_k^t(S_1) - h_k^t(S_2) \bigr]\biggr\} \\
    &\leq\sum_{k\neq n} p_{ik}(S)\bigl[\delta_k^t(s;S_1) - \delta_k^t(s;S_2)\bigr] \\
    &\leq 0.
  \end{align*}
  The three cases above have covered each $i\in V$.
  By induction, we conclude that $h_i^T(S)$ is a supermodular set function, and
  this completes the proof of Theorem~\ref{thm:FHT}.
\end{proof}

\section*{Proof of Theorem~\ref{thm:bound1}}
\label{app:proof-bound1}

\begin{proof}
  Define random variable $X_{ir}\triangleq b_{ir}/(nR)\in [0,(nR)^{-1}]$, and note
  that $\hFAP = 1/n\sum_{i\in V}\sum_{r=1}^Rb_{ir}/R = \sum_{i,r}b_{ir}/(nR) =
  \sum_{i,r}X_{ir}$.
  The Hoeffding inequality yields $P(|\hFAP-\FAP|\geq\delta) \leq
  2\exp(-2nR\delta^2)$.
  Letting the probability be less than $\epsilon$, we obtain $R\geq
  \frac{1}{2n\delta^2}\ln(\frac{2}{\epsilon})$.

  Similarly, to show the bound of $R$ in estimating D-HT, we can define another
  random variable $Y_{ir}\triangleq t_{ir}/(nR)\in [0,T/(nR)]$.
  Applying the Hoeffding inequality again yields $R\geq\frac{1}{2n\delta^2}
  \ln(\frac{2}{\epsilon})$.
\end{proof}

\section*{Proof of Theorem~\ref{thm:bound2}}
\label{app:proof-bound2}

\begin{proof}
  Given $S\subseteq V$, for a node $s\in V\backslash S$, and $S'\triangleq
  S\cup\{s\}$, we have
  \begin{align*}
     & P(|\hdAP(s;S)-\dAP(s;S)|\geq \delta/c_s) \\
    =\,& P(|[\hFAP(S')-\FAP(S')] - [\hFAP(S)-\FAP(S)]|\geq\delta) \\
    \leq\,& P(|\hFAP(S')-\FAP(S')| + |\hFAP(S)-\FAP(S)|\geq\delta) \\
    \leq\,& P(|\hFAP(S')-\FAP(S')|\geq\delta/2)+P(|\hFAP(S)-\FAP(S)|\geq\delta/2).
  \end{align*}
  Now we directly apply the conclusion in the proof of Theorem~\ref{thm:bound1}.
  The first probability of the right hand side satisfies
  \begin{align*}
    P(|\hFAP(S')-\FAP(S')|\geq\delta/2) \leq 2\exp(-nR\delta^2/2).
  \end{align*}
  The second probability of the right hand side satisfies
  \[
    P(|\hFAP(S)-\FAP(S)|\geq\delta/2)\leq 2\exp(-nR\delta^2/2).
  \]
  Together, we have
  \[
    P(|\hdAP(s;S)-\dAP(s;S)|\geq \delta/c_s)\leq 4\exp(-nR\delta^2/2).
  \]
  Applying the union bound, we obtain
  \begin{align*}
    P(\exists s\in V\backslash S,|\hdAP(s;S)-\dAP(s;S)|\geq \delta/c_s)
    &\leq 4(n-|S|)\exp(-nR\delta^2/2) \\
    &\leq 4n\exp(-nR\delta^2/2).
  \end{align*}
  Letting the upper bound be less than $\epsilon$, we get
  $R\geq \frac{2}{n\delta^2} \ln\frac{4n}{\epsilon}$.

  By exactly parallel reasoning, we can obtain that when
  $R\geq \frac{2}{n\delta^2} \ln\frac{4n}{\epsilon}$,
  then $P(\exists s\in V\backslash S,
  |\hdHT(s;S)-\dHT(s;S)| \geq \delta T/c_s) \leq \epsilon$.
\end{proof}

\end{document}